\documentclass[11pt]{article}
\usepackage[letterpaper,margin=0.92in]{geometry}
\usepackage[T1]{fontenc}
\usepackage{lmodern,microtype}
\usepackage{amsmath,amssymb,amsthm,mathtools,bm}
\usepackage{thmtools,thm-restate}
\usepackage{booktabs,tabularx,longtable,enumitem}
\usepackage[numbers,sort&compress]{natbib}
\usepackage{xcolor}
\usepackage[colorlinks=true,linkcolor=blue!48!black,citecolor=blue!48!black,
urlcolor=blue!48!black]{hyperref}
\usepackage[nameinlink,noabbrev]{cleveref}
\setlist[itemize]{leftmargin=1.6em,itemsep=2pt}
\setlist[enumerate]{leftmargin=1.7em,itemsep=2pt}
\newcommand{\bool}{\set{0,1}}

\newcommand{\Ac}{\mathcal{A}}

\newcommand{\Dc}{\mathcal{D}}

\newcommand{\Fc}{\mathcal{F}}

\newcommand{\Rc}{\mathcal{R}}

\newcommand{\Xc}{\mathcal{X}}
\newcommand{\Yc}{\mathcal{Y}}
\newcommand{\Zc}{\mathcal{Z}}

\newcommand{\wt}[1]{\widetilde{#1}}
\newcommand{\wh}[1]{\widehat{#1}}
\newcommand{\eps}{\varepsilon}
\newcommand{\Tr}{\text{Tr}}

\newcommand{\Uni}[1]{\mathcal{U}\left( #1 \right)}

\newcommand{\prob}[2][]{\operatorname*{\mathbb{P}}_{#1 }\brac*{#2}}
\newcommand{\E}[2][]{\operatorname*{\mathbb{E}}_{#1 }\brac*{#2}}
\newcommand{\Var}[2][]{\operatorname*{\normalfont{\text{Var}}}_{#1 }\paren*{#2}}

\newcommand{\bO}[1]{\operatorname*{O}\paren*{#1}}
\newcommand{\bOt}[1]{\operatorname*{\wt{O}}\paren*{#1}}
\newcommand{\lO}[1]{\operatorname*{o}\paren*{#1}}
\newcommand{\bOm}[2][]{\operatorname*{\Omega}_{#1}\paren*{#2}}

\newcommand{\bT}[2][]{\operatorname*{\Theta}_{#1}\paren*{#2}}

\DeclarePairedDelimiter{\floor}{\lfloor}{\rfloor}
\DeclarePairedDelimiter{\ceil}{\lceil}{\rceil}
\DeclarePairedDelimiter{\norm}{\lVert}{\rVert}
\DeclarePairedDelimiter{\abs}{|}{|}
\DeclarePairedDelimiter{\brac}{\lbrack}{\rbrack}
\DeclarePairedDelimiter{\set}{\lbrace}{\rbrace}
\DeclarePairedDelimiter{\paren}{\lparen}{\rparen}

\DeclarePairedDelimiter{\bra}{\langle}{|}
\DeclarePairedDelimiter{\ket}{|}{\rangle}
\DeclarePairedDelimiterX{\braket}[2]{\langle}{\rangle}{#1\delimsize\vert#2}

\newtheorem{theorem}{Theorem}[section]
\newtheorem{informaltheorem}{Theorem}

\crefname{informaltheorem}{Theorem}{Theorems}
\Crefname{informaltheorem}{Theorem}{Theorems}
\newtheorem{lemma}[theorem]{Lemma}
\newtheorem{corollary}[theorem]{Corollary}

\newtheorem{proposition}[theorem]{Proposition}
\newtheorem{observation}[theorem]{Observation}
\theoremstyle{definition}
\newtheorem{definition}[theorem]{Definition}

\theoremstyle{remark}
\newtheorem{remark}[theorem]{Remark}

\newcommand{\ARHM}{\mathsf{ARHM}}
\newcommand{\HM}{\mathsf{HM}}
\newcommand{\MHM}{\mathsf{MHM}}
\newcommand{\SeqHM}{\mathsf{SeqHM}}
\newcommand{\SeqMHM}{\mathsf{SeqMHM}}
\newcommand{\BHX}{\mathsf{BHX}}
\newcommand{\DHX}{\mathsf{DHX}}
\newcommand{\NGC}{\mathsf{NGC}}

\newcommand{\trnorm}[1]{\norm*{#1}_{\mathrm{tr}}}
\newcommand{\fs}{\operatorname{fs}}
\newcommand{\tfor}{\tau_{\mathrm{for}}}
\newcommand{\Squ}[1]{S^{\mathrm {Q}}_{#1}}
\newcommand{\Scl}[1]{S^{\mathrm {C}}_{#1}}

\newcommand{\polylog}{\operatorname{polylog}}

\newcommand{\Prb}{\mathbb{P}}
\newcommand{\wtO}{\widetilde {O}}
\newcommand{\ind}{\mathbf{1}}
\newcommand{\calT}{\mathcal{T}}
\newcommand{\calB}{\mathcal{B}}
\newcommand{\calG}{\mathcal{G}}

\title{Random Order in Quantum Streaming: Replenishment and Robust Lower Bounds}
\date{}
\author{Nadezhda Voronova\footnote{IRIF (CNRS \& Universit\'e Paris Cit\'e),
  \texttt{voronova@irif.fr}}}
\begin{document}
\maketitle

\begin{abstract}
\noindent
How can random order of the classical input data change the role of quantum memory in the streaming setting? In this work, we show that there is no single answer. The effect can depend on the nature of the data, the way information is distributed among items in the stream, and how later computation uses information from the earlier part of the input. In some cases, later stream items may contain enough information to build a new useful quantum state after an earlier one is no longer available; we call this \emph{replenishment}. In others, the relevant information requirement remains essentially unchanged even in the presence of quantum memory under random order.

To demonstrate the first behavior, we construct an artificial problem based on Hidden Matching. The stream contains multiple copies of the input coordinate data $x \in \{0, 1\}^N$ together with $r \leq \sqrt{N}$ independent online matching requests, and is presented in uniformly random order. We show the existence of a quantum streaming algorithm that uses $O(\log{N} \log{r+1})$ space by constructing approximate Hidden Matching phase states from arriving coordinate items, consuming them to answer requests, and replenishing them using later coordinate items. The state-preparation mechanism is closely related to \emph{quantum oracle sketching}, introduced by Zhao et al., which builds useful coherent quantum states from randomly arriving classical data.

We also show that, for polynomially growing $r$, the same problem requires exponentially more space both in the classical random-order streaming setting and in the quantum adversarial-order streaming setting. First, we confirm that the unconditional lower bound $\Omega(\sqrt{N})$ on Hidden Matching extends to our setting. Second, we show that even if all coordinate data instead appears before the requests, every quantum algorithm requires $\Omega(r)$ space unconditionally. This connects replenishment to the \emph{consumability} of quantum data studied by Gilboa, Jain, and McClean: a compact quantum encoding may answer one request but fail to remain useful for many requests about the same underlying data. For this, we strengthen their communication lower bounds for the Multiple Hidden Matching problem. For every fixed $\delta \in (0,1/2)$ and $r \leq N^{1/2-\delta}$, any single quantum encoding supporting $r$ independent matching requests requires $\Omega(r)$ qubits, even when all requests are revealed simultaneously and arbitrary joint decoding is allowed. For sequentially revealed requests, the same linear lower bound extends through $r=\Theta(\sqrt{N})$. Thus, in our streaming problem, later coordinate items do more than change when the requests are revealed: they replenish quantum states that would otherwise have to be replaced by a larger post-data representation. We also apply replenishment to triangle counting, adapting Kallaugher's quantum streaming algorithm to obtain a smaller expected-space bound when edges are repeated equally often in uniformly random order.

For the second behavior, we extend the robust Noisy Gap Cycle framework of Assadi and Sundaresan to quantum streaming. Their reduction is largely unchanged; the additional quantum ingredient is a lower bound for Block Hidden XOR, a problem implicit in the original paper, obtained by replacing the classical Fourier-analytic step with its standard matrix-valued analog. This gives an \(\Omega(n)\) one-pass quantum space lower bound for Noisy Gap Cycle in uniformly random edge order. Their reductions then give quantum random-order lower bounds for estimating connected components and maximum matching, planar maximum independent set, approximating minimum spanning tree weight, and generating random walks. These recover the same asymptotic lower bounds as those for classical random-order streaming algorithms, up to constant factors.

Together, these results give examples of different ways random order can interact with quantum memory: in our constructed problem it enables replenishment of a small consumable representation, while the Assadi-Sundaresan random-order lower-bound framework remains robust in the quantum setting.

\end{abstract}

\section{Introduction}
\label{sec:intro}

Quantum streaming algorithms process a classical input stream using a
small amount of quantum memory.  For some problems, this memory can be
exponentially smaller than what is possible classically
\cite{LeGall2006,KallaugherParekhVoronova2024,WangYang2026,WuYang2026}.
Following Le Gall's separation in online space
complexity~\cite{LeGall2006}, quantum space advantages have been shown
for triangle counting~\cite{Kallaugher2021}, Maximum Directed
Cut~\cite{KallaugherParekhVoronova2024}, and Max-$k$SAT
\cite{WangYang2026}.  Recent work by Feng et al.~\cite{FengEtAl2026}
studies entropy estimation, Zhao et al.~\cite{ZhaoEtAl2026} develop
quantum oracle sketching for processing classical data, and Wu and
Yang~\cite{WuYang2026} study decoded quantum interferometry in streams.
This activity also extends to experiments: Niroula et al.
\cite{NiroulaEtAl2025} implement a quantum streaming algorithm on
trapped-ion qubits.

A quantum sketch can retain useful information in little space, but
querying it may change or destroy it.  Kallaugher, Parekh, and
Voronova~\cite{KallaugherParekhVoronovaPair} describe a common pair-sampling
sketch underlying several quantum streaming algorithms, whose queries can remove
elements or consume the sketch.  This raises a natural question: how
much memory is needed when the algorithm must repeatedly answer
questions about the same input?

Consider the Hidden Matching problem of Bar-Yossef, Jayram, and
Kerenidis~\cite{BarYossefJayramKerenidis2008}.  Alice has a Boolean string
$x\in\bool^N$, and Bob has a perfect matching on $[N]$.  Bob's goal
is to output an edge of his matching together with the parity of its
two endpoints in $x$.  Their quantum protocol lets Alice send a phase
encoding of $x$ using $\bO{\log N}$ qubits, which Bob measures to obtain
a correct answer; they also prove a classical one-way lower bound of
$\bOm{\sqrt N}$ bits.  But this measurement consumes the phase state:
to use the same protocol for another matching, Bob would need another
copy.  Gilboa, Jain, and McClean~\cite{GilboaJainMcClean2025} formalize
this issue through \emph{consumable quantum data}, asking how the
communication needed for several requests grows with their number.
They prove an $\bOm{\sqrt m}$ quantum communication lower bound for
$m$ Hidden Matching requests when $m\le N/2$.  We prove stronger lower
bounds for this problem, showing that the cost persists even if Alice
uses an arbitrary quantum encoding.

In a stream, there is another possibility.  If data describing $x$
continues to arrive, an algorithm may use it to prepare new states
after earlier ones have been measured.  We call this \emph{replenishment}.
An occurrence of a coordinate already seen need not tell us anything
new about $x$, but it can still help build a new quantum state.
A related way to use classical input is quantum oracle sketching,
introduced by Zhao et al.~\cite{ZhaoEtAl2026}, which constructs coherent
quantum operations from random classical samples.  Our Hidden Matching
algorithm uses their phase-accumulation method repeatedly to prepare
replacement states as requests arrive.

The order of these occurrences matters.  If all data arrives before
the requests, the algorithm must store enough information to answer
every request from its memory alone.  If updates arrive between
requests, it may be possible to use the same small memory repeatedly.
Stream order also affects classical space complexity: Guha and
McGregor~\cite{GuhaMcGregor2009}, for example, study the advantages of
random order for quantile estimation.  Here we ask how order affects
the ability to rebuild a quantum state between uses.

Our main result shows that these two situations can have exponentially
different quantum space complexity.  We introduce Atomized Refreshable
Hidden Matching (ARHM), whose stream consists of repeated coordinates
of a fixed string and matching requests that must be answered immediately.
We give a small-space quantum algorithm when the stream is in uniformly
random order, and prove a much larger quantum lower bound when all
updates precede the requests.  We also prove a classical lower bound
under the same random-order promise.  Thus ARHM gives both an exponential
quantum advantage in random-order streaming and an exponential advantage
from interleaving updates with requests within the quantum model.
The random-order algorithm works for every fixed string and request
list, with all answers required to be correct together.

We then apply replenishment to triangle counting.  Kallaugher's
algorithm~\cite{Kallaugher2021} combines a quantum sketch with a
classical sampling estimator.  With repeated edges in random order,
we can rebuild the sketch and resample the classical component, using
the same memory to obtain observations over successive blocks of the
stream.  This improves the space bound under a stronger input promise.
Here we give an algorithmic use of replenishment; the arbitrary-encoding
consumability lower bounds concern Hidden Matching.

Finally, we prove a quantum lower bound for Noisy Gap Cycle (NGC) in
random-order streams where every edge appears once.  Robust communication
complexity, studied by Chakrabarti, Cormode, and
McGregor~\cite{ChakrabartiCormodeMcGregor2016}, provides a way to prove
streaming lower bounds that survive random input order.  Assadi and
Sundaresan~\cite{AssadiSundaresan2023} use this approach to prove
classical lower bounds for NGC and several other graph problems.
We extend their NGC argument to quantum messages, obtaining linear
space lower bounds for NGC and consequences for several graph problems.
This gives a complementary result on the limitations of quantum memory
in random order.  The effect of repeated edges on NGC remains open.

\subsection{Replenishment and a three-way space separation}
\label{sec:intro-arhm}

Can the timing of access to classical data change how much quantum
memory is needed?  To isolate this question, we introduce Atomized
Refreshable Hidden Matching (ARHM).  This is an artificial problem,
designed specifically to separate the three models below.  Its purpose
is to demonstrate an effect of replenishment for which both the upper
and lower bounds can be proved.

\begin{informaltheorem}[Three-way space separation; informal version of \cref{thm:ext-param-threeway}]
\label{thm:intro-threeway-informal}
There is a one-pass streaming problem with online outputs that can be
solved using polylogarithmic quantum space in uniformly random order,
but unconditionally requires polynomial space both for classical
random-order algorithms and for quantum algorithms in adversarial order.
All three models receive the same multiset of input items and must meet
the same output deadlines, with joint success probability at least $2/3$.
\end{informaltheorem}

Here space is measured as a function of the stream length.  To the best
of our knowledge, this is the first unconditional exponential space
separation between one-pass quantum random-order and quantum
adversarial-order streaming, allowing online outputs.
The adversarial-order lower bound already holds when all updates
precede all requests.

ARHM presents a fixed Boolean string through repeated coordinate
updates, interspersed with Hidden Matching requests.  Each request asks
for an edge of a specified matching and the parity of its endpoints,
and must be answered before the next item arrives.  Random order lets
the algorithm rebuild small quantum states between requests.  When all
updates arrive first, the algorithm must retain enough information to
answer the requests after access to the string has ended.

The two lower bounds explain different aspects of this result.
The classical random-order algorithm receives the same repeated data
as the quantum algorithm, so repetitions alone do not account for the
small space bound.  The quantum adversarial-order algorithm also receives
all the repetitions, so the amount of data supplied does not account
for the difference between the two quantum models.  What changes is
whether that data remains available while the algorithm answers requests.
The construction and its parameters are given in \cref{sec:replenishment}.

\paragraph{Lower bounds for repeated requests.}
A measurement destroying a particular sketch does not by itself show
that the data is consumable: a different encoding or decoding procedure
might permit repeated use.  Gilboa, Jain, and
McClean~\cite{GilboaJainMcClean2025} address this distinction through
one-way communication.  Alice holds one input, while Bob wants to answer
several questions about it.  They ask whether the message must grow with
the number of questions, even though Alice's underlying data stays the
same.  Their motivation includes the possibility of distributing useful
data while limiting its reuse, in contrast to classical data that can
be copied.

For streaming, the same question describes what happens once access to
the source ends.  The retained memory is Alice's message, and future
requests must be answered from that message alone.  Gilboa, Jain, and
McClean prove a square-root lower bound in the number of Hidden Matching
requests and conjecture linear growth for a suitable range.  We establish
linear growth in the following regimes.

\begin{informaltheorem}[Consumability; informal versions of \cref{thm:mhm-linear-main,cor:app-seqmhm-joint}]
\label{thm:intro-consumability-informal}
Without prior entanglement, answering multiple Hidden Matching requests
from one quantum message requires communication linear in the number of
requests.  For simultaneous requests, this holds when their number is
polynomially below the square root of the string length, even with
arbitrary joint decoding.  For sequential requests, it holds through
the square-root scale.  In both cases, all answers must be correct
with probability at least $2/3$.
\end{informaltheorem}

The simultaneous result allows Bob to see all his requests before
decoding; the sequential result requires him to answer them as they
arrive.  Both allow Alice to choose an arbitrary encoding.  Thus the
cost of repeated use is a property of the task, rather than of the
usual Hidden Matching protocol.  These bounds yield the quantum
update-first lower bound for ARHM.  They also apply when an interleaved
prefix is followed by a sufficiently long sequence of fresh independent
requests with no further updates (\cref{prop:last-replenishment-point}).
The consumability results are proved in \cref{sec:consumability}.

\subsection{Applying replenishment to triangle counting}
\label{sec:intro-triangles}

ARHM is designed to isolate replenishment.  We next ask whether the
same idea can improve an algorithm for an established graph problem.
Triangle counting provides such an example.  Kallaugher's
algorithm~\cite{Kallaugher2021}, the first polynomial quantum space
advantage for a natural streaming problem, combines a quantum sketch
with classical sampling.  Obtaining enough observations for an accurate
estimate requires many copies of the estimator to coexist.

Repeated input offers another way to obtain these observations.  We
consider a stream in which every edge of the underlying simple graph
appears equally often, with all occurrences uniformly permuted.  Later
occurrences let us rebuild the quantum sketch and refresh the classical
sample, obtaining further observations using the same memory.  The
quantity being estimated is still the number of triangles in the
underlying graph; repetitions change access to the graph, not the
answer.

\begin{informaltheorem}[Triangle replenishment; informal version of \cref{thm:main-parameterized,cor:optimized}]
\label{thm:intro-triangle-informal}
With sufficiently many equal repetitions in uniformly random order,
and supplied estimates of the graph parameters, triangle counting
admits a one-pass quantum algorithm with constant relative error.
In the regimes specified in \cref{cor:optimized}, its expected space
bound improves Kallaugher's single-arrival bound by a factor proportional
to the number of repetitions, up to logarithmic factors and an additive
polylogarithmic term.
\end{informaltheorem}

For example, consider graphs with a triangle count proportional to the
number of edges and with each edge belonging to at most constantly many
triangles.  In this regime, Kallaugher's bound is polynomial in the
number of edges, whereas our bound becomes polylogarithmic with a
suitable sublinear polynomial number of repetitions per edge.  A linear
number of repetitions per edge is more than sufficient.  Thus repeated
access can give an exponential improvement over the space bound of
this known algorithm, while keeping the total stream length polynomial.

This is a comparison of algorithmic upper bounds under different input
promises.  We do not prove a triangle-counting lower bound, or a
quantum--classical separation for the repeated-edge model.  The result
shows how replenishment improves an existing quantum algorithm for a
natural problem; the arbitrary-encoding consumability guarantees are
established for Hidden Matching.  We use expected space complexity,
with a measure comparable to Kallaugher's~\cite{Kallaugher2021}.
The precise convention, supplied-parameter assumptions, and applicable
regimes are given in \cref{sec:triangles}.

\subsection{Quantum lower bounds for Noisy Gap Cycle}
\label{sec:intro-ngc}

Random order and quantum memory both offer ways to reduce space, but
together they need not make a graph problem easy.  To study this
limitation, we consider Noisy Gap Cycle (NGC), a problem used by Assadi
and Sundaresan~\cite{AssadiSundaresan2023} to obtain random-order graph
streaming lower bounds.

The task is to distinguish a graph containing many disjoint cycles
from one containing half as many cycles of twice the length.  Both
cases also contain disjoint paths.  Although these graphs have simple
components, distinguishing them requires information about how edges
connect over long distances.  Short portions of a cycle do not reveal
its length.  This makes the problem useful for studying the difficulty
of recovering global graph structure from a stream.

NGC is also a source of lower bounds for other graph problems.
Assadi and Sundaresan show how its classical hardness implies limitations
on estimating connected components and other graph quantities in random
order~\cite{AssadiSundaresan2023}.  We prove that the central obstruction
persists with quantum memory.

\begin{informaltheorem}[Quantum graph lower bound; informal version of \cref{thm:ngc-robust,thm:ngc-rand}]
\label{thm:intro-ngc-informal}
For admissible cycle lengths at least $\bOm{\log n}$ where $n$ is the number of vertices, Noisy Gap Cycle
requires linear one-pass quantum space in random edge order.
The corresponding one-way quantum communication problem under a uniformly
random edge partition also requires linear communication, without prior
entanglement.  Both bounds hold with error at most $1/3$.
\end{informaltheorem}

The resulting applications include lower bounds for approximating
connected components, matching, independent set, and minimum spanning
tree weight, and for generating random walks
(\cref{cor:ngc-applications}).  In particular, the approximation results
give restrictions on the dependence of space on accuracy; they do not
assert linear space for every fixed approximation error.  We state
the precise tradeoffs with the results in \cref{sec:ngc}.

These are single-arrival graph streams, with a decision made after the
pass.  Thus NGC supplies a complementary obstruction to the advantages
above: quantum memory can still require linear space despite random
order.  Whether repeated access can reduce this cost remains open.

\subsection{Organization}
\label{sec:intro-perspective}

\Cref{sec:overview} explains the mechanisms behind the results, followed
by related work and model conventions in \cref{sec:related,sec:prelim}.
\Cref{sec:replenishment} gives the ARHM separation, using the consumability
bounds proved in \cref{sec:consumability}.  \Cref{sec:triangles} develops
the repeated-edge estimator.  \Cref{sec:ngc} proves the quantum Hidden
XOR bound for NGC, and \cref{app:ngc-full-proof} supplies the graph
reductions and their consequences.  The remaining appendices contain
consumability refinements and triangle-estimator details.
We discuss open questions at the end of each technical section.

\subsection{Concurrent work}
\label{sec:concurrent-work}

Independently and concurrently, Bouland, Ding, and
Jain~\cite{BoulandDingJain2026} prove exponential quantum--classical
space separations in random-order streaming for Code Intersection and
low-rate Optimal Polynomial Intersection.  Their quantum algorithms work
in every order, while their classical lower bounds hold in random order.
Our ARHM result gives a concurrent exponential quantum--classical
separation under different input and output requirements.  We also compare
two orders of the same stream within the quantum model, showing how
continued access to the data reduces the memory needed for online requests.

\section{Technical overview}
\label{sec:overview}

We first describe the lower bound for repeated Hidden Matching requests,
and then explain how the ARHM algorithm uses later updates to prepare
new states as requests arrive.  We next describe the changes needed to
reuse the triangle estimator across blocks.  Finally, we give an overview
of the quantum communication lower bound used for NGC.
The statements repeated here retain their numbers from the technical
sections, where we give the full proofs.  All communication protocols
are one-way and have no prior entanglement; their messages include any
transmitted classical information.

\subsection{Lower bounds for repeated Hidden Matching requests}
\label{sec:overview-mhm}

Suppose Alice sends Bob a quantum encoding of $x$, and Bob must use it
to answer several matching requests.  To lower-bound the message length,
we cannot assume that Alice sends phase states or that Bob measures
them separately.  Instead, we ask how much information about $x$ is
contained in any correct list of answers.

For even $N$, write $\HM_N^{(m)}$ for $m$ Hidden Matching requests
sharing Alice's string.  The precise task is as follows.

\begin{restatable*}[Multiple Hidden Matching]{definition}{multipleHiddenMatchingDefinition}
\label{def:mhm}
In $\MHM_{N,m}=\HM_N^{(m)}$, Alice receives $x\in\bool^N$ and Bob
receives perfect matchings $M_1,\ldots,M_m$ on $[N]$.  For each $i$ he
must output $e_i=\set{u_i,v_i}\in M_i$ and
$b_i=x_{u_i}\oplus x_{v_i}$.  Success requires every answer to be correct.
\end{restatable*}

We call a choice of one edge $e_i\in M_i$ from each requested matching
a \emph{transversal}.  Thus, in our consumability and ARHM arguments,
a transversal is simply the list of edges selected to answer the requests,
ignoring their reported parity bits.  The edges need not be vertex-disjoint
or distinct.  When discussing its forest size, we consider the graph formed
by the selected edges, with repetitions counted only once.

Our simultaneous lower bound is:

\begin{restatable*}[Linear Multiple Hidden Matching lower bound]{theorem}{linearMHMTheorem}
\label{thm:mhm-linear-main}
Fix $\delta\in(0,1/2)$ and $\eps\in(0,1/2)$.  For
$m\le N^{1/2-\delta}$, every one-way protocol for
$\MHM_{N,m}$ with worst-case joint error at most $\eps$ uses
$\bOm[{\delta,\eps}]{m}$ qubits.
\end{restatable*}

We use the forest-information approach of Gilboa, Jain, and
McClean~\cite{GilboaJainMcClean2025}.  Their argument relates the information in Bob's answers to the size
of a forest in his output edges.  We show that this forest must be larger
for the request distributions considered here, and account for errors
without conditioning the requests on successful decoding.  Correct edge parities on a forest with $s$ edges impose
$s$ independent constraints on a uniform string $X$.  Once the output transcript is fixed,
at most $2^{N-s}$ strings satisfy those constraints.  This remains true
when the edges were chosen using information about $X$: we count
compatible inputs \emph{after} fixing the transcript.  Cycles add no
independent parity constraint, so the useful quantity is the spanning
forest size $K$ of the selected edges.

To state the information bound, let $X$ be a uniform $N$-bit string,
and let the matching list $M$ be independent of $X$ and the encoding
randomness.  Alice sends a $q$-qubit message without knowing $M$;
independent public coins are allowed.  Let $O$ be Bob's output list of
legal edges and parity bits, let $S$ indicate that all answers are
correct, and let $K(M,O)$ be the number of edges in a spanning forest
of the selected edges.  The common information bound is:

\begin{restatable*}[Forest-information bound]{lemma}{forestInformationLemma}
\label{lem:forest-information-main}
\label{lem:app-conditional-forest-information}
If $p=\prob{S=1}$, then
\[
  q\ge I(X:O\mid M)\ge\E{S K(M,O)}-h_2(p).
\]
In particular, if the matchings are fixed and every transversal
$e_i\in M_i$ has forest size at least $s$, then
$q\ge ps-h_2(p)\ge ps-1$.
\end{restatable*}

The first inequality follows from data processing and the dimension
of the message.  Keeping the product $SK$ is important: a decoder's success
can be correlated with the forest it selects.

\paragraph{Simultaneous requests.}
If all matchings are known before decoding, we must force a large forest
for \emph{every} coordinated edge choice.  A graph with forest size at
most $s$ has at most $2s$ incident vertices.  Thus a small-forest answer
requires a small vertex set containing an edge of every requested
matching.  For independent balanced matchings, each pair occurs with
probability $1/(N-1)$, giving
\[
 \prob{\tfor(M_1,\ldots,M_m)\le s}
 \le\binom N{2s}
       \left(\frac{\binom{2s}{2}}{N-1}\right)^m.
\]
Here we take a union bound over sets of $2s$ vertices.  For a fixed
set, the expression raised to the $m$th power bounds the probability
that every matching has an edge inside it.  When $m\le N^{1/2-\delta}$,
we can take $s$ to be a sufficiently small multiple of $\delta m$
and still make this probability small after the union bound.  With high probability,
all transversals have linear forest size.  Applying the information bound and accounting for the exceptional
matchings gives a linear communication lower bound, even when Bob
decodes all requests jointly (\cref{thm:ext-balanced-quant}).

Balancedness holds both for uniform perfect matchings and for uniform
factors of any one-factorization.  Independent factor draws may repeat;
the assertion concerns the resulting distribution, not every fixed
request list.  Indeed, a list repeating one matching can be answered by
repeating one correct edge and parity.  As $m$ approaches $\sqrt N$, this counting argument gives a weaker
bound.  We give the resulting estimates in
\cref{app:transversal-forest-optimization,sec:ext-endpoint}.

\paragraph{Sequential requests.}
When Bob must answer before seeing the next matching, we can analyze the
forest as it grows.  After $i-1$ output edges, adding an edge fails to
increase forest size only if its endpoints already belong to the same
component.  There are at most $\binom i2$ such pairs.  The next matching
is fresh even conditional on the preceding history, so it contains any
of these pairs with probability at most $\binom i2/(N-1)$.
This bounds the chance of failed growth regardless of Bob's residual
quantum state or his choice of measurement.

Summing these losses gives expected forest size at least
$m-m(m^2-1)/(6(N-1))$.  To account for wrong answers, let $Z$ record
which output parity bits must be flipped.  Giving the decoder this error
vector repairs the forest constraints; its conditional entropy is the
information we must subtract.  Consequently,
\[
 q\ge m-\frac{m(m^2-1)}{6(N-1)}-H(Z\mid O,M)
 \ge m-\frac{m(m^2-1)}{6(N-1)}-H(Z).
\]
This is \cref{thm:sequential-hidden-matching}.  Joint success at least
$p$ implies $H(Z)\le1+(1-p)m$.  For $p=2/3$ and $m\le\sqrt N$,
the resulting lower bound is $q\ge m/2-1$.  Other bounds on error
entropy yield the average-advantage and constant-fraction guarantees in
\cref{app:consumability-extra}.

For update-first ARHM, Alice sends the complete memory after processing
the data and Bob continues on fresh factor requests.  This transfers the
sequential lower bound directly.  The same reasoning works after an
arbitrary simulatable prefix under the independence and state-accounting
conditions of \cref{prop:last-replenishment-point}.  Earlier successful
answers are not conditioned upon, and any earlier output needed for
continuation must be retained in the counted state.  Simultaneous
consumability similarly lower-bounds a post-data representation even
when the decoder is subsequently given the entire request list.

\subsection{ARHM: preparing accurate states before they are needed}
\label{sec:overview-arhm}

To prepare a Hidden Matching state from coordinate updates, we use
the phase-accumulation method of quantum oracle
sketching~\cite{ZhaoEtAl2026}.  We need to show both that the resulting
state is accurate enough and that we can prepare enough states before
the requests arrive.
Consider a build interval containing $H=\mu N$ coordinate updates.
Starting with $N^{-1/2}\sum_j\ket j$, each occurrence of $(j,x_j)$
applies phase $\pi x_j/\mu$ to coordinate $j$.  If $C_j$ occurrences
of coordinate $j$ appear, the completed register is
\[
 \ket{\phi_C(x)}=\frac1{\sqrt N}\sum_{j=1}^N
                    e^{i\pi C_jx_j/\mu}\ket j.
\]
Equal counts $C_j=\mu$ would give the exact Hidden Matching phase state.
We do not keep a table of the counts or require every count to be close
to its mean.  The matching measurement depends on phase differences
between paired coordinates, so its error is bounded by the average
squared count deviation.  Hypergeometric sampling gives
$\E{(C_j-\mu)^2}\le\mu$, hence expected decoding error at most
$\pi^2/\mu$ for any fixed string and matching
(\cref{lem:atomized-rebuilding}).  Taking $\mu=Cr$ leaves a constant
error budget for all $r$ requests together.

Accuracy alone does not ensure that a register is available when a
request arrives.  We use $B=Ar$ build intervals and maintain two banks:
one receives updates, and the other supplies completed registers.
There are therefore $U=BH=\bT{Nr^2}$ updates.  A random stream can be
generated by independently choosing the update order, the request order,
and their interleaving.  This separates the approximation analysis from
the schedule of consumption.

The interleaving spreads requests among the gaps between updates.  For
any bank, the expected demand is bounded by a constant, and the
beta-binomial tail in \cref{lem:queries-one-rebuild} decays sufficiently
fast that $\bO{\log(r+1)}$ registers per bank suffice over the whole
execution.  The initial interval has no completed state; its chance of
containing a request is at most $1/A$.  On a good schedule, a union bound
over decoding errors completes the argument.  Registers built in the
same interval share their count errors, but the proof needs only their
individual error bounds, not independence of those errors.

We can therefore keep only $\bO{\log(r+1)}$ registers in each bank,
reusing them as successive batches arrive.  Since each register uses
$\bO{\log N}$ qubits, the total space is
$\bO{\log N\log(r+1)}$, the upper bound in
\cref{thm:ext-param-threeway}.

For the classical lower bound, it is enough to consider the first request.  Take
all requests to name the same unknown factor.  Alice can simulate the
update prefix knowing only $x$, then send its memory to Bob to answer
that first request.  A level-two Fourier argument proves the classical
Hidden Matching lower bound for this restricted matching family:
averaging over uniform factors covers each possible edge with probability
$1/(N-1)$ (\cref{lem:classical-factorized-hm}).  Thus success on the
first answer alone forces $\bOm{\sqrt N}$ bits.

\subsection{Reusing the triangle estimator}
\label{sec:overview-pair-sketch}

Our triangle-counting algorithm adapts Kallaugher's
algorithm~\cite{Kallaugher2021} to repeated-edge streams.  We reuse its
quantum measurement primitive, the division into quantum low-weight and
classical high-weight contributions, and the classical sampling approach.
Our changes let us rebuild and reuse these estimators across blocks,
handle repeated occurrences, and combine the block estimates into an
estimate for the underlying simple graph.

For triangle counting, the queries are determined by the arriving
edges, and we only need to output an estimate at the end of the stream.
We use repeated edges to run the same estimator over several blocks
and average its observations.  A closing edge $uv$ probes the quantum representation
for pairs of directed edges $\overrightarrow{wu},\overrightarrow{wv}$.
A signed measurement outcome detects the contribution of intact pairs;
a zero outcome removes states pointing into the queried endpoints.
This is the destructive measurement underlying Kallaugher's hybrid
estimator~\cite{Kallaugher2021}.

Partition an exact-$R$ random-repetition stream into blocks of $Lm$
occurrences.  For logarithmic $L$, with high probability every full
block contains every edge between one and $2L$ times.  A block can
therefore prepare the graph representation used in the next block.
The construction must tolerate duplicate arrivals without a dictionary
of all previously seen edges.  Start with uniform amplitude on $2Lm$
dummy states, and swap fresh dummies into the two orientations of each
arriving edge.  Its first occurrence installs the desired amplitude;
later occurrences exchange equal amplitudes and leave it in place.
The completed state has a uniform directed-edge component of squared
mass $1/L$ (\cref{lem:build-state}).

Each lane keeps a READY register and a BUILD register.  A query block
uses the former while preparing the latter.  Activating each occurrence
with probability $1/(kL)$ controls how quickly queries remove useful
pairs.  Along the branch of zero outcomes, the probability of reaching
a query cancels the normalization of the surviving state.  This
telescoping identity makes the expected signed outcome an explicit sum
of triangle-witness survival probabilities
(\cref{lem:telescoping,lem:quantum-mean}).  Retaining the dummy component
also bounds the second moment by $1/L$, since a nonzero outcome can
occur only from the initial directed-edge mass.

A classical lane estimates the complementary weight by sampling centers
and directed edges.  Sampling decisions are keyed to the underlying
objects, so repeated arrivals reuse a decision and never create duplicate
stored entries.  Counters at the endpoints track earlier destructive
occurrences.  Their sum double-counts earlier copies of the closing
edge, whereas the quantum survival weight counts their union.  The
resulting discrepancy is at most $2T/k$ per block
(\cref{lem:high-bias}); this explains the condition $k\ge C/\eps$.

An individual block need not estimate $T$ without bias.  Its exact low
and high targets sum to
\[
 T_j^{<}+T_j^{>}=\frac1{3L}\sum_{e\in E}C_{j,e}\tau_e,
\]
where $C_{j,e}$ is the block multiplicity and $\tau_e$ the number of
triangles containing $e$.  Summing over blocks uses the exact identity
$\sum_j C_{j,e}+C_{\star,e}=R$ for each edge.  The omitted initial
build block and final partial block cost at most $8LT/R$, and the
computable classical complement adds at most $2T/k$.
Thus the average target is close to $T$ when $R\ge CL/\eps$ and
$k\ge C/\eps$ (\cref{lem:structural-bias}).

There are $J=\floor{R/L}-1$ query blocks.  Fixing the whole stream
makes their targets deterministic.  Fresh coins, hash seeds, and states
then give independent samples across blocks, even though the stream's
blocks themselves are dependent.  One physical lane supplies $J$
samples, storing only two buffers and a running sum.  The required
numbers of quantum and classical lanes are, respectively,
\[
 \bT{1+\frac{k^2L^2m^2}{R\eps^2T^2}}
 \quad\text{and}\quad
 \bT{1+\frac{Lm^{3/2}\Delta_E}{R\eps^2T\sqrt k}}.
\]
Balancing these quantities gives \cref{cor:optimized}; the full
parameterized guarantee is \cref{thm:main-parameterized}.
We give the precise space convention
in \cref{sec:triangles} and the implementation details in
\cref{app:triangle-implementation}.

\subsection{The quantum Noisy Gap Cycle lower bound}
\label{sec:overview-triangle-ngc}

We prove the NGC lower bound through robust communication complexity.
Each edge is assigned independently to Alice or Bob with probability
$1/2$.  Ordering each share uniformly
and concatenating Alice's order with Bob's gives an exactly uniform global permutation.  A
streaming algorithm therefore induces a one-way protocol by transferring
its complete live memory once.  Bob continues the computation and obtains
the final answer; no copying or rewinding of the state is needed
(\cref{prop:robust-transfer}).

The construction, neighboring-hybrid argument, and clean-coordinate
reduction come from Assadi and Sundaresan~\cite{AssadiSundaresan2023}.
Their construction hides the distinction between cycle lengths in path
parities across a layered graph.  A neighboring-hybrid
argument isolates one parity and reduces NGC to Distributional Hidden
XOR, losing advantage down to order $1/w$, where $w$ is the construction's
width.  The next reduction uses the random edge partition.  At a clean
coordinate, Alice owns the edges specifying a bit and Bob owns the
permutation edges selecting its position.  With high probability many
blocks contain enough clean coordinates and select one of them.

Fixing classical generation variables then leaves independent strings
for Alice and independent uniform indices for Bob.  To obtain a communication reduction, we show that the remaining inputs
have the required product distribution and that each party can generate
its edge share from its own input
(\cref{lem:residual-usable}).  Bob can generate the unused parts of the
residual permutations locally.  The fixing preserves sufficient average
advantage without conditioning on a successful measurement outcome.
The resulting task is unequal-block Hidden XOR.

\begin{restatable*}[Unequal-block Hidden XOR]{definition}{blockHiddenXORDefinition}
\label{def:ubh-main}
Fix a positive integer $r$ and positive integers $w_1,\ldots,w_r$.  Alice receives independent
uniform strings
\[
    X^i\in\bool^{w_i},
    \qquad i\in[r],
\]
and Bob receives independent uniform indices
$J_i\in[w_i], \qquad i\in[r]$,
independent of all of Alice's strings.  Communication is a
one-way message from Alice to Bob.
Bob must output
\[
    Z_J(X)
    :=
    \bigoplus_{i=1}^r X^i_{J_i}.
\]
When every $w_i=w$, we write $\BHX_{w,r}$.
\end{restatable*}

We need a lower bound even when Bob has only a small advantage over
random guessing:

\begin{restatable*}[Quantum Block Hidden XOR in the DHX regime]{theorem}{blockHiddenXORTheorem}
\label{thm:qbhx-main}
\label{cor:app-bhx-as-regime}
Suppose
\[
    r\ge\ln w,
    \qquad
    w_i\ge\frac{w}{100}
    \quad\text{for all }i\in[r].
\]
For all sufficiently large $w$, every one-way quantum protocol
for unequal-block Hidden XOR with success at least
$\frac{1}{2}+\frac{1}{8w}$
communicates $\bOm{wr}$ qubits.
\end{restatable*}

For classical communication, scalar Fourier bounds apply to the input
distribution conditioned on Alice's message.  For a quantum message we
analyze the full matrix-valued
encoding $X\mapsto\rho_X$.  If $S_J$ selects Bob's one coordinate per
block, the two parity-conditioned message states differ by
$2\wh\rho(S_J)$.  The optimal probability of distinguishing these
equiprobable states is $1/2+\trnorm{\wh\rho(S_J)}/2$, so success advantage
lower-bounds the average trace norm of these Fourier coefficients.
Because there are many choices of one coordinate per block, prediction requires Fourier
mass spread over many degree-$r$ sets.  The matrix hypercontractive
inequality of Ben-Aroya, Regev, and de Wolf~\cite{BenAroyaRegevDeWolf2008}
bounds that mass in terms of the message length:
\[
 \sum_S\lambda^{|S|}\trnorm{\wh\rho(S)}^2
 \le 2^{2\lambda q/(1+\lambda)},\qquad 0\le\lambda\le1.
\]

Comparing Fourier mass with the message dimension proves
\cref{thm:qbhx-main}.  The important feature is
that the advantage enters through a power of order $1/r$: with at least
logarithmically many blocks, inverse-width advantage costs only a constant
factor in the lower bound.  Since the reduction retains a constant
fraction of the original blocks, this yields $\bOm{n}$ communication.
The appendix handles bounded width and the cycle-length congruences,
then transfers the bound to random-order streaming.

To obtain the graph consequences in \cref{cor:ngc-applications}, we
check that each reduction preserves the required input order and uses
only one execution of the quantum algorithm.
A graph-to-matrix transformation that emits two coupled entries per edge
does not automatically preserve uniform order of individual entries,
so a matrix-rank consequence in that model needs a separate argument.

\section{Related work}
\label{sec:related}

\subsection{Quantum memory and stream order}

Le Gall~\cite{LeGall2006} demonstrated an exponential separation between
quantum and classical online space complexity.  Natural graph-streaming
applications subsequently included triangle counting
\cite{Kallaugher2021}, cut problems~\cite{KallaugherParekh2022}, and
Maximum Directed Cut~\cite{KallaugherParekhVoronova2024}.  These algorithms use small quantum sketches that may be changed by
measurement.  We study how later classical input can be used to prepare
new sketches when further computations are needed.

Quantum streaming has seen substantial recent activity, in both theory
and implementation.  Work from 2025--2026 includes the trapped-ion
implementation of Niroula et al.~\cite{NiroulaEtAl2025},
quantum oracle sketching by Zhao et al.~\cite{ZhaoEtAl2026},
entropy estimation by Feng, Xu, Li, Guo, and Lin~\cite{FengEtAl2026},
and streaming approximation algorithms for Max-$k$SAT by Wang and
Yang~\cite{WangYang2026}.  Wu and Yang~\cite{WuYang2026} study decoded quantum interferometry
in streams, while Bouland, Ding, and Jain~\cite{BoulandDingJain2026}
study quantum space advantages in random-order streams.  Broader
frameworks include the pair sketch of Kallaugher, Parekh, and
Voronova~\cite{KallaugherParekhVoronovaPair}.

In random-order streaming, the algorithm receives a uniformly random
permutation of a fixed input.  Classical work includes algorithms for order
statistics~\cite{GuhaMcGregor2009} and lower-bound frameworks adapted to
random or almost-random arrival
\cite{ChakrabartiCormodeMcGregor2016,
ChiplunkarKallaugherKapralovPrice2022,AssadiSundaresan2023}.
Our ARHM separation uses this permutation promise to distribute repeated
source data among online requests.  It compares quantum and classical
space under the same promise and quantum space under different orders
of the same multiset.  The triangle application uses random repeated
edges to rebuild estimators.  Our additional NGC lower bound concerns
the ordinary single-arrival model.

\subsection{Consumable quantum data}

Hidden Matching~\cite{BarYossefJayramKerenidis2008} supplies the
one-request primitive.  Gilboa, Jain, and
McClean~\cite{GilboaJainMcClean2025} study multiple requests about one
input through asymmetric direct sums.  Their Multiple Hidden Matching
lower bound is $\bOm{\sqrt m}$ for $m\le N/2$, and they ask whether
linear dependence is possible below the square-root scale.
Their proof already extracts independent parity information from a
spanning forest of the output edges.  With edge-disjoint matchings,
distinctness of the selected edges forces a forest of square-root size.

Our improvement strengthens the combinatorial part of that argument.
Independent balanced requests force every transversal to have a linear
forest for $m\le N^{1/2-\delta}$; fresh sequential requests instead
force growth round by round.  The entropy arguments account for failure
and for output edges chosen using the quantum message.  The simultaneous
result establishes the proposed linear dependence throughout each
polynomially sub-square-root regime, while the sequential result reaches
$m=\floor{\sqrt N}$ at joint success $2/3$.  We use these communication bounds to prove the ARHM lower bounds
when all coordinate updates precede the requests.

\subsection{Classical samples and repeated access}

Quantum oracle sketching, introduced by Zhao et al.~\cite{ZhaoEtAl2026},
constructs coherent operations from a stream of random classical samples.
Their applications show how a small quantum device can process classical
data much larger than its memory.  ARHM's weak phase updates use the
same type of accumulation mechanism, specialized to Boolean coordinates.
We use this construction repeatedly to serve externally arriving requests.
The additional task is to control both the accuracy of replacement states
and their availability while completed states are consumed.  The samples
here come from a fixed multiset without replacement, and the lower bounds
concern the representation retained after coordinate access ends.

Related tradeoffs arise in learning and multipass streaming.
Liu, Raz, and Zhan~\cite{LiuRazZhan2023} study the sample requirements
of learning with bounded classical and quantum memory.
Hamoudi and Magniez~\cite{HamoudiMagniez2019} obtain quantum pass--space
tradeoffs for frequency-moment estimation.  These models vary sample
supply or permit deliberate rereading of an input.  In our constructions,
repeated occurrences are already part of one stream, and the algorithm
has no control over their positions.  ARHM compares different placements
of that supply relative to the requests it must answer.

\subsection{Triangle estimators under different access promises}

Kallaugher's hybrid estimator~\cite{Kallaugher2021} is the direct
algorithmic starting point.  At constant error, its single-arrival,
arbitrary-order space bound is
$\bOt{m^{8/5}\Delta_E^{4/5}/T^{6/5}}$.
The pair-sketch abstraction of Kallaugher, Parekh, and
Voronova~\cite{KallaugherParekhVoronovaPair} describes the destructive
quantum primitive shared by this and other streaming algorithms.
Our replenished estimator reconstructs that primitive from repeated
blocks and refreshes the classical sample, allowing successive observations
to replace simultaneous estimator copies.  This requires new blockwise
accounting for the hybrid estimator.  Duplicate handling,
complementary weights, and the averaging identity are needed to turn
reconstruction into a bound on the underlying simple graph's triangle
count.

We compare these bounds in expected space, using a measure comparable
to Kallaugher's~\cite{Kallaugher2021}; the precise convention and comparison
are given in \cref{sec:triangles}.

Classical single-arrival work includes hybrid sampling
\cite{KallaugherPrice2017} and the optimal estimator of
Jayaram and Kallaugher~\cite{JayaramKallaugher2021}.
Chakraborty et al.~\cite{ChakrabortyEtAl2026} extend triangle counting
to arbitrary repeated-edge streams, obtaining optimal-space algorithms
for the underlying simple graph.  Their promise permits adversarial
multiplicities and placements; our exact multiplicities and random order
supply the coverage needed to rebuild and reuse estimators.  A comparison
of classical and quantum complexity under our particular promise
remains open.

Single-arrival random-order algorithms form another comparison class,
including McGregor and Vorotnikova~\cite{McGregorVorotnikova2020} and
the threshold algorithm of Etgar et al.~\cite{EtgarEtAl2026}.
They exploit arrival order without receiving additional copies of edges.
Gao et al.~\cite{GaoEtAl2026} give quantum communication techniques
for search problems, including streaming lower bounds for finding a
triangle.  That witness-finding task differs from estimating the triangle
count in our repeated-input model.

\subsection{Robust communication and quantum NGC}

Chakrabarti, Cormode, and McGregor~\cite{ChakrabartiCormodeMcGregor2016}
develop robust communication lower bounds for randomly partitioned input.
Assadi and Sundaresan~\cite{AssadiSundaresan2023} use this framework
for Noisy Gap Cycle and derive graph-streaming consequences.  Their
linear classical robust lower bound holds once the cycle length exceeds
a sufficiently large constant times $\log n$.

We extend this framework to quantum communication.  We use the same
layered graph construction, hybrid reduction, and random edge partition,
but change the analysis of Alice's message:
a quantum encoding requires control of matrix Fourier coefficients
instead of scalar biases in the input distribution conditioned on Alice's
classical message.  The matrix
hypercontractive inequality of Ben-Aroya, Regev, and
de Wolf~\cite{BenAroyaRegevDeWolf2008} provides that control.
We state the communication bound for Block Hidden XOR, allowing the
unequal block widths produced by the random partition and the small
prediction advantage left by the reduction.  The resulting streaming
consequences are stated in \cref{cor:qro-consequences-main}.

\section{Preliminaries and models}
\label{sec:prelim}

We use $[n]=\{1,\ldots,n\}$, base-two logarithms unless written $\ln$,
and binary entropy $h_2$.  Entropies and mutual information are measured
in bits.  The notation $\wt O$ suppresses polylogarithmic factors in the
parameters of the relevant statement.  Parameters such as $m,L,r$ are
local to their problem: in particular, $m$ counts requests in Multiple
Hidden Matching and distinct edges in triangle counting.

\subsection{Input order, access, and output}

Our results use three different types of streams.  ARHM has repeated
coordinate updates and requests that must be answered immediately; we
compare random order with the order in which all updates arrive first.
For triangle counting, edges are repeated and the estimate is output
after the pass.  For NGC, each edge appears once and the decision is
also made after the pass.  We specify the common conventions below.

A random-order input is a fixed multiset of classical items, with its
occurrences labeled to distinguish repetitions in the probability space.
The stream is a uniformly random permutation of those occurrences.
Success is required for every allowed underlying input, averaging over
the permutation and the algorithm's randomness and measurements.  Labels
are visible only when included in the problem's input definition.
ARHM explicitly includes update labels and request tags; the repeated-edge
triangle algorithm sees only edges.  In the latter model, $m$ distinct
edges give $Rm$ stream occurrences.

We use the following quantum streaming model.  The input is classical
and is read once, one item at a time.  The algorithm uses a register of qubits, 
and applies a quantum channel every step of the stream. The choice of the quantum channel may depend
on the current item, supplied public parameters, fresh randomness, and
measurement outcomes retained in its classical memory.  Preprocessing
and postprocessing are allowed, as are measurements during the pass.
We do not require measurements to be deferred until the end: doing so
could require storing the earlier items that determine them. Our resource of interest is space.

\subsection{Space complexity}
\label{sec:space-convention}

Space is the total number of qubits in use.  Classical bits count alongside qubits.

We write $\Squ{\mathrm{ord},\eta}(P)$ and
$\Scl{\mathrm{ord},\eta}(P)$ for the minimum quantum and classical
space for problem $P$ under order promise $\mathrm{ord}$ and error at
most $\eta$, using the problem's specified output and success requirements.

For triangle counting we instead use expected space complexity, with a
measure comparable to that used by Kallaugher~\cite{Kallaugher2021}.
We give the precise convention in \cref{sec:triangles}.

\subsection{One-way quantum communication}
\label{sec:communication-convention}

Throughout, Alice and Bob share no prior entanglement.  Alice sends
one message of dimension at most $2^q$, including any transmitted classical
information.  Bob's local workspace is unrestricted, but starts without
additional information about Alice's input beyond the specified problem
inputs and permitted public randomness.  Alice's encoding cannot depend
on Bob's private requests.  In sequential models, future requests remain
unavailable until their rounds, including through public coins.

Independent public classical coins may be shared.  For an input $X$ and
its $q$-qubit encoding $Q$, the dimension bound is $I(X:Q)\le q$;
with public coins $C$, it is $I(X:Q\mid C)\le q$.  Data processing
bounds the information in any output obtained from that encoding.
We write $Q^\rightarrow$ for one-way quantum communication complexity;
the default error is worst-case joint error at most $1/3$, with other
or distributional guarantees stated explicitly.  Robust NGC additionally
averages over independent fair edge ownership for each fixed graph.

\subsection{The phase encoding for Hidden Matching}

For even $N$ and $x\in\bool^N$, define
\[
 \ket{\psi_x}=\frac1{\sqrt N}\sum_{j=1}^N(-1)^{x_j}\ket j.
\]
Given a perfect matching $M$, measure in the orthonormal basis
\[
 \left\{\frac{\ket u+(-1)^b\ket v}{\sqrt2}:
       \{u,v\}\in M,\ b\in\bool\right\}.
\]
For an edge $\{u,v\}$, the amplitude of the outcome with
$b\ne x_u\oplus x_v$ vanishes.  Each edge is returned with probability
$2/N$, together with its correct parity.  This is an exact one-request
protocol using $\ceil{\log N}$ qubits
\cite{BarYossefJayramKerenidis2008}.  The matching is known before
measurement, and the edge is chosen by its outcome.  The lower bounds
for repeated requests impose no requirement that the encoding or decoder
have this form.

\section{Replenishment with explicit requests: ARHM}
\label{sec:replenishment}

In this section we introduce Atomized Refreshable Hidden Matching (ARHM)
and prove the three-way separation of \cref{thm:ext-param-threeway}.
The stream gives coordinate updates for one fixed string, together with
matching requests that must be answered before the next item arrives.
We compare two orders of the same multiset: uniformly random order,
and an order in which all updates precede the requests.

Our quantum algorithm uses later coordinate updates to prepare new
states while requests consume the states already prepared.  We show
that only a small number of registers are needed in random order.
We then prove a classical lower bound under the same promise, followed
by a quantum lower bound for the update-first order using the
consumability results of \cref{sec:consumability}.

\subsection{Parameterized atomized Hidden Matching}
\label{sec:arhm-definition}
Fix an even $N$ and a public one-factorization
$\Fc=\set{M_1,\ldots,M_{N-1}}$ of $K_N$: each $M_i$ is a perfect
matching, and each unordered pair belongs to exactly one $M_i$.
Recall from \cref{sec:prelim} that the phase state $\ket{\psi_x}$ answers
one matching request exactly.

\begin{restatable}[Parameterized Atomized Refreshable Hidden Matching]{definition}{arhmDefinition}
\label{def:arhm}
Choose an integer $2\le r\le\floor{\sqrt N}$ and sufficiently large
absolute integer constants $A,C$.  Set
\begin{equation}\label{eq:ext-parameters}
\mu=Cr,\qquad L=A\mu r=ACr^2,\qquad
U=NL,\qquad H=\mu N,\qquad B=U/H=Ar.
\end{equation}
For a string $x\in\bool^N$, the stream $\ARHM_{N,r}$ contains
$L$ labeled update occurrences $(j,x_j,\ell)$ for each $j\in[N]$
and $\ell\in[L]$, and $r$ labeled requests $Q_{t,i_t}$, where
$t\in[r]$ and $i_t\in[N-1]$.  A request must immediately be answered
by $(e_t,b_t)$ with $e_t=\set{u_t,v_t}\in M_{i_t}$ and
$b_t=x_{u_t}\oplus x_{v_t}$.  Success means that all $r$ answers are
correct.  In random order, the $U+r$ labeled occurrences are uniformly
permuted; the tag $t$ identifies a request, not its arrival position.
\end{restatable}

Here $L$ is the repetition count per coordinate and $U$ is the total number
of updates.  The algorithm groups updates into $B$ batches of size $H$;
each batch contains $\mu$ occurrences of each coordinate in expectation.
The choices $\mu=Cr$ and $B=Ar$ control decoding error and the arrival of
requests, respectively.

\paragraph{Why immediate output?}
Outputting answers as they are obtained avoids having to store the output
list in working memory.  For the algorithm below, postponing output until
the end would require keeping the answers already obtained; storing $r$
edges and their parity bits explicitly uses $\bT{r\log N}$ classical bits.
Instead, each answer is written to the write-only output tape and need
not be retained, as in \cref{sec:space-convention}.  The requirement to
answer before the next stream item also lets us apply the sequential
consumability lower bound in the update-first order.

\begin{definition}[Update-first order]
\label{def:update-first}
All $U$ updates precede all $r$ requests, which are then revealed and
answered one at a time.  Correctness is required for every input, request
list, and legal order.
\end{definition}

\paragraph{Beyond update-first order.}
Putting all updates first is a convenient extreme, but the lower bound
does not require such a complete separation.  Updates and requests may
be interleaved in the prefix, provided sufficiently many fresh requests
remain after the last update.  More precisely,
\cref{prop:last-replenishment-point} shows that a suffix of
$s\le\sqrt N$ independent uniform matching factors, independent of the
input and prefix, requires $S\ge s/2-1$ space for joint success at least
$2/3$, under the simulation and memory-accounting conditions stated
there.  In particular, a growing suffix with $s=\bT{r}$ gives the same
$\bOm{r}$ lower bound.  Thus even a well-mixed prefix cannot compensate
for a long suffix of fresh requests with no further updates: the data
must remain available when the algorithm needs to rebuild its states.

\begin{restatable}[Parameterized three-way separation]{theorem}{arhmSeparationTheorem}
\label{thm:ext-param-threeway}
Let $r=r(N)\to\infty$ with $2\le r\le\floor{\sqrt N}$.
For sufficiently large absolute integer constants $A,C$ and joint error
at most $1/3$, \cref{def:arhm} satisfies
\begin{equation}\label{eq:ext-triple}
\boxed{\begin{aligned}
\Squ{\mathrm{rand},1/3}(\ARHM_{N,r})&=\bO{\log N\log(r+1)},\\
\Scl{\mathrm{rand},1/3}(\ARHM_{N,r})&=\bOm{\sqrt N},\\
\Squ{\mathrm{upd\text{-}first},1/3}(\ARHM_{N,r})&=\bOm{r},
\end{aligned}}
\end{equation}

where $U+r =\bT{Nr^2}$ is the length of the stream.
\end{restatable}

For polynomially growing $r$, the quantum random-order bound is
polylogarithmic in $N$, while both lower bounds are polynomial.
For the classical lower bound, it will suffice to consider the first
request, since joint success implies a correct answer to that request.
The quantum lower bound uses the whole sequence.  We give variants
with weaker correctness requirements in \cref{app:consumability-extra}.

\subsection{Random-order quantum upper bound: rebuilding from later data}
\label{sec:arhm-upper}

We maintain two banks of registers.  One holds completed states that
can answer requests, while updates prepare states in the other.
Once a build is complete, we exchange the roles of the two banks.
It is therefore enough to keep as many completed states as we need
before the next build finishes.
Partition the updates, in their arrival order, into $B=Ar$ batches of
$H=\mu N$ updates each.  Requests do not count toward a batch's length.
Set $P=\ceil{K\log(r+1)}$, for a sufficiently large absolute constant $K$.
Starting from the uniform state, each update adds a small phase:
\[
  \ket{+_N}=\frac1{\sqrt N}\sum_{j=1}^N\ket j,
  \qquad
  V_{j,x_j}=\exp\!\left(i\frac{\pi x_j}{\mu}\ket j\!\bra j\right).
\]
If every coordinate occurred exactly $\mu$ times in a batch, these
updates would prepare $\ket{\psi_x}$ exactly.  We will bound the error
caused by the actual coordinate counts.

\paragraph{Algorithm.}
Maintain a \textsc{Ready} bank and a \textsc{Build} bank, each containing
at most $P$ registers of dimension $N$, and an update counter $h$.
\begin{enumerate}[leftmargin=1.8em]
\item \textbf{Initialize.}  Leave \textsc{Ready} empty, prepare $P$
copies of $\ket{+_N}$ in \textsc{Build}, and set $h=0$.
\item \textbf{Update $(j,x_j,\ell)$.}  Apply $V_{j,x_j}$ to each
\textsc{Build} register and increment $h$.  When $h=H$, discard the old
\textsc{Ready} bank, make \textsc{Build} the new \textsc{Ready} bank,
prepare a fresh \textsc{Build} bank, and reset $h=0$.
\item \textbf{Request $Q_{t,i_t}$.}  If \textsc{Ready} is empty, declare
failure.  Otherwise consume one of its registers by measuring in the basis
\[
  \left\{\frac{\ket u+(-1)^b\ket v}{\sqrt2}:
           \set{u,v}\in M_{i_t},\ b\in\bool\right\},
\]
and output the observed edge and bit.  Leave \textsc{Build} and $h$
unchanged.
\end{enumerate}

\paragraph{Relation to quantum oracle sketching.}
The phase update specializes the mechanism of
Zhao et al.~\cite{ZhaoEtAl2026}: their operation
$\exp(i\tau f(j)\ket j\!\bra j/M)$ becomes $V_{j,x_j}$ with
$f(j)=x_j$, $M=\mu N$, and $\tau=\pi N$.
Here coordinate counts are hypergeometric because the stream permutes a
fixed multiset.  We need accuracy only for the matching measurement on the
accumulated state.  The additional streaming issue is to keep enough
completed states available while later updates build replacements.

The analysis has two parts: a completed register has small expected
decoding error, and random interleaving keeps demand on each bank small.

\subsubsection{A completed build register is accurate}
\label{app:atomized-rebuilding}

\paragraph{Hypergeometric sampling.}
Write $Y\sim\operatorname{Hypergeom}(a,b,h)$ when $Y$ counts marked items
in a uniformly chosen sample of $h$ distinct items from a population of
$a$ items, of which $b$ are marked.  For integers $a\ge2$ and
$0\le b,h\le a$, put $p=b/a$.  Then
\[
 \E{Y}=hp,\qquad
 \operatorname{Var}(Y)=hp(1-p)\frac{a-h}{a-1}\le hp.
\]
Indeed, each sampled position has a Bernoulli indicator of mean $p$,
and two distinct positions have covariance $-p(1-p)/(a-1)$.
The factor $(a-h)/(a-1)$ accounts for sampling without replacement.
For a one-item population the count is deterministic and has variance zero.

\begin{lemma}[Atomized rebuilding]
\label{lem:atomized-rebuilding}
\label{lem:app-atomized-rebuilding}
Fix $x\in\bool^N$, a perfect matching $M$, and a batch of $H=\mu N$
positions in the uniformly random update permutation.  Let $p_{\rm err}$
be the probability that a register built from this batch reports an
incorrect parity when measured in the basis of $M$.  Then
$\E{p_{\rm err}}\le\pi^2/\mu$, where the expectation is over the update
permutation.  The bound is uniform in $x$ and $M$.
\end{lemma}

\begin{proof}
Let $C_j$ count occurrences of coordinate $j$ in the batch.  The
commuting phase updates produce
\[
  \ket{\phi_C(x)}
   =\frac1{\sqrt N}\sum_{j=1}^N e^{i\pi C_jx_j/\mu}\ket j
   =\frac1{\sqrt N}\sum_{j=1}^N(-1)^{x_j}e^{i\pi x_j(C_j-\mu)/\mu}\ket j.
\]
The measurement selects each edge $\set{u,v}\in M$ with probability $2/N$.
Conditioned on that edge, the wrong-parity probability is
$(1-\cos(\pi[x_u(C_u-\mu)-x_v(C_v-\mu)]/\mu))/2$.  Thus, using
$1-\cos z\le z^2/2$ and $(a-b)^2\le2a^2+2b^2$,
\begin{equation}\label{eq:app-rebuild-error}
\begin{aligned}
  p_{\rm err}
  &=\frac1N\sum_{\set{u,v}\in M}\left(1-\cos\left(\frac{\pi[x_u(C_u-\mu)-x_v(C_v-\mu)]}{\mu}\right)\right)\\
  &\le\frac1N\sum_{j=1}^N\left(\frac{\pi x_j(C_j-\mu)}{\mu}\right)^2
   \le\frac{\pi^2}{N\mu^2}\sum_{j=1}^N(C_j-\mu)^2.
\end{aligned}
\end{equation}
Each $C_j$ has distribution $\operatorname{Hypergeom}(NL,L,\mu N)$,
with mean $\mu$ and variance
$\mu(1-1/N)(NL-\mu N)/(NL-1)\le\mu$.
Taking expectations in \eqref{eq:app-rebuild-error} gives
$\E{p_{\rm err}}\le\pi^2/\mu$.
\end{proof}

\begin{remark}[Several registers from one batch]
\label{rem:app-bank-correlations}
Conditional on the count vector $C$, a completed bank is in the product
state $\ket{\phi_C(x)}^{\otimes P}$.  Its registers share the same
approximation error, but each matching measurement has error at most
$\pi^2\sum_j(C_j-\mu)^2/(N\mu^2)$ by
\eqref{eq:app-rebuild-error}.  Serving $q\le P$ requests therefore has
conditional failure probability at most
$ q\pi^2\sum_j(C_j-\mu)^2/(N\mu^2)$ by a union bound.
Independence of the approximation errors is unnecessary.
\end{remark}

\subsubsection{Random interleaving limits requests between rebuilds}
\label{app:atomized-occupancy}

Generate the random stream by three independent choices: the order of
the labeled updates, the order of the labeled requests, and a uniform
interleaving of $U$ update symbols and $r$ request symbols.  The first
choice determines the states; the latter two determine which bank serves
each request and which matching it asks for.

Let $G_j$ be the number of requests arriving after exactly $j$ updates,
for $0\le j\le U$.  The vector $(G_0,\ldots,G_U)$ is uniform over weak
compositions of $r$ into $U+1$ parts: each such vector specifies exactly
one interleaving.  In particular, $\E{G_j}=r/(U+1)$.

The requests before the first bank is ready occupy gaps $0,\ldots,H-1$.
For $1\le b<B$, bank $b$ serves gaps $bH,\ldots,(b+1)H-1$; bank $B$
serves only gap $U$.  Thus every completed bank serves at most $H$ gaps.

\begin{lemma}[Queries during one rebuild]
\label{lem:queries-one-rebuild}
\label{lem:app-atomized-occupancy}
Let $Y$ count requests in any fixed set of at most $H$ update gaps.
For $1\le k\le\min\{r,H\}$,
\[
  \prob{Y\ge k}
  \le\binom rk\frac{H^{\overline k}}{(U+1)^{\overline k}}
  \le\left(\frac{2e}{Ak}\right)^k,
  \qquad a^{\overline k}:=a(a+1)\cdots(a+k-1).
\]
\end{lemma}

\begin{proof}
Enlarging the set of gaps can only increase $Y$, so it suffices to
consider exactly $H$ gaps.  Counting weak compositions gives
\begin{equation}\label{eq:app-beta-binomial}
  \prob{Y=t}
  =\frac{\binom{t+H-1}{H-1}\binom{r-t+U-H}{U-H}}
         {\binom{r+U}{U}},\qquad 0\le t\le r.
\end{equation}
This is the beta-binomial law with parameters $(r,H,U+1-H)$.
For a general beta-binomial variable $Z$ with parameters $(n,\alpha,\beta)$,
where $n$ is a nonnegative integer and $\alpha,\beta>0$, the mass function is
\[
 \prob{Z=t}=\binom nt
 \frac{\alpha^{\overline t}\beta^{\overline{n-t}}}
      {(\alpha+\beta)^{\overline n}},\qquad 0\le t\le n,
\]
with $a^{\overline0}=1$.  Its mean and variance are
\[
 \E{Z}=\frac{n\alpha}{\alpha+\beta},\qquad
 \operatorname{Var}(Z)=
 \frac{n\alpha\beta(\alpha+\beta+n)}
      {(\alpha+\beta)^2(\alpha+\beta+1)}.
\]
These follow from the factorial moments; below we use their general
binomial-moment form specialized to $Y$.
Multiplying by $\binom tk$ and summing, using Vandermonde's identity,
gives its binomial moment
\begin{equation}\label{eq:app-occupancy-moment}
  \E{\binom Yk}=\binom rk\frac{H^{\overline k}}{(U+1)^{\overline k}}.
\end{equation}
Since $\binom Yk\ge1$ on $Y\ge k$, Markov's inequality proves the first
bound.  For the second, use $H^{\overline k}\le(2H)^k$ for $k\le H$,
$(U+1)^{\overline k}\ge U^k$, $\binom rk\le(er/k)^k$, and
$H/U=1/(Ar)$.
\end{proof}

\begin{corollary}[No bank exhaustion]
\label{cor:app-no-bank-exhaustion}
For $P=\ceil{K\log(r+1)}$ and sufficiently large absolute $K$, the
probability that any completed bank must serve more than $P$ requests
is $\lO{1}$ as $r\to\infty$.
\end{corollary}
\begin{proof}
If $P\ge r$, exhaustion is impossible.  Otherwise $P\le r\le H$,
and a union bound over all $Ar$ completed banks, including the final
bank, gives
\[
  \prob{\text{some bank must serve more than }P\text{ requests}}
  \le Ar\left(\frac{2e}{AP}\right)^P=\lO{1}.
\]
We bounded the larger event of at least $P$ requests at a bank.
\end{proof}

The initial interval is different because no bank is ready there.  The
probability of any request before the first bank is ready satisfies
\begin{equation}\label{eq:app-initial-query}
  \prob{\sum_{j=0}^{H-1}G_j>0}
  \le\E{\sum_{j=0}^{H-1}G_j}=\frac{rH}{U+1}\le\frac1A.
\end{equation}

\begin{proposition}[Quantum random-order upper bound]
\label{prop:arhm-quantum-upper}
\label{prop:arhm-qrand-main}
For growing $r\le\floor{\sqrt N}$ and sufficiently large absolute
constants $A,C$,
$\Squ{\mathrm{rand},1/3}(\ARHM_{N,r})=\bO{\log N\log(r+1)}$.
\end{proposition}

\begin{proof}
Choose $P$ as in \cref{cor:app-no-bank-exhaustion}.  Call a schedule good
if no request precedes the first completed bank and no bank exhausts its
$P$ registers.  By \eqref{eq:app-initial-query} and the corollary, the
schedule is bad with probability at most $1/A+\lO{1}$.

Fix a good interleaving and an order of the labeled requests.  These
choices identify the bank and matching used by each request, while the
update permutation remains uniform.  By \cref{lem:atomized-rebuilding},
each request has error at most $\pi^2/\mu$, averaged over that permutation
and its measurement.  A union bound gives decoding failure at most
$r\pi^2/\mu=\pi^2/C$.  This does not condition on previous answers
being correct and does not require independence between banks or their
errors.  Averaging over good schedules, total failure is at most
\[
  \frac1A+\frac{\pi^2}{C}+\lO{1}<\frac13
\]
for sufficiently large $A,C$ and then sufficiently large $N$.

At most $2P$ registers are live, each using $\ceil{\log N}$ qubits.
Counters and classical control use $\bO{\log N}$ bits, since
$H,U$ are polynomial in $N$.  Initialization, phase updates, and
matching measurements can be performed one register at a time.
In the unrestricted-operation space model, these fixed public operations
need no stored matching table; their outcomes and control require only
$\bO{\log N}$ bits.  Discarding the old bank before allocating the next
keeps the complete live memory, including workspace, within
$\bO{\log N\log(r+1)}$.
\end{proof}

\subsection{Classical random-order lower bound}
\label{sec:arhm-classical-lower}
\label{app:factorized-hm-classical}

We prove the classical streaming lower bound by reducing to the first
request.  For this we need a Hidden Matching communication lower bound
when Bob's matching comes from the public one-factorization.
Alice receives $x$, and Bob receives a factor index $I$ and must return
an edge of $M_I$ together with its parity.

\begin{lemma}[Classical factorized Hidden Matching]
\label{lem:classical-factorized-hm}
\label{lem:factorized-hm-classical-main}
For any one-factorization $\Fc$ of $K_N$ and fixed $p>1/2$, a
public-coin classical one-way protocol with success at least $p$ requires
$\bOm[p]{\sqrt N}$ bits.  This holds already for independent uniform
$X\in\bool^N$ and $I\in[N-1]$.
\end{lemma}

\begin{proof}
First fix a deterministic $c$-bit protocol.  For each nonempty message
class $A_t$, let $\alpha_t=\abs{A_t}/2^N$ and define its parity biases
by $\beta_{\{u,v\}}(t)=\E{(-1)^{X_u+X_v}\mid X\in A_t}$.
The level-two Fourier inequality gives
\begin{equation}\label{eq:app-level-two}
  \sum_{e\in E(K_N)}\beta_e(t)^2
  \le e^2\ln^2(e/\alpha_t)
\end{equation}

For the indicator of $A_t$, scalar hypercontractivity
\cite{BenAroyaRegevDeWolf2008} gives
\[
 \lambda^2\sum_{|S|=2}
 \left(\E[X]{\ind{}_{A_t}(X)(-1)^{\sum_{j\in S}X_j}}\right)^2
 \le\alpha_t^{2/(1+\lambda)}
\]
for $0<\lambda\le1$.  Dividing by $\alpha_t^2$ gives the bound
on the $\beta_e(t)$ values after taking
$\lambda=1/\ln(e/\alpha_t)$: the resulting bound is
$e^2\ln^2(e/\alpha_t)$.

Given $(t,I)$, Bob's optimal parity-prediction advantage is
$\frac12\max_{e\in M_I}\abs{\beta_e(t)}$.  Each edge belongs to exactly
one factor, so Jensen's inequality and \eqref{eq:app-level-two} give
\begin{equation}\label{eq:app-factorized-message-advantage}
\begin{aligned}
  \E[I]{\max_{e\in M_I}\abs{\beta_e(t)}}
  &\le\left(\E[I]{\sum_{e\in M_I}\beta_e(t)^2}\right)^{1/2}\\
  &=\left(\frac1{N-1}\sum_e\beta_e(t)^2\right)^{1/2}
   \le\bO{\frac{1+\ln(1/\alpha_t)}{\sqrt N}}.
\end{aligned}
\end{equation}
The message $T$ has distribution $\prob{T=t}=\alpha_t$ and entropy
$H(T)\le c$.  Averaging the last bound and using
$\sum_t\alpha_t\ln(1/\alpha_t)=(\ln2)H(T)$ yields
\begin{equation}\label{eq:app-factorized-classical-advantage}
  \prob{\mathrm{success}}-\frac12
  \le\bO{\frac{c+1}{\sqrt N}}.
\end{equation}
Constant advantage therefore requires $c=\bOm{\sqrt N}$.
For a randomized protocol, fix public and private random tapes preserving
its average success under the same input distribution, and apply the
deterministic bound.
\end{proof}

\begin{proposition}[Classical random-order lower bound]
\label{prop:arhm-classical-lower}
\label{prop:arhm-crand-main}
$\Scl{\mathrm{rand},1/3}(\ARHM_{N,r})=\bOm{\sqrt N}$.
\end{proposition}
\begin{proof}
Let an $S$-bit streaming algorithm have joint success at least $2/3$.
Given Alice's string $x$ and Bob's factor index $i$, take the ARHM instance
whose $r$ requests all ask for $M_i$.  Using independent public randomness,
sample a uniform update permutation, a uniform permutation of the request
tags, and a uniform interleaving.  This generates exactly the random-order
stream distribution.

Alice simulates the updates before the first request and sends the
complete $S$-bit state to Bob.  She needs no knowledge of $i$ to do so.
Bob inserts the first labeled request and returns the algorithm's answer.
Joint streaming success implies success on this first request, so the
communication protocol succeeds with probability at least $2/3$.
Apply \cref{lem:classical-factorized-hm}.
\end{proof}

\subsection{When later data is unavailable}
\label{sec:arhm-update-first}
\label{app:arhm-update-first-transfer}

Suppose now that all updates arrive before the first request.
The algorithm must answer every request using the state it kept after
the updates.  We apply the sequential consumability bound to the public
factorization $\Fc$ to show that this state must be large.  The argument
allows any encoding of $x$.

\begin{definition}[Sequential factorized Hidden Matching]
\label{def:sequential-factorized-hm}
In $\SeqHM^{\Fc}_{N}(r)$, Alice receives $X\in\bool^N$ and sends one
quantum message.  Bob then receives independent uniform factor indices
$I_1,\ldots,I_r\in[N-1]$ one at a time.  Before learning $I_{t+1}$,
he must output $e_t=\set{u_t,v_t}\in M_{I_t}$ and a bit $b_t$ intended
to equal $X_{u_t}\oplus X_{v_t}$.  He may retain arbitrary quantum memory
and use adaptive operations.  The indices are independent of $X$, the
encoder's randomness, and the initial protocol systems; future indices
are unavailable until their rounds.
\end{definition}

\begin{corollary}[Sequential factorized Hidden Matching]
\label{cor:sequential-factorized-hm}
For uniform $X$, independent of the factor indices and public coins,
a $q$-qubit protocol with joint success at least $p\in[0,1]$
satisfies
\begin{equation}\label{eq:factorized-sequential-blackbox}
  q\ge pr-1-\frac{r(r^2-1)}{6(N-1)}.
\end{equation}
In particular, if $2\le r\le\floor{\sqrt N}$ and $p\ge2/3$, then
$q\ge r/2-1$, giving $q=\bOm{r}$ for growing $r$.
\end{corollary}
\begin{proof}
A uniform factor is balanced: each unordered pair appears with probability
$1/(N-1)$.  Apply \cref{cor:app-seqmhm-joint} with $m=r$.
For $r\le\sqrt N$, the loss term is at most $r/6$.
\end{proof}

\begin{proposition}[Update-first quantum lower bound]
\label{prop:arhm-update-first}
For growing $r\le\floor{\sqrt N}$,
\[
  \Squ{\mathrm{upd\text{-}first},1/3}(\ARHM_{N,r})=\bOm{r}.
\]
\end{proposition}
\begin{proof}
Alice simulates an $S$-space algorithm through all updates in any fixed
legal order and sends its complete live state to Bob.  This state has
dimension at most $2^S$, including all classical memory.  Bob feeds fresh
uniform factor requests into the simulation one at a time and returns
its answers.  No further update is needed.  The protocol has the
algorithm's joint success probability, at least $2/3$, so
\cref{cor:sequential-factorized-hm} gives $S\ge r/2-1$.
\end{proof}

The same reduction applies at any cut after the last coordinate update,
even if earlier requests have already been answered.  The following
statement records precisely the independence needed at such a cut.

\begin{proposition}[Lower bound after the last replenishment point]
\label{prop:last-replenishment-point}
Let $N\ge2$ be even, $s\ge1$, and $X\sim\Uni{\bool^N}$.
Fix the algorithm and public one-factorization independently of $X$.
Consider an execution with a cut satisfying the following conditions.
\begin{enumerate}[leftmargin=1.8em]
\item Alice can simulate the prefix and choose the cut from $X$ and
permitted randomness, without knowing future requests.  The prefix may
contain arbitrary coordinate updates and earlier requests whose indices
are fixed or generated by public randomness independent of $X$.
The experiment runs for every $X$, without conditioning on prefix success
or any other favorable outcome.
\item After the cut, no coordinate update arrives.  The suffix consists
of $s$ independent uniform factors, independent of $X$, the prefix, and
all prefix randomness and outcomes, revealed and answered sequentially.
\item Every input-dependent clock, control value, or earlier answer
needed to resume is part of the counted live state.  The suffix decoder
receives no such information separately.
\end{enumerate}
If the $S$-space algorithm answers all suffix requests correctly
with probability at least $p$, averaged over this experiment, then
\[
  S\ge ps-1-\frac{s(s^2-1)}{6(N-1)}.
\]
For $s\le\sqrt N$ and $p\ge2/3$, this gives $S\ge s/2-1$, and hence
$S=\bOm{s}$ as $s\to\infty$.
\end{proposition}
\begin{proof}
Alice simulates the entire prefix, including earlier requests, and sends
the complete state at the cut.  Keep independent public randomness as
shared coins; private outcomes and input-dependent information remain
inside the simulated state.  Bob continues on the fresh suffix requests.
This is an $S$-qubit protocol for $\SeqHM^{\Fc}_{N}(s)$
with the same joint success.  Apply
\cref{cor:sequential-factorized-hm} with $r=s$.
\end{proof}

\begin{remark}[Timing, not merely multiplicity]
\label{rem:timing-not-multiplicity}
The prefix may contain many complete repetitions of the coordinate data
and earlier request phases.  Once the final update has passed, the live
state must support the remaining fresh requests by itself.  Thus earlier
repeated access cannot replace the information required at this cut;
ARHM's space saving uses access interleaved with the requests.
\end{remark}

\begin{proof}[Proof of \cref{thm:ext-param-threeway}]
Combine the update count in \eqref{eq:ext-parameters} with
\cref{prop:arhm-quantum-upper,prop:arhm-classical-lower,prop:arhm-update-first}.
\end{proof}

\begin{corollary}[Power-law update--space tradeoff]
\label{cor:ext-alpha}
For fixed $\alpha\in(0,1/2)$ and $r=\floor{N^\alpha}$,
\[
\begin{aligned}
 U&=\bT{N^{1+2\alpha}},&
 \Squ{\mathrm{rand},1/3}&=\bO{\log^2N},\\
 \Scl{\mathrm{rand},1/3}&=\bOm{\sqrt N},&
 \Squ{\mathrm{upd\text{-}first},1/3}&=\bOm{N^\alpha}.
\end{aligned}
\]
At $r=\floor{\sqrt N}$, the sequential update-first lower bound is
$\bOm{\sqrt N}$ with $U=\bT{N^2}$.
\end{corollary}
\begin{proof}
Substitute the stated choices of $r$ into \cref{thm:ext-param-threeway}.
\end{proof}

\subsection{Advance knowledge of the requests does not replace later data}
\label{sec:ext-batch}

One might ask whether the lower bound depends on answering each request
before seeing the next one.  Suppose instead that, after all updates,
the decoder receives the entire request list before it must answer.
For $r\le N^{1/2-\delta}$, the simultaneous consumability bound still
gives a linear lower bound on the state retained after the updates.

\begin{definition}[Post-data batch-request model]
\label{def:ext-batch-model}
An $S$-qubit post-data representation is the complete state retained after
all $U$ updates, with no request indices exposed.  The update-prefix
order is generated without knowledge of the future factor list, including
indirect encodings of that list in the order.  The decoder then receives
the entire list $I_1,\ldots,I_r$ as a separate read-only input and outputs
all requested edge/parity pairs jointly.

There is no prior entanglement, as in
\cref{sec:communication-convention}.  All retained classical and quantum
information about the prefix counts toward $S$; the supplied request list
and unrestricted local decoding workspace do not.  Thus the decoder is
more permissive than an update-first streaming decoder: it knows every
request before answering and has unrestricted workspace, but receives no
further coordinate data.
\end{definition}

\begin{theorem}[Post-data batch-request quantum lower bound]
\label{thm:ext-batch-lower}
Fix $\delta\in(0,1/2)$ and $p_0\in(0,1]$.  For growing
$r\le N^{1/2-\delta}$, any post-data representation for $\ARHM_{N,r}$
with joint success at least $p_0$ for every input, legal prefix order,
and factor list requires $S=\bOm[{\delta,p_0}]{r}$.
The same bound holds for joint success at least $p_0$ averaged over
uniform $X$ and independent uniform factors, for any specified prefix-order
rule independent of the future factors.  Both statements allow arbitrary
joint decoding and count all retained classical information in $S$.
\end{theorem}
\begin{proof}
Alice simulates the update prefix and sends the complete retained state.
For the worst-case statement, choose any fixed legal prefix order; for
the distributional statement, simulate the specified rule.  The message
depends only on $X$ and permitted randomness, since the prefix has no
access to future factors.  Bob receives independent uniform factors
simultaneously and invokes the decoder.  The message uses at most $S$
qubits, embedding classical bits in the computational basis, and preserves
joint success.  Apply \cref{cor:ext-factor-linear} with $m=r$;
the condition $r\to\infty$ eventually meets its lower threshold.
\end{proof}

\begin{corollary}[Parameterized separation even against batch revelation]
\label{cor:ext-alpha-batch}
Fix $\alpha\in(0,1/2)$ and set $r=\floor{N^\alpha}$.
The random-order problem has $U=\bT{N^{1+2\alpha}}$ updates and a
$\bO{\log^2N}$-space quantum algorithm.  Every post-data representation
for the same updates and requests satisfying the joint-success and prefix
conventions of \cref{thm:ext-batch-lower}, with $p_0=2/3$, requires
$\bOm[\alpha]{N^\alpha}$ retained qubits even under batch revelation.
\end{corollary}
\begin{proof}
Use \cref{cor:ext-alpha} for the upper bound and
\cref{thm:ext-batch-lower} with $\delta=1/2-\alpha$ for the lower bound.
\end{proof}

\begin{remark}[No batched-request random-order upper bound is claimed]
\label{rem:ext-no-batched-upper}
The upper bound uses individual requests interspersed with updates.
The batch lower bound supplies all requests together only after the data
has passed.  It shows that advance knowledge of the requests does not
remove the post-data space cost; it makes no upper-bound claim for an
indivisible batch of requests placed randomly in the stream.
\end{remark}

\paragraph{Open questions.}
Our construction uses $\bT{Nr^2}$ coordinate updates to control the
phase error and supply enough completed registers.  Can fewer updates
give the same space separation?  It is also natural to ask which weaker
arrival assumptions suffice in place of a uniform permutation.  The
upper bound needs accurate states from each build interval and few
enough requests before the next interval finishes, while
\cref{prop:last-replenishment-point} gives a lower bound when a long
suffix of fresh requests receives no further updates.

We next prove the consumability bounds used in these reductions.

\section{Linear consumability of Hidden Matching}
\label{sec:consumability}

The Hidden Matching protocol uses a fresh phase state for each request.
Could a different encoding answer many requests with less communication?
We show that the message length must grow linearly with the number of
requests in the ranges below, even when Alice's encoding is arbitrary.
This strengthens the Multiple Hidden Matching lower bound of Gilboa,
Jain, and McClean~\cite{GilboaJainMcClean2025} and gives the quantum
lower bounds for ARHM when no further updates are available.

We consider two ways of revealing the requests.  In the simultaneous
model, Bob receives all matchings before decoding and may choose his
output edges jointly.  We prove a linear lower bound for
$m\le N^{1/2-\delta}$, for every fixed $\delta\in(0,1/2)$.
In the sequential model, Bob must answer each request before seeing
the next one.  Here our linear bound extends through $m=\floor{\sqrt N}$.

Both proofs use the forest-information argument of Gilboa, Jain, and
McClean.  Correct parities on the edges of a forest give independent
constraints on Alice's string.  It therefore suffices to show that
Bob's output edges contain a large forest and to account for the
possibility of incorrect answers.  For simultaneous requests, we prove
this for every choice of one edge from each matching.  For sequential
requests, we analyze how the forest grows as each new matching arrives.

\subsection{Repeated requests and Multiple Hidden Matching}
\label{subsec:consumable-data}
\label{subsec:mhm-background}

For a relation $\Rc\subseteq\Xc\times\Yc\times\Zc$, a one-way protocol
gives Alice $x$, Bob $y$, and asks Bob to output $z$ with
$(x,y,z)\in\Rc$ after receiving one message.  To study reuse of that
message, keep Alice's input fixed and give Bob several requests.

\begin{definition}[Asymmetric direct sum]
\label{def:asymmetric-direct-sum}
For $m\ge1$, the relation $\Rc^{(m)}$ gives Alice one input $x$ and Bob
inputs $y_1,\ldots,y_m$.  After one message from Alice, Bob must output
$z_1,\ldots,z_m$ with $(x,y_i,z_i)\in\Rc$ for every $i\in[m]$.
\end{definition}

Following \cite{GilboaJainMcClean2025}, the data is \emph{consumable}
when the required communication grows substantially with $m$.
Our target is $Q^\rightarrow(\Rc^{(m)})=\bOm{m}$, where
$Q^\rightarrow$ denotes one-way quantum communication with
worst-case joint error at most $1/3$, unless specified otherwise.
There is only one input $x$: the cost comes from answering many requests
about the same object.  A classical message can be copied, whereas a
measurement may consume information in a quantum message.

Throughout this section, $N\ge2$ is even and request counts are positive
integers.  All message budgets include transmitted classical information.

\begin{definition}[Hidden Matching]
\label{def:hidden-matching}
In $\HM_N$, Alice receives $x\in\bool^N$ and Bob receives a perfect
matching $M$ on $[N]$.  Bob must output an edge $\set{u,v}\in M$ and
its parity $x_u\oplus x_v$.
\end{definition}

The phase state $\ket{\psi_x}$ from \cref{sec:prelim} gives an exact
$\bO{\log N}$-qubit protocol \cite{BarYossefJayramKerenidis2008}:
measuring in the matching basis returns one edge and its correct parity.

\multipleHiddenMatchingDefinition

Sending $m$ copies of $\ket{\psi_x}$ uses $\bO{m\log N}$ qubits.
Gilboa, Jain, and McClean~\cite{GilboaJainMcClean2025} proved
$Q^\rightarrow(\MHM_{N,m})=\bOm{\sqrt m}$ for $m\le N/2$.
Our linear lower bound allows an arbitrary message and joint decoding,
including choices of output edges that depend on the message.

\subsection{Forests as independent parity information}
\label{subsec:mhm-forest-information}
\label{app:mhm-forest-information}

For an edge set $F$, let $V(F)$ be its incident vertices.
\begin{definition}[Forest size]
\label{def:forest-size}
The forest size is
$\fs(F):=\max\set{\abs{T}:T\subseteq F\text{ is a forest}}$,
the number of edges in a spanning forest of $(V(F),F)$.
\end{definition}

Fixing all parities on a tree leaves one free bit, at its root.  Thus a
forest with $k$ edges imposes $k$ independent constraints and has exactly
$2^{N-k}$ compatible strings in $\bool^N$.  We will apply this counting
fact \emph{after fixing Bob's transcript}.  It remains valid when Bob's
choice of edges depends on Alice's input through the message.

\paragraph{The common information bound.}
Let $X\sim\Uni{\bool^N}$, let $M=(M_1,\ldots,M_m)$ be independent of
$X$ and all encoding randomness, and let public coins $R$ be independent
of $(X,M)$.  Alice's
$q$-qubit message $Q$ depends on $(X,R)$ and private randomness, but not
on $M$.  For any classical transcript $O$ obtained by Bob from $(Q,M,R)$,
data processing and the dimension bound give
\begin{equation}\label{eq:consumability-information-budget}
\begin{aligned}
 I(X:O\mid M)
 &\le I(X:O,R\mid M)
 \le I(X:Q,R\mid M)\\
 &=I(X:Q\mid R)\le q.
\end{aligned}
\end{equation}
The equality uses independence of $M$ from $(X,Q,R)$ and of $R$ from
$X$.  Constant $R$ covers private-coin protocols.  This argument also
applies to a sequential decoder: for each fixed matching list its entire
adaptive strategy is a channel.  Conditioning on that list in the
analysis does not reveal future requests to the decoder.

Take $O=((e_1,b_1),\ldots,(e_m,b_m))$, with $e_i\in M_i$.
We may require legal edges even on failure: replacing an invalid edge or
abort by a fixed legal edge and bit preserves all successful executions.
Let $S$ indicate joint success, and write
$K(M,O)=\fs(\set{e_1,\ldots,e_m})$.

\forestInformationLemma
\begin{proof}
The upper bound is \eqref{eq:consumability-information-budget}.
For the lower bound, fix $M=\mu$ and $O=o$ on success.  The parities
on a spanning forest of the selected edges restrict $X$ to at most
$2^{N-K(\mu,o)}$ strings.  Hence
$H(X\mid M=\mu,O=o,S=1)\le N-K(\mu,o)$.
On failure use the trivial bound $N$.  Averaging and charging at most
$h_2(p)$ bits to specify success gives
\[
\begin{aligned}
 H(X\mid O,M)
 &\le H(S\mid O,M)+H(X\mid O,M,S)\\
 &\le h_2(p)+N-\E{SK(M,O)}.
\end{aligned}
\]
Subtract from $H(X\mid M)=N$.
\end{proof}

Note that the selected forest may depend on whether decoding succeeds.
This is why the bound contains $\E{SK}$.  The counting argument applies
after Bob's classical output has been fixed, regardless of the quantum
computation he used to choose it.

\subsection{Simultaneous requests}
\label{subsec:transversal-forests}
\label{subsec:factorized-simultaneous}

We now show that independent random matchings force a large forest for
every coordinated edge choice.  Only their edge marginals are needed.

\begin{definition}[Balanced matching distribution]
\label{def:balanced-matching-distribution}
A distribution $\Dc$ on perfect matchings is balanced if
$\prob[M\sim\Dc]{e\in M}=1/(N-1)$ for every unordered pair $e$.
Examples are a uniform perfect matching and a uniform factor of any
one-factorization of $K_N$.
\end{definition}

\begin{definition}[Transversal forest size]
\label{def:transversal-forest}
For perfect matchings $M_1,\ldots,M_m$, set
\[
  \tfor(M_1,\ldots,M_m)
  :=\min_{e_i\in M_i}\fs(\set{e_1,\ldots,e_m}).
\]
\end{definition}

Thus $\tfor$ lower-bounds the number of independent parity constraints
in any jointly correct output.  A small value would require all selected
edges to lie in a small vertex set.

\begin{lemma}[Small forests have small support]
\label{lem:forest-support-main}
For every nonempty edge set $F$, $\abs{V(F)}\le2\fs(F)$.
Consequently, if $\tfor(M_1,\ldots,M_m)\le s$, some set $U$ of at most
$2s$ vertices contains an edge from every matching.
\end{lemma}
\begin{proof}
If the nonempty connected components have sizes $v_1,\ldots,v_t$, then
$\fs(F)=\sum_j(v_j-1)$ and $v_j\ge2$.  Summing
$v_j\le2(v_j-1)$ proves the first claim.  For the second, take the
vertices incident to a transversal witnessing $\tfor\le s$.
\end{proof}

\begin{proposition}[Transversal-forest tail for balanced matchings]
\label{prop:ext-balanced-tail}
\label{prop:random-transversal-forest-main}
For independent draws $M_1,\ldots,M_m$ from any balanced
matching distribution and every integer $1\le s\le N/2$,
\begin{equation}\label{eq:ext-tail}
 \prob{\tfor(M_1,\ldots,M_m)\le s}
 \le T_{N,m}(s)
 :=\binom N{2s}\left(\frac{\binom{2s}{2}}{N-1}\right)^m.
\end{equation}
\end{proposition}
\begin{proof}
By \cref{lem:forest-support-main}, enlarge the support of a small
transversal to a set $U$ of exactly $2s$ vertices.  For fixed $U$,
balancedness and a union bound over its pairs give
\[
 \prob{M_i\text{ contains an edge inside }U}
 \le\frac{\binom{2s}{2}}{N-1}.
\]
Independence of the matching draws gives the $m$th power for this same
$U$ to work for every matching.  Union-bound over the $\binom N{2s}$
possible sets $U$.
\end{proof}

\begin{theorem}[Quantitative simultaneous MHM lower bound]
\label{thm:ext-balanced-quant}
\label{thm:app-mhm-quantitative}
Let $X$ be uniform, and let $M_1,\ldots,M_m$ be independent
draws from a balanced matching distribution, independent of $X$ and the
encoding randomness.  Alice sends an arbitrary $q$-qubit
encoding, and Bob receives all matchings and decodes jointly.  If his
actual joint success probability is $p$, then for every integer
$1\le s\le N/2$,
\begin{equation}\label{eq:ext-quant}
 q\ge s\bigl[p-T_{N,m}(s)\bigr]-h_2(p).
\end{equation}
Independent public coins are allowed under the common convention above.
\end{theorem}
\begin{proof}
Whenever $\tfor\ge s$, every legal output has forest size at least
$s$.  The tail bound gives $\prob{\tfor<s}\le T_{N,m}(s)$, so
\[
 \E{SK}\ge s\prob{S=1,\ \tfor\ge s}
 \ge s\bigl[p-T_{N,m}(s)\bigr].
\]
Apply \cref{lem:forest-information-main}.  This subtracts the probability
of a bad matching list without assuming that it is independent of success.
\end{proof}

\begin{corollary}[Linear transversal forests for balanced matchings]
\label{cor:linear-transversal-forest-main}
\label{cor:app-delta-transversal}
Fix $\delta\in(0,1/2)$.  There is $c_\delta>0$ such that, whenever
$m\le N^{1/2-\delta}$, independent balanced matchings satisfy
$\tfor(M_1,\ldots,M_m)\ge c_\delta m$ except with probability
$\exp(-\bOm[\delta]{m\ln N})$.
\end{corollary}
\begin{proof}
Take $c_\delta=\delta/16$ and $s=\floor{c_\delta m}$.
If $s=0$, every transversal has forest size at least $1>c_\delta m$.
Otherwise $1\le s\le N/2$, and the numerical bound in
\eqref{eq:ext-tail} satisfies
\begin{align*}
 \ln T_{N,m}(s)
 &\le 2s\ln\frac{eN}{2s}+m\ln\frac{2s^2}{N-1}\\
 &\le 2c_\delta m\ln(eN)+m\ln\frac{4c_\delta^2m^2}{N}\\
 &\le\left[2c_\delta\left(1+\frac1{\ln N}\right)-2\delta\right]
          m\ln N+m\ln(4c_\delta^2)\\
 &\le-\delta m\ln N.
\end{align*}
Here $N-1\ge N/2$, $m^2\le N^{1-2\delta}$, and
$2c_\delta(1+1/\ln2)<\delta$.  Outside the event $\tfor\le s$,
integrality gives $\tfor>s$ and hence $\tfor\ge c_\delta m$.
\end{proof}

Combining the explicit estimate $T_{N,m}(s)\le\exp(-\delta m\ln N)$
with \cref{thm:ext-balanced-quant} gives a useful common consequence.
For uniform $X$, independent balanced requests, and joint success at least
$p_0$, whenever $m\le N^{1/2-\delta}$ and $\floor{c_\delta m}\ge1$,
\begin{equation}\label{eq:consumability-linear-transfer}
 q\ge\floor{c_\delta m}
          \bigl[p_0-\exp(-\delta m\ln N)\bigr]-1.
\end{equation}
This one inequality supplies both linear communication bounds below.

\linearMHMTheorem
\begin{proof}
Evaluate the protocol on uniform $X$ and independent uniform perfect
matchings, a balanced distribution.  Its joint success is at least $1-\eps$.  Write $c=c_\delta$.  For $m\ge2/c$, the floor in
\eqref{eq:consumability-linear-transfer} is at least $cm/2$.
For sufficiently large $N$, the exponential term is at most $(1-\eps)/2$;
thus $q\ge (1-\eps)cm/4-1\ge (1-\eps)cm/8$ when $m\ge8/(c(1-\eps))$.

For the remaining bounded values of $m$, at least one qubit is needed:
with no message, Bob's first edge and bit are independent of uniform $X$,
so joint success is at most $1/2$.  This also absorbs the finitely many
smaller values of $N$ by decreasing the implicit constant.
\end{proof}

The same bound applies to the random factors used in ARHM.  Here it is
useful to allow any fixed positive joint success probability, with an
explicit lower threshold on the number of requests.

\begin{corollary}[Linear simultaneous consumability for random factors]
\label{cor:ext-factor-linear}
Fix $\delta\in(0,1/2)$ and $p_0\in(0,1]$.  There are constants
$c_\delta>0$ and $m_0=m_0(\delta,p_0)$ such that, for sufficiently
large even $N$ and $m_0\le m\le N^{1/2-\delta}$, independent uniform
factors $M_{I_1},\ldots,M_{I_m}$ of any public one-factorization satisfy:
\begin{enumerate}[leftmargin=1.6em,label=\textnormal{(\roman*)}]
\item $\tfor(M_{I_1},\ldots,M_{I_m})\ge c_\delta m$ with probability
at least $1-\exp[-\bOm[\delta]{m\log N}]$;
\item every $q$-qubit encoding with joint success at least
$p_0$ under uniform $X$ and these factor draws, independent of $X$ and
the encoding randomness, has $q=\bOm[{\delta,p_0}]{m}$, even with
simultaneous revelation and unrestricted joint decoding.
\end{enumerate}
\end{corollary}
\begin{proof}
Uniform factors are balanced, so (i) is
\cref{cor:linear-transversal-forest-main}.  For (ii), choose
$m_0\ge\max\{2/c_\delta,8/(c_\delta p_0)\}$ and take $N$ large
enough that $\exp(-\delta m\ln N)\le p_0/2$ for all $m\ge m_0$.
Then \eqref{eq:consumability-linear-transfer} gives
$q\ge c_\delta p_0m/4-1\ge c_\delta p_0m/8$.
\end{proof}

\begin{remark}[Scope of the simultaneous bound]
\label{rem:ext-sim}
Factors are drawn independently with replacement, so repetitions are
allowed.  The lower bound is distributional; a fixed list repeating one
factor can be answered by measuring one phase state and repeating its
answer.  Bob may use any joint operation after seeing the entire list.
\end{remark}

\subsection{Sequential requests}
\label{app:seqmhm}

When each answer is due before the next matching is known, fresh requests
force forest growth even if Bob chooses each edge adaptively.  This
extends the linear lower bound to the square-root scale.

\begin{definition}[Sequential Hidden Matching]
\label{def:app-seqmhm}
\label{def:sequential-hidden-matching}
Fix a distribution $\Dc$ on perfect matchings.  In
$\SeqHM_{N,m}(\Dc)$, Alice receives $X\in\bool^N$ and sends one
quantum message.  Bob receives independent matchings
$M_1,\ldots,M_m\sim\Dc$ one at a time.  After receiving $M_i$ and
before receiving $M_{i+1}$, he outputs $e_i=\set{u_i,v_i}\in M_i$ and
a bit $b_i$, correct when $b_i=X_{u_i}\oplus X_{v_i}$.
He may retain arbitrary quantum memory and use adaptive operations.
The matching draws are independent of $X$ and protocol randomness;
no future matching is revealed, even through public coins, before its
round.  Outputs are legal even on failure, using the normalization above.
\end{definition}

For uniform perfect matchings this is $\SeqMHM_{N,m}$.
Throughout the sequential lemmas and their corollaries, $X$ is uniform
and requests follow the fresh-request experiment just defined.
Let $F_i=\set{e_1,\ldots,e_i}$ and $L_i=\fs(F_i)$, with $L_0=0$.
There are two tasks: lower-bound $\E{L_m}$ and account for incorrect
parity labels.

\subsubsection{Online forest growth}

\begin{lemma}[Online forest growth]
\label{lem:app-online-forest-growth}
If $\Dc$ is balanced, then for every positive-probability classical
history before $M_i$ is revealed,
\[
 \prob{L_i=L_{i-1}\mid\text{history}}
 \le\min\set*{1,\frac{\binom i2}{N-1}}.
\]
Consequently, $\E{L_m}\ge m-m(m^2-1)/(6(N-1))$.
\end{lemma}
\begin{proof}
Fix the history and let $C_1,\ldots,C_t$ be the nonempty components
of the previously selected edges.  Adding an edge fails to increase
forest size exactly when its endpoints already lie in the same component.
Thus failure to grow requires the fresh matching to meet
$B:=\bigcup_j\binom{C_j}{2}$, regardless of how Bob chooses his edge.

Let $t_j$ count prior edge occurrences in component $C_j$, including
repetitions.  Since $\sum_jt_j=i-1$ and $\abs{C_j}\le t_j+1$,
\[
 \abs B=\sum_j\binom{\abs{C_j}}2
 \le\sum_j\binom{t_j+1}2
 \le\binom{1+\sum_jt_j}2=\binom i2.
\]
Freshness makes $M_i$ independent of the entire pre-round state,
including the quantum memory, conditional on this history.
Balancedness and a union bound therefore give probability at most
$\abs B/(N-1)$ of meeting $B$.  We need not fix or measure Bob's
residual quantum state.

Finally $L_i-L_{i-1}\in\bool$, so
\[
 \E{L_m}=\sum_{i=1}^m\prob{L_i>L_{i-1}}
 \ge m-\frac1{N-1}\sum_{i=1}^m\binom i2
 =m-\frac{m(m^2-1)}{6(N-1)}.
\]
\end{proof}

\subsubsection{Sequential forest information}

Write $M=(M_1,\ldots,M_m)$ and
$O=((e_1,b_1),\ldots,(e_m,b_m))$.  Define the error vector
$Z=(Z_1,\ldots,Z_m)$ by $Z_i=b_i\oplus X_{u_i}\oplus X_{v_i}$.
Knowing $Z$ repairs all reported parity labels, so we can use the entire
forest and subtract the entropy needed to describe the errors.

\begin{lemma}[Sequential forest-information bound]
\label{lem:app-sequential-forest-information}
Every $q$-qubit message satisfies
\[
 q\ge\E{L_m}-H(Z\mid O,M_1,\ldots,M_m).
\]
In particular, $q\ge\E{L_m}-H(Z)$.
\end{lemma}
\begin{proof}
By \eqref{eq:consumability-information-budget}, $q\ge I(X:O\mid M)$,
including with public coins.  For fixed $(O,M,Z)=(o,\mu,z)$, choose
a spanning forest of the output edges and one occurrence of each forest
edge.  Its true parity is $b_i\oplus z_i$, so these labels impose
$L_m(o)$ independent constraints.  Consequently,
$H(X\mid O,M,Z)\le N-\E{L_m}$.

Since $Z$ is a function of $(X,O,M)$,
\[
\begin{aligned}
 H(X\mid O,M)
 &=H(Z\mid O,M)+H(X\mid O,M,Z)\\
 &\le H(Z\mid O,M)+N-\E{L_m}.
\end{aligned}
\]
Subtract from $H(X\mid M)=N$ and use $H(Z\mid O,M)\le H(Z)$ for
the final assertion.  Equivalently, conditioning throughout on public
coins $R$ gives the stronger subtraction $H(Z\mid O,M,R)$.
\end{proof}

\begin{restatable}[Sequential error-entropy lower bound]{theorem}{sequentialEntropyTheorem}
\label{thm:app-seqmhm}
\label{thm:sequential-hidden-matching}
For balanced $\Dc$ and the uniform-input, fresh-request experiment above,
every protocol for $\SeqHM_{N,m}(\Dc)$ satisfies
\[
 q\ge m-\frac{m(m^2-1)}{6(N-1)}-H(Z\mid O,M_1,\ldots,M_m).
\]
In particular, $q\ge m-m(m^2-1)/(6(N-1))-H(Z)$.
\end{restatable}
\begin{proof}
Combine \cref{lem:app-online-forest-growth} with
\cref{lem:app-sequential-forest-information}.
\end{proof}

For $m=o(\sqrt N)$ the forest has $(1-o(1))m$ edges in expectation.
The only remaining loss is the error entropy.  Joint success at least
$p$ gives $H(Z)\le1+(1-p)m$, hence
$q\ge pm-1-m(m^2-1)/(6(N-1))$; the details and other correctness
guarantees appear in \cref{app:consumability-extra}.
In particular, success $2/3$ and $m\le\sqrt N$ give $q\ge m/2-1$.

\paragraph{Back to streaming.}
For update-first ARHM, Alice sends the memory remaining after all
updates, and Bob uses it to answer the requests.  The sequential bound
therefore gives the required streaming lower bound.  The simultaneous
bound also applies if Bob learns all requests at once after receiving
the message, in the stated range below $\sqrt N$.

\paragraph{Open questions.}
For simultaneous requests, our linear lower bound holds when
$m\le N^{1/2-\delta}$ for fixed $\delta>0$.  Can it reach
$m=\floor{\sqrt N}$, as the sequential bound does?
The endpoint estimate in \cref{sec:ext-endpoint} is
$\bOm{\sqrt N\log\log N/\log N}$, leaving a gap to linear dependence.
Even where our lower bound is linear, sending separate phase states
uses $\bO{m\log N}$ qubits.  What is the optimal joint dependence on
$N$ and $m$?

In the next section, we use later input to rebuild a quantum sketch
for triangle counting, where the goal is to obtain more observations
for an estimator.

\section{Replenishment in a graph estimator: triangle counting}
\label{sec:triangles}

We now use replenishment to adapt Kallaugher's triangle-counting
algorithm~\cite{Kallaugher2021} to repeated-edge streams.  His algorithm
combines a quantum sketch with a classical sampling estimator.  Queries
remove directed-edge states from the quantum sketch, so running the
estimator many times requires many copies.

Suppose every edge appears exactly $R$ times in uniformly random order.
We divide the stream into blocks and use each block to rebuild the
quantum sketch and resample the classical component.  This lets us
obtain observations from successive blocks using the same memory.
Both the estimator and much of the analysis below are adaptations of
Kallaugher's work~\cite{Kallaugher2021} to repeated edges.  We reuse his
measurement, the division into quantum low weights and classical high
weights, and the main arguments for the expectation and variance of
these estimators.  The additional analysis handles rebuilding from
repeated occurrences, the bias caused by repeated closing edges, and
averaging the block estimates to recover the triangle count of the
underlying simple graph.
We only need to produce one estimate at the end of the stream.

\subsection{Model and statement}

Let $G=(V,E)$ be a simple graph with $n=\abs V$, $m=\abs E$, and
$T=\abs{\calT(G)}\ge1$.  Set
\[\Delta_E=\max_{e \in E}\{\text{number of triangles containing }e\}.\]

\begin{definition}[Exact-$R$ random-repetition stream]
For each $e\in E$, create $R$ labeled occurrences $(e,1),\ldots,(e,R)$
and uniformly permute all $Rm$ occurrences.  The algorithm sees only the
underlying edge; labels define the probability space.  Its target is the
number of triangles in the underlying simple graph.
\end{definition}

\paragraph{Space complexity.}
We use \emph{expected simultaneous space}.
Hold a legal stream $\sigma$ fixed, and let $M_t(\sigma)$
be its live memory usage, in bits and qubits, at stage $t$.  Stages include
intermediate steps within item transitions as well as initialization
and final processing.
We define
\begin{equation}\label{eq:expected-sim-space}
 S_{\mathrm{sim}}^{\mathrm{exp}}
 :=\sup_\sigma\sup_t\E[\Ac]{M_t(\sigma)}.
\end{equation}
The expectation is over internal algorithmic randomness and measurement
outcomes, not the input permutation.  The stage is fixed before this
expectation; it is not selected after observing an unusually large
sample. 

The stronger expected peak measure is
$S_{\mathrm{peak}}^{\mathrm{exp}}
 :=\sup_\sigma\E[\Ac]{\sup_t M_t(\sigma)}$.
In general, exchanging the expectation and the supremum can increase
the bound.  Our triangle theorem establishes only
\eqref{eq:expected-sim-space}; its space guarantee holds on every fixed
legal stream, although its accuracy guarantee uses random order.
ARHM and NGC keep the worst-case live-space convention.

Choose an integer block multiplicity $L\ge1$.  Put
$B=\floor{R/L}$ and $q=R-BL\in\{0,\ldots,L-1\}$.
The first $BLm$ positions form full blocks $\calB_1,\ldots,\calB_B$,
each of length $H:=Lm$; the remaining $qm$ positions form $\calB_\star$.
Let $C_{j,e}$ and $C_{\star,e}$ count copies of $e$ in those blocks.
Exactly $R$ copies of each edge occur overall:
\begin{equation}\label{eq:total-multiplicity}
 \sum_{j=1}^B C_{j,e}+C_{\star,e}=R\qquad(e\in E).
\end{equation}

An estimator \emph{lane} is one physical instance reused across blocks.
Block $\calB_1$ only builds its first state.  In block $j\ge2$, it queries
the READY state built in block $j-1$ while a separate BUILD buffer prepares
the next state.  There are $J:=B-1$ usable query blocks.  Each occurrence
in a quantum lane's query block is activated independently with probability
\begin{equation}\label{eq:rho}
 \rho:=\frac1{kL},\qquad k\in[1,m].
\end{equation}

\begin{restatable}[Random-repetition triangle counting: parameterized form]{theorem}{triangleCountingTheorem}
\label{thm:main-parameterized}
Suppose the vertex universe $[n]$, exact $m,R$, and estimates
$T/2\le T_0\le T$, $\Delta_E\le\Delta_0\le2\Delta_E$ are supplied
before processing begins (\cref{rem:triangle-known-parameters}).
Fix $\eps,\delta\in(0,1/10)$ and take
$L\ge C\log(m(R+1)/\delta)$ for a sufficiently large absolute $C$.
Assume
\begin{align}
 R&\ge CL/\eps,\label{eq:R-large-condition}\\
 C/\eps&\le k\le m.\label{eq:k-bias-condition}
\end{align}
There is a one-pass quantum streaming algorithm with
$\prob{\abs{\wh T-T}\le\eps T}\ge1-\delta$ and expected simultaneous
space, as defined in \eqref{eq:expected-sim-space}, at most
\begin{equation}\label{eq:main-space-explicit}
 \bO{\log\frac1\delta\;\polylog(n,m,R,1/\eps)
  \left[1+\frac{L^2k^2m^2}{R\eps^2T^2}
           +\frac{Lm^{3/2}\Delta_E}{R\eps^2T\sqrt k}\right]}.
\end{equation}
The budget counts qubits and classical bits.  Choosing
$L=\ceil{C\log(m(R+1)/\delta)}$, subject to
\eqref{eq:R-large-condition}, gives
\begin{equation}\label{eq:main-space-tilde}
 \wtO\!\left(1+\frac{k^2m^2}{R\eps^2T^2}
                +\frac{m^{3/2}\Delta_E}{R\eps^2T\sqrt k}\right).
\end{equation}
\end{restatable}

\begin{restatable}[Optimized constant-error form]{corollary}{optimizedTriangleCorollary}
\label{cor:optimized}
Fix constant $\eps,\delta\in(0,1/10)$ and choose
$L=\ceil{C_0\log(m(R+1)/\delta)}$ as in the theorem.
Suppose $R\ge C_{\eps,\delta}\log(mR)$, with a constant large enough
for \eqref{eq:R-large-condition}, and the supplied scales give
\[
 k_0:=\left(\frac{T_0\Delta_0}{L\sqrt m}\right)^{2/5}
 \in[C/\eps,m].
\]
Taking $k=k_0$ gives expected simultaneous space
\begin{equation}\label{eq:optimized-space}
 \bO{\polylog(n,m,R)
   \left[1+\frac{L^{6/5}m^{8/5}\Delta_E^{4/5}}{RT^{6/5}}\right]}
 =\wtO\!\left(1+\frac{m^{8/5}\Delta_E^{4/5}}{RT^{6/5}}\right).
\end{equation}
If $T=\bT{m}$ and $\Delta_E=\bO{1}$, under the same feasibility
conditions this is $\wtO(1+m^{2/5}/R)$.
\end{restatable}

\begin{remark}[Space convention and comparison with Kallaugher]
Kallaugher's single-arrival, arbitrary-order bound is
$\bOt{m^{8/5}\Delta_E^{4/5}/T^{6/5}}$ \cite{Kallaugher2021}.
Our essentially linear gain in $R$ uses the stronger repetition and
random-order promises.  This compares upper bounds in different models.

The comparison uses $S_{\mathrm{sim}}^{\mathrm{exp}}$ from
\eqref{eq:expected-sim-space}.  Kallaugher's classical sample is built
once; its expected size also controls expected peak space, up to the fixed
polylogarithmic overhead.  Here a fresh sample is built in every block.
Each has constant expected size at every fixed stage, but the maximum
over $J=\Theta(R/L)$ independent refreshes has repeated opportunities for
a rare large sample.  Thus this theorem makes no expected-peak-space
claim.  Repetition helps statistical averaging and reduces simultaneous
quantum lanes, while making this stronger classical space measure harder
to control.
\end{remark}

\subsection{Balanced blocks in the random permutation}

We first show that every full block contains every underlying edge
with high probability.  This will let us rebuild the graph state in
each block.  We also bound the number of occurrences of any edge,
which we will need when analyzing the estimator.

\begin{lemma}[Block coverage and upper multiplicity]
\label{lem:good-blocks}
For an absolute $c>0$, each full-block count has mean $L$ and satisfies
\[
 \prob{C_{j,e}=0}\le e^{-L},\qquad
 \prob{C_{j,e}>2L}\le e^{-cL}.
\]
The tail count has mean $q\le L$ and
$\prob{C_{\star,e}>2L}\le e^{-cL}$.
For $L\ge C\log(m(R+1)/\delta)$ and sufficiently large $C$, the event
\begin{equation}\label{eq:good-event}
 \calG:=\left\{
 \begin{array}{ll}
 1\le C_{j,e}\le2L&\text{for every full block and every edge},\\
 C_{\star,e}\le2L&\text{for every edge}
 \end{array}\right\}
\end{equation}
has probability at least $1-\delta/10$.
\end{lemma}
\begin{proof}
A full block samples $Lm$ positions without replacement from $Rm$
occurrences, of which $R$ represent $e$.  Thus $C_{j,e}$ is
distributed as $\operatorname{Hypergeom}(Rm,R,Lm)$, with mean $L$.  Until an $e$ has appeared, its fraction
among remaining items is at least $1/m$, so
$\prob{C_{j,e}=0}\le(1-1/m)^{Lm}\le e^{-L}$.
For a hypergeometric count $Y$ with mean $\nu>0$, the Chernoff upper
bound is $\prob{Y\ge(1+\eta)\nu}\le
[e^\eta/(1+\eta)^{1+\eta}]^\nu$ for $\eta>0$.
At $\eta=1$ it gives $\prob{C_{j,e}>2L}\le e^{-L/3}$.
The tail samples only $qm\le Lm$ positions and is stochastically
dominated by a full-block count, giving the same upper-tail estimate.
Union-bound over at most $m(B+1)\le m(R+1)$ block-edge pairs.
\end{proof}

For the estimator analysis, fix an arbitrary stream $\sigma\in\calG$.
All counts are now deterministic; expectations and variances
below are over lane randomness, conditional on this entire stream.
Different blocks need not be independent under random stream order.

\subsection{Rebuilding the quantum state}

A BUILD register has directed-edge basis states
$\ket{\overrightarrow{uv}}$ and $2H$ dummy states, initially
\begin{equation}\label{eq:dummy-initial}
 \ket{D_H}=\frac1{\sqrt{2H}}\sum_{s=1}^{2H}\ket{d_s}.
\end{equation}
The dummy labels refer to positions in the block, not to edge identities:
position $t$ has its own pair $d_{2t-1},d_{2t}$, regardless of which edge
arrives there.  All dummy states have amplitude $1/\sqrt{2H}$ initially.
At BUILD position $t$ with edge $uv$, swap
$\ket{d_{2t-1}}\leftrightarrow\ket{\overrightarrow{uv}}$ and
$\ket{d_{2t}}\leftrightarrow\ket{\overrightarrow{vu}}$.
Each dummy participates in exactly one swap during the block.  Thus these
updates require only the current edge and the position counter; the
algorithm does not test whether the edge has appeared before or store
a set of previously seen edges.  The dummy states are basis states of
one register, not separate registers.  No query acts on this register
until its block is complete.

\begin{lemma}[Repeated occurrences rebuild a uniform graph component]
\label{lem:build-state}
If the BUILD block contains every edge, its final state is
\begin{equation}\label{eq:build-state}
 \ket\Psi=\frac1{\sqrt{2H}}
   \left(\sum_{\overrightarrow e\in\overrightarrow E}
             \ket{\overrightarrow e}+\sum_{d\in D}\ket d\right),
 \qquad \abs D=2H-2m,
\end{equation}
where $\overrightarrow E$ contains both orientations of each edge.
Its squared mass in the directed-edge subspace is
\begin{equation}\label{eq:directed-mass}
 2m/(2H)=1/L.
\end{equation}
\end{lemma}
\begin{proof}
Write $a=1/\sqrt{2H}$.  When position $t$ is processed, neither of its
dummy states has been touched, so each still has amplitude $a$.
If this is the first occurrence of $uv$, both directed-edge states
have amplitude zero.  The swaps put amplitude $a$ on each orientation
and leave amplitude zero on the two dummies.

If $uv$ has appeared earlier, each directed-edge state already has
amplitude $a$.  The same swaps move that amplitude into the current
dummies and move the equal dummy amplitudes onto the directed-edge
states.  Hence the state is unchanged.  These dummies will not be used
in another BUILD update: any later occurrence has its own pair of
position labels.  In particular, repetitions neither add amplitude to
an edge nor remove its amplitude.  This distinction between first and
later occurrences is only for the analysis; the algorithm performs the
same swaps in both cases.

Consequently, every directed graph edge has amplitude $a$ once it has
appeared.  Exactly two dummy amplitudes are removed per distinct edge,
and every other dummy retains amplitude $a$.  Block coverage therefore
gives the stated state and directed-edge mass.
\end{proof}

Projecting onto the directed-edge subspace would give the exact uniform
graph state with probability $1/L$.  We keep the dummy component instead:
it reduces both the mean and the second moment of the query output.

\subsection{The replenished quantum estimator}

During the next block, the completed register is READY and receives no
swap updates.  At each activated occurrence $e=uv$, perform Kallaugher's
three-outcome measurement \cite{Kallaugher2021} with projectors
\[
 O_{e,\pm}=\frac12\sum_{w\notin\{u,v\}}
 (\ket{\overrightarrow{wu}}\pm\ket{\overrightarrow{wv}})
 (\bra{\overrightarrow{wu}}\pm\bra{\overrightarrow{wv}}),
 \qquad O_{e,0}=I-O_{e,+}-O_{e,-}.
\]
The lane returns its first nonzero outcome and stops querying for that
block; if there is none, it returns zero.  The BUILD register continues
independently regardless of that outcome.

We describe the state by the set of elements represented in the sketch,
as in the pair-sketch framework of Kallaugher, Parekh, and
Voronova~\cite{KallaugherParekhVoronovaPair}.  Here an element is either a
directed edge or a dummy label, and the quantum state is the uniform
superposition over these elements.  This set is used only for the
analysis; the algorithm does not store a classical list of its elements.
An outcome of $+1$ or $-1$ consumes the READY sketch for this block.
After an outcome of $0$, some elements are removed and the sketch can
be queried again.

Two elementary properties explain the estimator.  Suppose the sketch
currently contains $D$ elements, and the query edge is $uv$.  Let $a$
be the number of vertices $w\notin\{u,v\}$ for which both directed edges
$\overrightarrow{wu}$ and $\overrightarrow{wv}$ remain in the sketch.
\begin{enumerate}[label=(Q\arabic*)]
\item The expected query outcome, with values $-1$, $0$, and $+1$, is
$2a/D$.  If only one of the two directed edges remains, its contributions
to the probabilities of $+1$ and $-1$ are equal and cancel in the
expectation.
\item An outcome of $0$ removes all directed edges pointing into $u$
or $v$ from a third vertex.  The resulting state is the uniform
superposition over the remaining elements.  Dummy labels and the two
orientations of $uv$ itself are unchanged.  If $D'$ elements remain,
the probability of this outcome is $D'/D$.
\end{enumerate}

\subsubsection{Expected output over a sequence of queries}

The calculation below uses the same survival-probability argument as
for the pair sketch: fewer elements remain after each query that does
not consume the sketch, but their amplitudes increase upon
renormalization.  The probability of reaching a later query cancels
this increase when we compute its contribution to the expected output.

\begin{lemma}[Expected output for a fixed query schedule]
\label{lem:telescoping}
Fix a state of the form \eqref{eq:build-state}, a query block, and its
activated positions.  List the activated query edges as
$u_1v_1,u_2v_2,\ldots$.  Starting with the elements of the initial
sketch, apply the deletion rule in (Q2) in this order.  Just before
the $i$th deletion, let $a_i$ count vertices $w\notin\{u_i,v_i\}$
for which both $\overrightarrow{wu_i}$ and $\overrightarrow{wv_i}$
are still present.  If the set becomes empty, all later $a_i$ are zero.

This deterministic sequence describes the elements remaining whenever
the actual algorithm has not yet returned $+1$ or $-1$.  The final lane
output $Y\in\{-1,0,+1\}$ satisfies
\begin{equation}\label{eq:fixed-schedule-mean}
 \E{Y\mid\text{activation schedule}}=\frac1H\sum_i a_i.
\end{equation}
\end{lemma}
\begin{proof}
Let $D_i$ be the number of elements remaining after the first $i$
deletions, with $D_0=2H$.  The algorithm reaches query $i$ exactly
when every earlier query has returned $0$.  If $D_{i-1}>0$, (Q2)
therefore gives
\[
 \Prb[\text{query }i\text{ is reached}\mid\text{activation schedule}]
 =\prod_{r<i}\frac{D_r}{D_{r-1}}
 =\frac{D_{i-1}}{D_0}.
\]
Conditional on reaching this query, its expected outcome is
$2a_i/D_{i-1}$ by (Q1).  Multiplying the two expressions cancels
$D_{i-1}$: each pair still present contributes $2/D_0$ to the
unconditional expectation.  Since the lane returns its first nonzero
outcome, we may sum these contributions:
\[
 \E{Y\mid\text{activation schedule}}
 =\sum_{i:\,D_{i-1}>0}
    \frac{D_{i-1}}{D_0}\frac{2a_i}{D_{i-1}}
 =\frac1H\sum_i a_i.
\]
If no elements remain after some deletion, the probability of all
queries up to that point returning $0$ is zero.  The algorithm
therefore cannot reach a later query, and our convention $a_i=0$
accounts for these indices without dividing by zero.
\end{proof}

\subsubsection{Low weights and the quantum moments}

Each triangle has three orientations $\omega=(w;u,v)$, specified by its
center $w$ and its unordered closing edge $uv$.  All sums over $\omega$
below run over these $3T$ orientations.  For a closing occurrence $t$
in block $j$, define
\begin{equation}\label{eq:Dcup}
 D^{\cup}_{j,w;u,v}(t)
 :=\#\left\{s\in\calB_j:s<t,\quad
 \begin{array}{l}
 (e_s\ni u\text{ and }e_s\ne wu)\ \text{or}\\
 (e_s\ni v\text{ and }e_s\ne wv)
 \end{array}\right\}.
\end{equation}
These are precisely the earlier occurrences whose activation would
remove $\overrightarrow{wu}$ or $\overrightarrow{wv}$ under the
deletion rule (Q2).  A previous closing edge $uv$ counts once.
An occurrence of $wu$ does not remove $\overrightarrow{wu}$,
and similarly for $wv$, since a query leaves the orientations of its
own edge unchanged.

Define
\begin{align}
 \ell_j(w;u,v)&:=\frac1L
    \sum_{\substack{t\in\calB_j\\e_t=uv}}
        (1-\rho)^{D^{\cup}_{j,w;u,v}(t)},\label{eq:low-weight}\\
 T_j^{<}&:=\frac13\sum_\omega\ell_j(\omega).
       \label{eq:T-low-j}
\end{align}

\begin{lemma}[Quantum mean]
\label{lem:quantum-mean}
For a single READY lane in block $j$,
\begin{equation}\label{eq:quantum-mean}
 \E{Y_j}=\frac{3T_j^{<}}{kLm}.
\end{equation}
\end{lemma}
\begin{proof}
Both directed edges $\overrightarrow{wu}$ and $\overrightarrow{wv}$
are initially in the sketch by \cref{lem:build-state}.
For a freshly sampled activation schedule, the deterministic deletion
process in \cref{lem:telescoping} retains both until $t$ exactly when
none of the $D^{\cup}_{j,w;u,v}(t)$ occurrences counted above is
activated.  This has probability
$(1-\rho)^{D^{\cup}_{j,w;u,v}(t)}$; activation at $t$ independently
has probability $\rho$.
Averaging \cref{lem:telescoping} over these schedules gives
\[
 \E{Y_j}=\frac\rho H\sum_\omega
   \sum_{\substack{t\in\calB_j\\e_t=uv}}
       (1-\rho)^{D^{\cup}_{j,w;u,v}(t)}
 =\frac\rho H\,3LT_j^{<}=\frac{3T_j^{<}}{kLm}.
\]
Here the randomness is only in the independently sampled activation
bits.  The probability that earlier measurements allow the sketch to
reach a query has already been accounted for in
\cref{lem:telescoping}; we do not condition the activation schedule
on those measurement outcomes.
\end{proof}

\begin{lemma}[Quantum second moment]
\label{lem:quantum-second-moment}
For the same lane,
\begin{equation}\label{eq:quantum-second-moment}
 \E{Y_j^2}\le1/L,\qquad \Var{Y_j}\le1/L.
\end{equation}
\end{lemma}
\begin{proof}
$Y_j^2$ indicates that some query returned a nonzero outcome.  Nonzero
projectors act only on directed-edge states, and zero projections never
move dummy amplitude into that subspace.  The probability of ever
obtaining a nonzero outcome is therefore at most its initial mass $1/L$.
\end{proof}

\subsection{The classical complementary estimator}

\subsubsection{Exact and computable high weights}

The exact complement to $\ell_j$ is
\begin{equation}\label{eq:true-high}
 h_j(w;u,v):=\frac1L\sum_{\substack{t\in\calB_j\\e_t=uv}}
       \left[1-(1-\rho)^{D^{\cup}_{j,w;u,v}(t)}\right].
\end{equation}
Thus
\begin{align}
 \ell_j(w;u,v)+h_j(w;u,v)&=C_{j,uv}/L,
       \label{eq:low-high-orientation}\\
 T_j^{>}&:=\frac13\sum_\omega h_j(\omega).
       \label{eq:T-high-j}
\end{align}
The union count in $h_j$ is inconvenient to maintain.  Instead, one
counter per sampled directed edge records
$D_{j,wu}(t)$, the earlier occurrences incident to its head $u$, excluding
$wu$ itself.  Define $D_{j,wv}(t)$ analogously, and let $N_{j,uv}(t)$
count earlier copies of $uv$.  The two counter sets intersect exactly
in those copies, so
\begin{equation}\label{eq:double-count-identity}
 D_{j,wu}(t)+D_{j,wv}(t)
 =D^{\cup}_{j,w;u,v}(t)+N_{j,uv}(t).
\end{equation}
The counters therefore give a slightly larger, computable weight:
\begin{align}
 \wt h_j(w;u,v)&:=\frac1L\sum_{\substack{t\in\calB_j\\e_t=uv}}
     \left[1-(1-\rho)^{D_{j,wu}(t)+D_{j,wv}(t)}\right],
          \label{eq:tilde-high}\\
 \wt T_j^{>}&:=\frac13\sum_\omega\wt h_j(\omega).
          \label{eq:tilde-T-high}
\end{align}
The next two bounds, proved in \cref{app:triangle-classical-proofs},
control this approximation and its sampling variance.

\begin{lemma}[Bias from double-counting repeated closing edges]
\label{lem:high-bias}
For every orientation in a good block,
\begin{align}
 0\le\wt h_j(\omega)-h_j(\omega)&<2/k,
       \label{eq:orientation-bias}\\
 0\le\wt T_j^{>}-T_j^{>}&<2T/k.
       \label{eq:T-high-bias}
\end{align}
\end{lemma}

\begin{lemma}[Two bounds on the computable high weight]
\label{lem:high-weight-bounds}
If $d_x$ is the degree of $x$ in the underlying graph, then on $\calG$
\begin{equation}\label{eq:high-bounds}
 0\le\wt h_j(w;u,v)\le2,\qquad
 \wt h_j(w;u,v)\le\frac{4(d_u+d_v)}k.
\end{equation}
\end{lemma}

\subsubsection{Sampling and querying a classical lane}

During BUILD block $j-1$, sample a center indicator $G_w$ and a
directed-edge indicator $H_{wu}$ with probabilities
\begin{equation}\label{eq:pV-pE}
 p_V:=1/\sqrt{km},\qquad p_E:=\sqrt{k/m}.
\end{equation}
For now take these indicators mutually independent.  Store
$\overrightarrow{wu}$ precisely when it appears and $G_wH_{wu}=1$.
Decisions are keyed by underlying vertices and directed edges, so copies
reuse the same decision and are stored only once.  On $\calG$, every edge
appears in the BUILD block, giving the intended sample of $G$.

At the start of QUERY block $j$, reset all READY counters to zero and set
$X_j=0$.  When $uv$ arrives, first scan the sample for centers $w$ with
both $\overrightarrow{wu},\overrightarrow{wv}$ stored, adding
\[
 \frac1L\left[1-(1-\rho)^{D[\overrightarrow{wu}]+D[\overrightarrow{wv}]}\right]
\]
to $X_j$ for each such center.  Then increment each stored edge's counter
if the occurrence meets its head and is not that edge itself.  This
order of operations ensures that counters include only earlier occurrences.
Consequently,
\begin{equation}\label{eq:X-sum}
 X_j=\sum_\omega\wt h_j(\omega)I_\omega,
 \qquad I_{(w;u,v)}:=G_wH_{wu}H_{wv},
\end{equation}
with inclusion probability
\begin{equation}\label{eq:alpha}
 \alpha:=\E{I_\omega}=p_Vp_E^2=\sqrt{k}/m^{3/2}.
\end{equation}

The moments needed for averaging are as follows; their proofs are in
\cref{app:triangle-classical-proofs}.
\begin{lemma}[Classical expectation]
\label{lem:classical-mean}
For a fixed good stream,
\begin{equation}\label{eq:classical-mean}
 \E{X_j}=3\alpha\wt T_j^{>}.
\end{equation}
\end{lemma}
\begin{lemma}[Classical variance]
\label{lem:classical-var}
For a fixed good stream,
\begin{equation}\label{eq:classical-var}
 \Var{X_j}\le14\alpha\Delta_E\sum_\omega\wt h_j(\omega)
             =42\alpha\Delta_E\wt T_j^{>}.
\end{equation}
\end{lemma}

\paragraph{Memory of a lane.}
Each directed edge is included with probability $p_Vp_E=1/m$, so a
buffer holds at most two edges in expectation.  This holds at every fixed
stage of \emph{every} stream: even without block coverage, its distinct
stored edges form a subset of the $2m$ possible directed edges.
Two buffers, counters, and an accumulator use only polylogarithmic
expected space.  Pairwise-independent center hashes and four-wise-independent
directed-edge hashes suffice; fresh seeds are used for each block and lane.
\Cref{app:triangle-implementation} gives the seeds, dyadic probabilities,
and arithmetic precision, including their error budget.

\subsection{The two targets reconstruct the triangle count}

Rebuilding alone does not ensure that successive observations estimate
the same quantity.  Each block has its own edge multiplicities and hence
its own hybrid target.  The following decomposition uses the exact total
multiplicity to recover the triangle count after averaging these targets.

\begin{lemma}[Exact blockwise decomposition]
\label{lem:block-decomp}
For every full query block,
\begin{equation}\label{eq:block-decomp}
 T_j^{<}+T_j^{>}=\frac1{3L}\sum_{e\in E}C_{j,e}\tau_e.
\end{equation}
\end{lemma}
\begin{proof}
Sum \eqref{eq:low-high-orientation} over the three orientations of each
triangle and divide by three.  An edge $e$ is the closing edge in
exactly $\tau_e$ such orientations, giving the displayed sum.
\end{proof}

Write $\overline T^{<}$, $\overline T^{>}$, and
$\overline{\wt T}^{>}$ for the averages of the corresponding targets
over $j=2,\ldots,B$.

\begin{lemma}[Warm-up/tail loss]
\label{lem:warmup-tail}
On $\calG$, if $R\ge4L$, then
\begin{equation}\label{eq:warmup-tail}
 \abs{\overline T^{<}+\overline T^{>}-T}\le\frac{8L}R T.
\end{equation}
\end{lemma}
\begin{proof}
Use \eqref{eq:total-multiplicity}, \cref{lem:block-decomp},
$LJ=R-q-L$, and $\sum_e\tau_e=3T$ to obtain
\[
 \overline T^{<}+\overline T^{>}-T
 =\frac1{3(R-q-L)}\sum_e(q+L-C_{1,e}-C_{\star,e})\tau_e.
\]
On $\calG$, each coefficient in parentheses has absolute value at most
$4L$.  Also $R-q-L\ge R-2L\ge R/2$.  The claimed bound follows.
\end{proof}

\begin{lemma}[Structural bias of the computable hybrid target]
\label{lem:structural-bias}
On $\calG$, if $R\ge4L$, then
\begin{equation}\label{eq:structural-bias}
 \abs{\overline T^{<}+\overline{\wt T}^{>}-T}
 \le\left(\frac{8L}R+\frac2k\right)T.
\end{equation}
\end{lemma}
\begin{proof}
Average \eqref{eq:T-high-bias} across query blocks and combine with
\cref{lem:warmup-tail}.
\end{proof}

\subsection{Reusing lanes across blocks}

At a block boundary, discard the old READY buffer, promote BUILD to
READY, and initialize a new BUILD buffer.  Each lane keeps only a running
sum of past outputs.  For our fixed stream, all block samples use fresh
coins or seeds and fresh states.  They are therefore independent across
blocks, lanes, and amplification groups, even though a block supplies both
one sample's queries and the next sample's build.

\paragraph{Quantum averaging.}
With $P_Q$ lanes, set
\begin{equation}\label{eq:quantum-average-estimator}
 \wh T^{<}:=\frac{kLm}{3P_QJ}
                 \sum_{p=1}^{P_Q}\sum_{j=2}^B Y_{p,j}.
\end{equation}
By \cref{lem:quantum-mean,lem:quantum-second-moment},
\begin{align}
 \E{\wh T^{<}}&=\overline T^{<}:=\frac1J\sum_{j=2}^B T_j^{<},
       \label{eq:quantum-estimator-mean}\\
 \Var{\wh T^{<}}&\le\frac{k^2Lm^2}{9P_QJ}.
       \label{eq:quantum-estimator-var}
\end{align}
For $R\ge4L$, $J=B-1\ge R/(2L)$.  Chebyshev's inequality therefore
gives additive error $\eps T/4$ with constant success when
\begin{equation}\label{eq:PQ}
 P_Q=\bT{1+\frac{k^2L^2m^2}{R\eps^2T^2}},
\end{equation}
with a sufficiently large implicit constant.

\paragraph{Classical averaging.}
The weight bound $\wt h_j\le2$ gives $\wt T_j^{>}\le2T$.
Thus \cref{lem:classical-var} implies, for an absolute $C_1$,
\begin{equation}\label{eq:Z-classical-var}
 \Var{\frac{X_{p,j}}{3\alpha}}
 \le C_1\frac{T\Delta_E}{\alpha}
 =C_1\frac{m^{3/2}T\Delta_E}{\sqrt k}.
\end{equation}
With $P_C$ lanes, set
\begin{equation}\label{eq:classical-average}
 \wh T^{>}:=\frac1{3\alpha P_CJ}
                  \sum_{p=1}^{P_C}\sum_{j=2}^B X_{p,j}.
\end{equation}
Independence and \cref{lem:classical-mean} give
\begin{align}
 \E{\wh T^{>}}&=\overline{\wt T}^{>}:=\frac1J\sum_{j=2}^B\wt T_j^{>},
       \label{eq:classical-average-mean}\\
 \Var{\wh T^{>}}&\le C_1\frac{m^{3/2}T\Delta_E}{P_CJ\sqrt k}.
       \label{eq:classical-average-var}
\end{align}
Again using $J\ge R/(2L)$, it suffices to take
\begin{equation}\label{eq:PC}
 P_C=\bT{1+\frac{Lm^{3/2}\Delta_E}{R\eps^2T\sqrt k}}
\end{equation}
for additive error $\eps T/4$ with constant success.
Each lane now contributes $J$ observations, one per query block, while
storing only its current buffers and a running sum.  This accounts for
the factor $J$ in the variance bounds and reduces the number of lanes
that we need to keep at once.

\subsection{Proof of the main theorem}

\begin{proof}[Proof of \cref{thm:main-parameterized}]
Use $T_0$ and $\Delta_0$ in the lane counts; these substitutions are
conservative and change the counts by only constant factors.
Choose the absolute constants so that \cref{lem:good-blocks} gives
$\prob{\calG}\ge1-\delta/10$ and
\begin{equation}\label{eq:structural-small}
 8L/R+2/k\le\eps/4.
\end{equation}
For every fixed good stream, \cref{lem:structural-bias} then gives
\begin{equation}\label{eq:hybrid-target-close}
 \abs{\overline T^{<}+\overline{\wt T}^{>}-T}\le\eps T/4.
\end{equation}

Take the constants in \cref{eq:PQ,eq:PC} large enough that each
estimator has additive error at most $\eps T/4$ with probability at
least $9/10$.  Run $\bO{\log(1/\delta)}$ independent groups and take
separate medians of their quantum and classical estimates.  Conditional
on the fixed stream, both medians meet their error bounds with probability
at least $1-9\delta/10$.

The output is always the sum of these medians.  On this success event,
the two statistical errors and \eqref{eq:hybrid-target-close} total at
most $3\eps T/4$.  The implementation in
\cref{app:triangle-implementation} adds less than $\eps T/4$.
Averaging the conditional guarantee over good streams proves total
success at least $1-\delta$.  The algorithm need not recognize $\calG$
or test whether its estimates succeeded.

For space, a quantum lane holds two registers of dimension
$\bO{n^2+Lm}$, hence $\bO{\log(n^2+Lm)}$ qubits.  In our
unrestricted-operation model, initialization, swaps, and measurements
act on these registers with logarithmic classical control and no stored
adjacency table.  A classical lane holds two buffers, each with expected
size at most two at every fixed stage of every stream, plus counters,
hash seeds, and polylogarithmic-precision arithmetic.  This space argument
does not condition on $\calG$.  Processing buffers one at a time and
discarding the old READY buffer before allocation bounds temporary
workspace as well.  Running sums replace stored past block outputs.

Multiplying these per-lane bounds by $P_Q+P_C$ and the number of groups
proves \eqref{eq:main-space-explicit} under the expected simultaneous-space
convention.  The logarithmic choice of $L$ gives
\eqref{eq:main-space-tilde}.
\end{proof}

\begin{proof}[Proof of \cref{cor:optimized}]
Apart from common factors and the additive constant, the space is
\[
 \frac1R\left(\frac{L^2k^2m^2}{T^2}
                +\frac{Lm^{3/2}\Delta_E}{T\sqrt k}\right).
\]
Balancing gives $k^{5/2}\asymp T\Delta_E/(L\sqrt m)$.
The supplied estimates make $k_0$ a constant-factor approximation to
this value.  Substitution gives
$L^{6/5}m^{8/5}\Delta_E^{4/5}/(RT^{6/5})$.
For fixed error parameters $L=\bT{\log(mR)}$ is absorbed into
$\wt O$; if $T=\bT{m}$ and $\Delta_E=\bO{1}$, the remaining
polynomial term is $m^{2/5}/R$.
\end{proof}

\paragraph{Open questions.}
What are the optimal classical and quantum space bounds for triangle
counting under the same exact-$R$ random-order promise and a common
space measure?  In particular, can one prove a quantum--classical
separation in this model?  More broadly, ARHM establishes a space
separation for an artificial problem, while triangle counting gives an
improvement over an existing algorithm.  It remains open to establish
a comparable separation between random and adversarial order for a
natural problem using replenishment.  For the algorithm here, it would
also be useful to relax the assumptions of equal edge multiplicities
and supplied parameter estimates.

We next turn to Noisy Gap Cycle, for which we prove a quantum lower
bound in the model where each edge appears once.

\section{Robust quantum lower bounds: Noisy Gap Cycle}
\label{sec:ngc}

In this section we prove a linear quantum space lower bound for Noisy
Gap Cycle (NGC) in random-order streams.  Every edge appears once,
and the algorithm outputs its decision after the pass.  We follow the
robust communication approach of Assadi and
Sundaresan~\cite{AssadiSundaresan2023}, using their graph construction,
neighboring hybrids, and clean-coordinate reduction.

The new step is a communication lower bound for Block Hidden XOR with
quantum messages.  We prove it using matrix-valued Fourier analysis
and the hypercontractive inequality of Ben-Aroya, Regev, and
de Wolf~\cite{BenAroyaRegevDeWolf2008}.  We give this bound here and
the graph reductions in \cref{app:ngc-full-proof}.

\subsection{Noisy Gap Cycle and its streaming consequences}
\label{sec:ngc-main}
\begin{definition}[Noisy Gap Cycle]
\label{def:ngc}
Let $n$ and $k\ge3$ be integers with $4k\mid n$.
The robust one-way communication problem $\NGC_{n,k}$ asks Alice and Bob
to distinguish
the disjoint union of $n/(2k)$ disjoint $k$-cycles from the disjoint union of
$n/(4k)$ disjoint $2k$-cycles.  Both cases additionally contain
$n/(2k)$ disjoint paths of length $k-1$, disjoint from the cycles.
Each edge of the fixed underlying graph is independently assigned to
Alice or Bob with probability $1/2$.  Each player receives only their
assigned edge set.  Alice sends one message (a quantum state in the
quantum model) to Bob, who must output which of the two cases holds.
\end{definition}

Here \emph{robust} means that the error guarantee holds for every fixed
promise graph, averaging over its independent fair edge-ownership coins
and protocol randomness.  In the corresponding random-order streaming
problem, the algorithm sees a uniform permutation of all edges and
outputs the decision bit at the end of the stream.  Its error guarantee
holds for every fixed promise graph, averaging over the edge permutation
and algorithm randomness.

\begin{restatable}[Quantum robust Noisy Gap Cycle]{theorem}{robustNGCTheorem}
\label{thm:ngc-robust}
For all sufficiently large admissible $n,k$ with $k\ge1600\ln n$,
every robust one-way quantum protocol for $\NGC_{n,k}$
with error at most $1/3$ sends $\bOm{n}$ qubits.
\end{restatable}

\begin{restatable}[Quantum random-order Noisy Gap Cycle]{theorem}{randomOrderNGCTheorem}
\label{thm:ngc-rand}
Under the same conditions, every one-pass quantum streaming
algorithm solving $\NGC_{n,k}$ with error at most $1/3$ in
uniformly random edge order uses $\bOm{n}$ space.
\end{restatable}

\begin{corollary}[Quantum random-order graph consequences]
\label{cor:ngc-applications}
The algorithms are one-pass random-order algorithms with end-of-stream
output, and all success guarantees are at least $2/3$ for every input.
The approximation impossibilities exclude space bounds uniform in the
accuracy parameter and graph size; the contradiction uses linked sequences
with $1/\eps=\bT{\log n}$ (and $W/\eps=\bT{\log n}$ for MST),
not a linear lower bound for each fixed accuracy.
The one-shot classical reductions of
Assadi--Sundaresan~\cite{AssadiSundaresan2023} transfer the NGC
lower bound to quantum random-order streaming with the parameters of
their framework.  In particular:
\begin{enumerate}[leftmargin=1.7em]
\item Estimating the number of connected components within additive
  error $\eps n$ for every $\eps>0$ cannot be done using
  $2^{\lO{1/\eps}}n^{\lO{1}}$ space.
\item The same statement holds for maximum matching size and for
  maximum independent set size on planar graphs of maximum degree two.
\item Approximating MST weight within $1+\eps$ on connected graphs
  with integer weights in $[W]$, $W\ge2$, cannot be done, for every $\eps>0$,
  in $2^{\lO{W/\eps}}n^{\lO{1}}$ space.
\item For the sampling problem specified in
  \cref{cor:qro-consequences-main}, generating a sufficiently large
  constant number of length-$\ell$ random walks requires
  $2^{\bOm{\sqrt{\ell}}}$ space when $\ell=\bT{\log^2 n}$.
  Generating $c_2 4^\ell$ such walks, for a sufficiently large absolute
  constant $c_2$, requires $2^{\bOm{\ell}}$ space when
  $\ell=\bT{\log n}$.
\end{enumerate}
\end{corollary}

\subsection{A quantum communication lower bound for Block Hidden XOR}
\label{sec:ngc-bhx}

The communication problem we need is the following.  Alice has
independent bit strings, and Bob has one uniformly random index in
each string.  His goal is to predict the XOR of the selected bits.
We allow different string lengths because the graph reduction leaves
a different number of usable coordinates in each block.

\blockHiddenXORDefinition

In the classical proof, one conditions the input distribution on Alice's
message.  For a quantum encoding $X\mapsto\rho_X$, we instead analyze
the matrix Fourier coefficients of the encoding.  Put $d=2^q$ (padding the message space if necessary), and set
\[
 \mathcal I=\bigsqcup_{i=1}^r(\set{i}\times[w_i]),\qquad
 \chi_S(X)=(-1)^{\sum_{(i,j)\in S}X^i_j},\qquad
 \wh{\rho}(S)=\E[X]{\chi_S(X)\rho_X}.
\]
We use the ordinary trace norm $\trnorm{A}=\Tr\abs A$.
The following density-matrix form of hypercontractivity is due to
Ben-Aroya, Regev, and de Wolf~\cite{BenAroyaRegevDeWolf2008}; it is also
used by Shi, Wu, and Yu in their quantum Hidden Hypermatching lower
bound~\cite[Corollary~2.3]{ShiWuYu2015}.

\begin{lemma}[Trace-norm matrix hypercontractivity {\cite[proof of Lemma~6, full version]{BenAroyaRegevDeWolf2008}}]
\label{lem:qbhx-matrix-hc}
\label{lem:app-trace-matrix-hc}
Let $d=2^q$.  For every $0\le\lambda\le1$,
\begin{equation}
\label{eq:qbhx-matrix-hc}
    \sum_{S\subseteq\mathcal {I}}
    \lambda^{\abs{S}}
    \trnorm{\wh{\rho}(S)}^2
    \le
    2^{\frac{2\lambda}{1+\lambda}q}.
\end{equation}
\end{lemma}

We retain the bound $2^{2q(1-1/p)}$ from the last displayed estimate in
that proof, with $p=1+\lambda$, before its simplification to
$2^{2\lambda q}$.

For Bob's tuple $J$, the set
$S_J=\set{(1,J_1),\ldots,(r,J_r)}$ contains one coordinate from each
block and has size $r$, with $\chi_{S_J}(X)=(-1)^{Z_J(X)}$.
The optimal probability of distinguishing the two parity-conditioned
states links this Fourier coefficient to prediction success.

\begin{lemma}[Parity-conditioned states]
\label{lem:qbhx-parity-conditioned}
\label{lem:app-parity-conditioned}
For fixed $J$ and $b\in\bool$ let
\[
    \rho_{J,b}
    :=
    \E{\rho_X\mid Z_J(X)=b}.
\]
Then
\[
    \rho_{J,0}-\rho_{J,1}
    =
    2\wh{\rho}(S_J).
\]
Consequently, the optimal probability of predicting $Z_J(X)$ from the
message is
\[
    \frac{1}{2}
    +
    \frac{1}{2}\trnorm{\wh{\rho}(S_J)}.
\]
\end{lemma}

\begin{proof}
The parity is balanced, so
$\rho_{J,0}-\rho_{J,1}=2\E[X]{(-1)^{Z_J(X)}\rho_X}
=2\wh{\rho}(S_J)$.  The optimal probability of distinguishing two equiprobable states is
$1/2+\trnorm{\rho_{J,0}-\rho_{J,1}}/4$.
\end{proof}

Thus success forces Fourier mass on the degree-$r$ sets $S_J$, while
hypercontractivity bounds that mass in terms of message dimension.

\begin{theorem}[Unequal-block quantum Hidden XOR]
\label{thm:qbhx-master}
\label{thm:app-unequal-bhx}
Let $r\ge1$ be an integer, let all block widths be positive integers, and let
$0\le\gamma\le1/2$.  Suppose a one-way protocol, with Alice sending a state of
dimension at most $2^q$ to Bob, for
unequal-block Hidden XOR succeeds with probability at least
$1/2+\gamma$, averaged over the product input distribution in
\cref{def:ubh-main} and all protocol randomness.  Let
\[
    W_{\mathrm{geo}}
    :=
    \left(\prod_{i=1}^r w_i\right)^{1/r}
\]
and
\[
    A:=W_{\mathrm{geo}}(2\gamma)^{2/r}.
\]
For every $0<\lambda\le1$ with $A\lambda>1$,
\[
    q
    \ge
    \frac{r(1+\lambda)}{2\lambda}
    \log_2(A\lambda).
\]
In particular, if $A\ge e$, then
\[
    q
    \ge
    \frac{rW_{\mathrm{geo}}}{2e\ln2}
    (2\gamma)^{2/r}.
\]
\end{theorem}

\begin{proof}
Fix independent public coins to a seed with at least the average success;
the product input law is unchanged.  Include private encoding randomness
in $\rho_X$.  By \cref{lem:qbhx-parity-conditioned}, success
$1/2+\gamma$ implies
$\E[J]{\trnorm{\wh{\rho}(S_J)}}\ge2\gamma$.
Jensen's inequality and the $\prod_i w_i=W_{\mathrm{geo}}^r$
equiprobable choices of $J$ give
\begin{equation}
\label{eq:qbhx-lower-fourier}
 \sum_J\trnorm{\wh{\rho}(S_J)}^2
 =W_{\mathrm{geo}}^r\E[J]{\trnorm{\wh{\rho}(S_J)}^2}
 \ge4\gamma^2W_{\mathrm{geo}}^r.
\end{equation}
Since all these coefficients have degree $r$, for $0<\lambda\le1$,
\begin{equation}
\label{eq:qbhx-master-ineq}
 (A\lambda)^r
 =4\gamma^2W_{\mathrm{geo}}^r\lambda^r
 \le\sum_S\lambda^{\abs{S}}\trnorm{\wh{\rho}(S)}^2
 \le2^{2\lambda q/(1+\lambda)}.
\end{equation}
For $\gamma=0$, we have $A=0$, so neither conclusion is asserted.
Otherwise, taking base-two logarithms when $A\lambda>1$ proves the first
bound.  If $A\ge e$, choose
$\lambda=e/A$ to obtain
$q\ge r(1+\lambda)/(2\lambda\ln2)\ge rA/(2e\ln2)$.
\end{proof}

We will use the theorem with block widths $\bT{w}$, at least $\ln w$
blocks, and advantage $\bT{1/w}$.

\blockHiddenXORTheorem

\begin{proof}
Apply \cref{thm:qbhx-master} with $\gamma=1/(8w)$, so that
$A=W_{\mathrm{geo}}(1/(4w))^{2/r}$.  For $w\ge4$,
\[
 (1/(4w))^{2/r}=(4w)^{-2/r}
 \ge\exp\left(-\frac{2\ln(4w)}{\ln w}\right)\ge e^{-4}.
\]
Also $W_{\mathrm{geo}}\ge w/100$, so $A\ge w/(100e^4)\ge e$
for sufficiently large $w$.  The second bound in
\cref{thm:qbhx-master} now gives $q=\bOm{wr}$.
\end{proof}

\subsection{From Block Hidden XOR to random-order NGC}
\label{sec:ngc-chain}
\begin{definition}[Distributional Hidden XOR]
\label{def:dhx}
In $\DHX_{w,t}$, Alice and Bob receive complementary shares, obtained by independently
assigning each edge to one party with probability $1/2$, of a depth-$(3t+1)$
Multi-Block graph generated from independent uniform Boolean
strings and permutations.  Bob must recover the designated hidden
path parity $z_G(1)$; see \cref{def:multi-block,def:DHX-main}
for the full construction and input distribution.
\end{definition}

\begin{theorem}[Quantum Distributional Hidden XOR]
\label{thm:dhx-quantum}
For sufficiently large $w$ and $t\ge512\ln w$, any one-way quantum
protocol for $\DHX_{w,t}$ with success at least
$1/2+1/(6w)$ communicates $\bOm{wt}$ qubits.
\end{theorem}

\begin{proposition}[Robust-to-random-order transfer]
\label{prop:robust-transfer}
An $S$-space one-pass quantum random-edge-order streaming
algorithm for a graph problem, with all output produced after the pass,
induces a robust one-way quantum protocol using a message of
dimension at most $2^S$, provided continuation is determined by the complete
live state and public parameters together with Bob's remaining input.
In particular, any needed position or control information must be counted
in the state or locally reconstructible.  Error for each fixed graph is
preserved after averaging over the robust partition and local orders.
\end{proposition}

Here is how the bounds fit together; the details are in
\cref{app:ngc-full-proof}.
\begin{enumerate}[leftmargin=1.7em]
\item For $n=4km$ and $k=3t+1$, neighboring hybrids isolate one NGC
  path parity.  A local embedding yields $\DHX_{m+1,t}$ with advantage
  at least $1/(6(m+1))$.
\item In DHX, a favorable edge partition leaves $r\ge t/512$ usable
  blocks of widths $w_i\ge w/100$.  Fixing specified classical generation
  variables preserves advantage at least $1/(8w)$.  The residual-law and
  local-simulation arguments establish independent BHX inputs: Alice has
  the strings, and Bob can simulate his share from the indices.
\item The BHX theorem gives $\bOm{wt}=\bOm{n}$ communication.
  The appendix treats bounded $w$ separately and pads the construction
  to cover all admissible $k$.
\item A fair edge partition followed by independent uniform local orders
  produces an exactly uniform global order.  Transferring the complete
  live memory once proves the streaming lower bound.
\end{enumerate}

\begin{remark}[A limitation of the stated consequences]
The current reduction does not automatically give a rank lower bound
under uniformly random \emph{individual matrix-entry} order: one
graph edge produces a coupled pair of nonzero Tutte-matrix entries,
and random graph-edge order yields random order of pairs rather than
necessarily random order of the individual entries.
\end{remark}

\paragraph{Open question.}
Can repeated edges reduce the quantum space required for NGC?
Our lower bound concerns a single arrival of each edge.  Under
repetition, a random partition of stream items can expose copies of
the same edge to both parties, changing the information available in
the communication problem.  Thus the present argument does not
establish a lower bound for repeated-edge streams.  Determining whether
repeated input permits a smaller quantum algorithm would connect this
obstruction to the replenishment effects studied above.

\section*{AI Disclosure}
All the problems and key ideas were developed by the author. AI tools (GPT-6 Astra) were used to prototype the proofs, search for the right parameter regimes, or shorter/cleaner proofs, as well as to get a better understanding of \cite{AssadiSundaresan2023}, \cite{GilboaJainMcClean2025}, and \cite{Kallaugher2021} needed for the proof. All the results in this paper were either developed without AI tools or manually verified by the author. The author is responsible for the correctness of all results in the paper.

\section*{Acknowledgments}

This research was supported in part by the French PEPR integrated projects EPIQ (ANR-22-PETQ-0007) and HQI (ANR-22-PNCQ-0002).

\bibliographystyle{plain}
\bibliography{refs}

\appendix
\section{Additional consumability variants and endpoint estimates}
\label{app:consumability-extra}

We give three consequences of the sequential consumability theorem,
corresponding to different requirements on the number of correct answers.
The joint-success bound is the one used for ARHM; the other two allow
errors on a fraction of the requests.

\subsection{Consequences for different correctness guarantees}

The sequential theorem turns each correctness guarantee into an upper
bound on the entropy of the error vector $Z$.  Throughout, use the
experiment of \cref{thm:sequential-hidden-matching}: $N\ge2$ is even,
$m\ge1$, $X$ is uniform, and requests are fresh independent draws from
a balanced matching distribution, independent of $X$ and protocol coins.
Each answer precedes the next request, including any revelation through
public coins.  The message has dimension at most $2^q$, counting
classical information, and outputs are legal even on failure as in
\cref{def:sequential-hidden-matching}.

The common bound is
\[
 q\ge m-\frac{m(m^2-1)}{6(N-1)}-H(Z).
\]
We apply it without conditioning the request distribution or the
forest-growth estimate on success.

\begin{corollary}[Constant average advantage]
\label{cor:app-seqmhm-average}
Let $0<\eps\le1/2$.  If
$m^{-1}\sum_{i=1}^m\prob{b_i=X_{u_i}\oplus X_{v_i}}\ge1/2+\eps$,
then
\[
 q\ge m\bigl(1-h_2(1/2-\eps)\bigr)-\frac{m(m^2-1)}{6(N-1)}.
\]
For fixed $\eps>0$, some $c_\eps>0$ therefore gives
$q=\bOm[\eps]{m}$ whenever $m\le c_\eps\sqrt N$.
\end{corollary}
\begin{proof}
Since $m^{-1}\sum_i\prob{Z_i=1}\le1/2-\eps$, subadditivity,
concavity, and monotonicity of binary entropy on $[0,1/2]$ give
\[
 H(Z)\le\sum_i h_2(\prob{Z_i=1})
 \le m h_2\left(\frac1m\sum_i\prob{Z_i=1}\right)\le m h_2(1/2-\eps).
\]
Substitute into the common bound.  For the linear consequence, if
$m\le\sqrt{\frac32\bigl(1-h_2(1/2-\eps)\bigr)}\,\sqrt N$, then
\[
 \frac{m^2-1}{6(N-1)}\le
 \frac{\frac32\bigl(1-h_2(1/2-\eps)\bigr)N}{6(N-1)}\le
 \frac{1-h_2(1/2-\eps)}2.
\]
Thus $q\ge\frac m2\bigl(1-h_2(1/2-\eps)\bigr)$, including bounded $m$.
\end{proof}

Joint success gives the version used in ARHM, extending through
$m=\floor{\sqrt N}$ at success probability $2/3$.

\begin{corollary}[Joint success]
\label{cor:app-seqmhm-joint}
If all $m$ answers are correct with probability at least $p\in[0,1]$,
then
\[
 q\ge pm-1-\frac{m(m^2-1)}{6(N-1)}.
\]
In particular, $p=2/3$ and $m=\floor{\sqrt N}$ give
$q=\bOm{\sqrt N}$.
\end{corollary}
\begin{proof}
Let $S=\ind{}[Z=0^m]$.  On $S=1$, $Z$ is fixed; on $S=0$, use
$H(Z\mid S=0)\le m$.  Therefore
\[
 H(Z)\le H(S)+H(Z\mid S)\le1+(1-p)m.
\]
Substitution proves the bound.  If $m\le\sqrt N$, the forest-growth
loss is at most $m/6$, so $p=2/3$ gives $q\ge m/2-1$.
\end{proof}

Requiring only a constant fraction of correct answers still forces a
linear message once the number of requests is sufficiently large.

\begin{corollary}[Correctness on a constant fraction]
\label{cor:app-seqmhm-fraction}
Fix $1/2<\alpha\le1$ and $p\in[0,1]$.  If Bob answers at least
$\alpha m$ requests correctly with probability at least $p$, then
\[
 q\ge pm\bigl(1-h_2(1-\alpha)\bigr)
       -1-\frac{m(m^2-1)}{6(N-1)}.
\]
For fixed $p>0$ and $\alpha>1/2$, there are
$c_{p,\alpha}>0$ and $m_0=m_0(p,\alpha)$ such that
$m_0\le m\le c_{p,\alpha}\sqrt N$ implies $q=\bOm[{p,\alpha}]{m}$.
\end{corollary}

On the stated event, $Z$ lies in a Hamming ball of radius $(1-\alpha)m$.
Counting that ball bounds its entropy; \cref{app:seqmhm-fraction-proof}
gives the calculation and the lower threshold $m_0$.

\paragraph{Role in ARHM.}
The sequential joint-success bound applies to uniform factors of the
public one-factorization and yields the update-first streaming lower
bound.  The simultaneous result \cref{cor:ext-factor-linear} instead
permits all requests to be known before decoding, for
$m_0(\delta,p_0)\le m\le N^{1/2-\delta}$ and fixed $\delta,p_0>0$.
Both transfer the complete post-data state as one message.  The
simultaneous bound is for independent factor draws and does not assert
hardness of every fixed list or include the sequential square-root
endpoint.

\section{A sharper transversal-forest bound below the square-root scale}
\label{app:transversal-forest-optimization}

We refine the transversal-forest estimate for $m=o(\sqrt N)$.
This includes values closer to $\sqrt N$ than those covered by the
linear bound for $m\le N^{1/2-\delta}$.
Write $a:=\ln(N/m)$ and $b:=\ln(N/m^2)$.

\begin{proposition}[Optimized transversal-forest scale]
\label{prop:app-optimized-transversal-forest}
Let $N$ tend to infinity through even integers, and let $M_1,\ldots,M_m$
be independent draws from any balanced matching distribution.  Suppose
$m=m(N)\to\infty$ and $m=o(\sqrt N)$.  For a universal $c>0$ and all
sufficiently large $N$,
\[
 \prob{\tfor(M_1,\ldots,M_m)\le cm\frac{b+\ln a}{a}}
 \le\exp[-\bOm{m(b+\ln a)}].
\]
Consequently, some deterministic list in the distribution's support has
\[
 \tfor(M_1,\ldots,M_m)
 =\bOm{m\frac{\ln(N/m^2)+\ln\ln(N/m)}{\ln(N/m)}}.
\]
\end{proposition}

The event controls every coordinated edge choice at once, including the
output of a joint quantum decoder.  The deterministic conclusion is an
existence statement about a list in the support.

\begin{proof}
For integers $1\le s\le N/2$, \cref{prop:random-transversal-forest-main}
and the binomial bound give
\begin{equation}\label{eq:app-transversal-estimate}
 \prob{\tfor\le s}
 \le\left(\frac{eN}{2s}\right)^{2s}
       \left(\frac{2s^2}{N-1}\right)^m.
\end{equation}
With $y=s/m$, the logarithm of the right-hand side, divided by $m$, is
\begin{equation}\label{eq:app-transversal-exponent}
 -b+2ya+2y\ln\frac{e}{2y}+2\ln y+\ln2+\lO{1}.
\end{equation}
The term $2ya$ is the cost of choosing a small supporting vertex set;
the negative terms from the matching probabilities must dominate it.

Put $d=b+\ln a$ and take $s=\floor{c_0md/a}$ for a sufficiently small
universal constant $c_0>0$.  If $s=0$, the claimed bad event is empty:
its real threshold is less than one, but every transversal is nonempty.
Otherwise, $a,b\to\infty$, $b\le a$, and eventually $d\le2a$.
For $c_0\le1/8$, we have $y\le1/4$ and hence $s\le m/4\le N/2$
eventually, so \eqref{eq:app-transversal-estimate} applies.

Since $y\ln(e/(2y))$ is bounded on $(0,1/4]$ and $y\le c_0d/a$,
\eqref{eq:app-transversal-exponent} yields
\begin{equation}\label{eq:app-transversal-exponent-two}
 \frac1m\ln\prob{\tfor\le s}
 \le-\left(b+2\ln\frac a d\right)+2c_0d+\bO{1}.
\end{equation}
If the probability vanishes, interpret its logarithm as $-\infty$.
We claim that, for sufficiently large $N$,
\begin{equation}\label{eq:app-transversal-claim}
 b+2\ln(a/d)\ge d/3.
\end{equation}
There are two cases.  If $b\ge\ln a$, then $d\le2b$ and $d\le2a$,
so $b+2\ln(a/d)\ge b-2\ln2\ge d/3$ eventually.
If $b<\ln a$, then $d<2\ln a$, and
$\ln(a/d)\ge\ln a-\ln(2\ln a)$; hence $2\ln(a/d)\ge d/2$
eventually.  This proves the claim.

Choose, for example, $c_0\le1/24$.  As $d\to\infty$,
\eqref{eq:app-transversal-exponent-two} is at most $-\bOm{d}$.
Thus $\prob{\tfor\le s}\le\exp[-\bOm{md}]$.
Because $\tfor$ is integer-valued, this is exactly the event with real
threshold $c_0md/a$.  Its complement has positive probability, proving
the deterministic existence statement as well.
\end{proof}

\begin{remark}
Balancing $2ay+2\ln y$ against $b$ suggests
$y\asymp W(ae^{b/2})/a$, where $W$ is the Lambert $W$ function.
Since $W(ae^{b/2})=\bT{b+\ln a}$, this agrees with the proposition's
scale up to constant factors.
\end{remark}

\subsection{Proof of the constant-fraction sequential bound}
\label{app:seqmhm-fraction-proof}

We use the uniform-input, fresh-request experiment of
\cref{cor:app-seqmhm-fraction}.  The event below is used only to bound
error entropy; the forest-growth estimate remains unconditional.

\begin{proof}[Proof of \cref{cor:app-seqmhm-fraction}]
Let $J$ indicate that $\abs Z\le(1-\alpha)m$, where $\abs Z$
denotes Hamming weight.  The Hamming-ball bound
$\abs{\set{z\in\bool^m:\abs z\le(1-\alpha)m}}\le2^{m h_2(1-\alpha)}$
gives
\[
\begin{aligned}
 H(Z)&\le H(J)+H(Z\mid J)\\
 &\le1+\prob{J=1}m h_2(1-\alpha)+(1-\prob{J=1})m\\
 &\le1+m-pm(1-h_2(1-\alpha)).
\end{aligned}
\]
Substitute into \cref{thm:sequential-hidden-matching} to obtain the
claimed lower bound.

For fixed $p>0$, put $\gamma=p(1-h_2(1-\alpha))>0$.
If $m\le\sqrt{3\gamma/2}\,\sqrt N$, then
$m(m^2-1)/(6(N-1))\le\gamma m/2$, using $N/(N-1)\le2$.
Hence $q\ge\gamma m/2-1\ge\gamma m/4$ whenever
$m\ge\ceil{4/\gamma}$.
The lower threshold is needed: when $m=1$, $\alpha=3/4$, and $p=1/2$,
a parity guess with no message meets the correctness requirement.
\end{proof}

\section{Simultaneous-request estimate at \texorpdfstring{$r=\floor{\sqrt N}$}{r = floor(sqrt N)}}
\label{sec:ext-endpoint}

For sequential requests, we proved a linear lower bound through
$m=\floor{\sqrt N}$.  For simultaneous requests at this endpoint,
the balanced-matching tail bound gives the following weaker estimate,
losing a factor of $\log N/\log\log N$.

\begin{proposition}[Near-square-root simultaneous factorized consumability]
\label{prop:ext-endpoint}
Let $m=\floor{\sqrt N}$ and draw $m$ independent uniform factors of any
one-factorization of $K_N$.  For a universal $c>0$ and all sufficiently
large even $N$,
\[
  \prob{\tfor(M_{I_1},\ldots,M_{I_m})
       \le c\frac{\sqrt N\log\log N}{\log N}}
  \le\exp[-\bOm{\sqrt N\log\log N}].
\]
Consequently, for uniform $X\in\bool^N$ independent of the factor draws
and protocol coins, every one-way protocol with joint success
at least a fixed $p_0\in(0,1]$ requires
\[
  q=\bOm[{p_0}]{\frac{\sqrt N\log\log N}{\log N}}.
\]
Alice encodes $X$ before learning any factor indices; Bob receives all
indices simultaneously and may decode jointly.  The budget $q$ includes
all transmitted classical information, and independent public coins are
allowed as in \cref{lem:forest-information-main}.
The same lower bound holds for post-data representations of
$\ARHM_{N,\floor{\sqrt N}}$ with the prefix-order and joint-success
conventions of \cref{thm:ext-batch-lower}.
\end{proposition}

This is a distributional lower bound for independent factor draws.
A fixed list repeating one factor can be answered from one phase state
by repeating the same edge and parity.

\begin{proof}
Write $\ell=\log\log N$ and let
$s=\ceil{\sqrt N\ell/(2\log N)}$.  For large $N$, $1\le s\le N/2$.
By \cref{prop:ext-balanced-tail} and $\binom N{2s}\le(eN/(2s))^{2s}$,
\begin{align*}
  \log T_{N,m}(s)
  &\le 2s\log\frac{eN}{2s}+m\log\frac{2s^2}{N-1}\\
  &= (1/2+\lO{1})m\ell+(-2+\lO{1})m\ell\\
  &= -(3/2+\lO{1})m\ell.
\end{align*}
Indeed, $\log(eN/(2s))=(1/2+\lO{1})\log N$, whereas
$\log(2s^2/(N-1))=-2\ell+2\log\ell+\bO{1}$.
Rounding $m$ and $s$ changes only lower-order terms.  Thus $T_{N,m}(s)\le\exp[-\bOm{m\ell}]$.  Since $s$ is the
ceiling of the claimed threshold, this also bounds the probability in
the proposition.

For communication, \cref{thm:ext-balanced-quant} and $h_2(p)\le1$ give
\[
  q\ge s\bigl(p_0-T_{N,m}(s)\bigr)-1.
\]
Since $T_{N,m}(s)$ tends to zero, for sufficiently large $N$ it is at most
$p_0/2$.  If also $s\ge4/p_0$, this gives $q\ge p_0s/4$.  Both conditions hold for
$s=\bT{\sqrt N\log\log N/\log N}$.
For the post-data bound, use the same prefix-to-message simulation as
in \cref{thm:ext-batch-lower}, applying this endpoint estimate in place
of \cref{cor:ext-factor-linear}.
\end{proof}

\section{Triangle-counting parameter and implementation details}

We give the parameter choices and implementation details for the
triangle estimator.  In particular, we show how to store the sampling
seeds and perform the arithmetic within the expected simultaneous-space
bound, using the graph parameters supplied in the theorem.

\begin{remark}[Known parameters]
\label{rem:triangle-known-parameters}
The theorem supplies $[n]$, exact $m,R$, and estimates
$T/2\le T_0\le T$, $\Delta_E\le\Delta_0\le2\Delta_E$ before
processing.  Substituting $T_0,\Delta_0$ in the lane counts is conservative
and changes space by constant factors.  Exact $T$ is not supplied, and
the promise $T\ge1$ remains in force.  If the stream length and $R$ are
known, their ratio determines $m$.

A parameter-free guarantee would need a proved procedure for choosing
and validating scales within the total space budget.  Running all
geometric guesses in parallel can be polynomially more expensive than
the correctly tuned copy; no such extension is asserted here.
\end{remark}

\begin{remark}[Small $R$]
The requirement $R\gtrsim L/\eps$ makes the warm-up and tail error
$\bO{\eps T}$.  If it fails, storing the distinct edge set and counting
exactly is a valid fallback using $\bO{m\log n}$ bits.
Kallaugher's original algorithm also applies when $R=1$; for $R>1$, its
arbitrary-order guarantee alone does not handle repeated edges.

With $L=\ceil{C\log(m(R+1)/\delta)}$ and fixed $\eps,\delta$, failure
of this condition confines $R$ to logarithmic scale in $m$, excluding
no polynomial factor-$R$ improvement.  For variable $\eps$, the excluded
range is $R=\bO{L/\eps}$ and may contain polynomially many repetitions.
\end{remark}

\begin{remark}[General accuracy and the choice of $k$]
The optimized corollary assumes the balancing value lies in $[C/\eps,m]$.
When this interval is nonempty and the repetition condition holds, use
\[
 k=\min\left\{m,\ \max\left\{\frac C\eps,
        \left(\frac{T_0\Delta_0}{L\sqrt m}\right)^{2/5}\right\}\right\}
\]
and substitute in \eqref{eq:main-space-explicit}.  Round up to an integer
for implementation.  If $m<C/\eps$, the exact-storage fallback has its
separate $\bO{m\log n}$ bound.  The lower constraint on $k$ controls the
classical counter bias from repeated closing edges
(\cref{lem:high-bias}); it is not required by the quantum moment bounds.
\end{remark}

\subsection{Hashes, arithmetic, and live space}
\label{app:triangle-implementation}

We implement the ideal classical sampler using small seeds and finite
precision.  Quantum circuit generation and control are subject to the
same space accounting as the other classical computation, according to
\cref{sec:space-convention}.

\paragraph{Consistent decisions on repetitions.}
Use pairwise-independent hashes on center vertices and four-wise-independent
hashes on directed edges, with the two families independent.  Each decision
is keyed by the underlying vertex or directed edge.  Reevaluate the same
hash on repeated arrivals, and insert an eligible directed edge only if
it is not already in the stored buffer.  No record of rejected edges is
needed.  Hashes use fresh seeds for each build block, lane, and group.
Products of two wedge indicators involve at most two center keys and
four directed keys, so these families reproduce all first and second
moments used in \cref{app:triangle-classical-proofs}.  Fixed-degree
polynomials over finite fields provide such families with polylogarithmic
seed length.

\paragraph{Finite sampling probabilities.}
First replace $k$ by $\ceil k\le m$.  This preserves the bias condition
and changes the lane bounds by only absolute factors; use the rounded
value throughout.  The activation probability $1/(kL)$ is now rational
and can be sampled exactly using $\bO{\log(kL)}$ bits of workspace.

Choose dyadic thresholds $p'_V\in[p_V/2,p_V]$ and
$p'_E\in[p_E/2,p_E]$ using $\bO{\log m}$ fractional bits, and normalize
by the actual inclusion probability $\alpha'=p'_V(p'_E)^2$.
Then $\alpha/8\le\alpha'\le\alpha$ and $p'_Vp'_E\le1/m$.
The within-center second-moment terms in the variance proof are
nonnegative and only decrease with the probabilities; distinct centers
still have zero covariance.  Therefore
\[
 \E{X_j}=3\alpha'\wt T_j^{>},\qquad
 \Var{X_j}\le112\alpha'\Delta_E\sum_\omega\wt h_j(\omega).
\]
Replacing $\alpha$ by $\alpha'$ in normalization thus changes the
required lane counts by constants only.

\paragraph{Arithmetic error on a good stream.}
Each block has at most $6LT$ oriented closing occurrences, since every
edge occurs at most $2L$ times and $\sum_e\tau_e=3T$.
Approximate each power contribution $1-(1-\rho)^D$ to error $\zeta$
and round each accumulator increment to error $\zeta/L$.
The resulting error in $X_j$ is at most $12T\zeta$, and after division
by $3\alpha'$ it is at most $4T\zeta/\alpha'$.
Take $\zeta\le\eps\alpha'/1000$; allow a further $\eps T/100$ for
accumulator and final normalization rounding.  This fits within the
reserved $\eps T/4$ implementation budget.  Averaging and medians
preserve a uniform additive bound on their inputs.

Repeated squaring with guard bits evaluates the powers in polylogarithmic
space.  The exponent is at most $2Lm$, and an elementary rounding error
$\bO{\zeta/(Lm)}$ suffices.  Since $\alpha'\ge m^{-3/2}/8$, the required
precision is polylogarithmic in $n,m,R,1/\eps$.  The quantum outputs are
integers in $\{-1,0,1\}$, so their running sums can be stored exactly;
only their final scaling needs rounding.

\paragraph{Space on every stream.}
At every fixed stage, each buffer is a subset of the $2m$ underlying
directed edges, each included with probability at most $1/m$.
Its expected size is therefore at most two, even on streams outside
$\calG$.  Counters, edge identifiers, and hashes need logarithmic or
polylogarithmic storage.  On all streams, counters are bounded by the
block length and accumulators sum polynomially many bounded terms, so
their allocated precision also has a deterministic polylogarithmic bound.
Scan stored samples using small work registers, padding scans to a fixed
length if necessary so the stage does not select a buffer by its random
size.  Discard old READY buffers before allocating replacements;
transient workspace adds only a constant factor to the expected buffer
size.  This proves the required
expected simultaneous-space bound at intermediate stages as well as
block boundaries, without an expected-peak claim.

\section{Proofs for the classical complementary triangle estimator}
\label{app:triangle-classical-proofs}

We prove the bias and variance bounds for the classical component of
the triangle estimator.  These determine how many classical lanes are
needed when each lane is resampled over successive blocks.

Fix a stream in $\calG$ and a query block $\calB_j$.  The weights and
occurrence counts are deterministic; expectations are over the lane's
sampling randomness.  Counters start at zero and closing contributions
are evaluated before increments.  We use the ideal probabilities
$p_V,p_E$; the implementation preserves the moments up to constants.
Directed keys $wu$ and $uw$ are distinct.

\begin{proof}[Proof of \cref{lem:high-bias}]
For integers $a,c\ge0$ and $0\le\rho\le1$,
\[
 0\le(1-\rho)^a-(1-\rho)^{a+c}
   =(1-\rho)^a[1-(1-\rho)^c]\le\rho c,
\]
with the convention $(1-\rho)^0=1$.
Use \eqref{eq:double-count-identity} and let $C=C_{j,uv}$.
The numbers of earlier $uv$ copies at its successive occurrences are
$0,1,\ldots,C-1$, so
\[
 0\le\wt h_j(w;u,v)-h_j(w;u,v)
 \le\frac\rho L\sum_{\substack{t\in\calB_j\\e_t=uv}}N_{j,uv}(t)
 =\frac\rho L\binom C2<\frac2k.
\]
The last step uses $C\le2L$ and $\rho=1/(kL)$.
Sum over the $3T$ orientations and divide by three to get
\eqref{eq:T-high-bias}.
\end{proof}

\begin{proof}[Proof of \cref{lem:high-weight-bounds}]
There are at most $2L$ summands in \eqref{eq:tilde-high}, each in $[0,1]$,
so $0\le\wt h_j\le2$.  Also
$D_{j,wu}(t)\le2Ld_u$ and $D_{j,wv}(t)\le2Ld_v$.
Using $1-(1-\rho)^a\le\rho a$, each summand is at most
$2(d_u+d_v)/k$.  Multiplying by at most $2L/L$ gives
$\wt h_j\le4(d_u+d_v)/k$.
\end{proof}

\begin{proof}[Proof of \cref{lem:classical-mean}]
In \eqref{eq:X-sum}, every indicator has mean $\alpha$, and
$\sum_\omega\wt h_j(\omega)=3\wt T_j^{>}$.
\end{proof}

\begin{proof}[Proof of \cref{lem:classical-var}]
For each center $w$, let $\calT_w$ be the unordered pairs
$\{u,v\}\subseteq N(w)$ with $uv\in E$.  Put
\[
 z_{uv}:=\wt h_j(w;u,v),\qquad
 W_w:=\sum_{\{u,v\}\in\calT_w}z_{uv},\qquad
 X_w:=G_w\sum_{\{u,v\}\in\calT_w}z_{uv}H_{wu}H_{wv}.
\]
Then $X_j=\sum_wX_w$.  Distinct centers use distinct directed keys;
their contributions are independent under ideal sampling and have zero
covariance under the limited-independence implementation.  Indeed, a
product of two wedge indicators uses at most two centers and four
directed keys, so pairwise-independent center hashes and four-wise-independent
edge hashes reproduce every moment below.  It follows that
$\Var{X_j}=\sum_w\Var{X_w}\le\sum_w\E{X_w^2}$.

Fix $w$.  Expanding $X_w^2$ gives three types of ordered pairs.
\begin{enumerate}[leftmargin=1.8em]
\item \textbf{The same triangle.}  Since $z_{uv}^2\le2z_{uv}$,
the diagonal contribution is at most
$2p_Vp_E^2W_w=2\alpha W_w$.

\item \textbf{One shared radial edge.}  A centered triangle has at most
$2\Delta_E$ partners sharing $wu$ or $wv$.  Joint inclusion has
probability $p_Vp_E^3\le\alpha$, and the partner's weight is at most
two.  Their total contribution is at most
$4\alpha\Delta_EW_w$.

\item \textbf{Four distinct radial edges.}  Joint inclusion has
probability $p_Vp_E^4$.  The degree-weight bound gives
\begin{equation}\label{eq:Ww-bound}
 W_w\le\frac4k\sum_{\{u,v\}\in\calT_w}(d_u+d_v)
 \le\frac{4\Delta_E}k\sum_{u\in N(w)}d_u
 \le\frac{8m\Delta_E}k.
\end{equation}
Each $u\in N(w)$ appears in at most $\Delta_E$ pairs because all
those triangles contain $wu$.  Hence the disjoint-pair contribution is
at most
\[
 p_Vp_E^4W_w^2\le p_Vp_E^4\frac{8m\Delta_E}k W_w
              =8\alpha\Delta_EW_w,
\]
where $p_E^2=k/m$ and $\alpha=p_Vp_E^2$.
\end{enumerate}
Since $\Delta_E\ge1$, the three cases total at most
$14\alpha\Delta_EW_w$.  Summing over centers proves
\eqref{eq:classical-var}.
\end{proof}

\section{Full quantum Noisy Gap Cycle lower bound}
\label{app:ngc-full-proof}
\label{app:random-order}

We give the reductions from random-order Noisy Gap Cycle to Block
Hidden XOR, completing the lower bound from \cref{sec:ngc}.
The graph construction and reductions follow Assadi and
Sundaresan~\cite{AssadiSundaresan2023}; we use the quantum communication
bound of \cref{thm:qbhx-master} in the final step.
\Cref{fig:qro-pipeline} summarizes the reductions.  A small-space
streaming algorithm would give a short quantum message for BHX,
contradicting that bound.

\begin{figure}[t]
\centering
\begin{tabular}{p{0.29\linewidth} p{0.62\linewidth}}
\toprule
Reduction & Information retained \\
\midrule
Random-order NGC to robust NGC & The complete live streaming state is sent once. \\
Robust NGC to DHX & A neighboring hybrid isolates one hidden path parity. \\
DHX to unequal-block BHX & Clean bits belong to Alice; independent uniform indices belong to Bob. \\
\bottomrule
\end{tabular}
\caption{Protocol reductions.  The BHX lower bound propagates back to random-order NGC.}
\label{fig:qro-pipeline}
\end{figure}

\subsection{Noisy Gap Cycle}
\label{subsec:qro-main-result}

We restate the target promise and guarantees to fix the parameters used
throughout the construction.

\begin{definition}[Noisy Gap Cycle]
\label{def:NGC-main}
Let $n$ and $k\ge3$ be integers with $4k\mid n$.  In $\NGC_{n,k}$ the input is an
$n$-vertex graph promised to be in one of the following two cases:
\begin{enumerate}[leftmargin=1.6em]
    \item it contains $n/(2k)$ vertex-disjoint cycles of length $k$; or
    \item it contains $n/(4k)$ vertex-disjoint cycles of length $2k$.
\end{enumerate}
In both cases the graph additionally contains $n/(2k)$ vertex-disjoint
paths of length $k-1$, disjoint from all cycles and from one another.
There are no other edges: the graph is exactly the disjoint union of
these cycles and paths.  In the robust one-way communication problem,
each edge is independently assigned to Alice or Bob with probability
$1/2$.  Each player receives their assigned edge set; Alice sends one
message to Bob, who must distinguish the two cases.  In the quantum
model Alice's message is a quantum state.
\end{definition}

The term \emph{robust} requires success for every fixed promise graph,
averaged over the independent fair edge-ownership coins and protocol
randomness.  In streaming, average instead over a uniform edge
order and algorithm randomness, with output at the end of the pass.

\begin{theorem}[Quantum robust Noisy Gap Cycle]
\label{thm:quantum-robust-ngc-main}
For all sufficiently large admissible $n,k$ satisfying
$k\ge1600\ln n$,
every robust one-way quantum protocol for $\NGC_{n,k}$ with
error at most $1/3$ communicates $\bOm{n}$ qubits.
\end{theorem}

The corresponding streaming theorem follows from the robust-to-random
order transfer proved in \cref{subsec:robust-to-streaming-main}.

\begin{theorem}[Quantum random-order Noisy Gap Cycle]
\label{thm:quantum-random-order-ngc-main}
For all sufficiently large admissible $n,k$ satisfying
$k\ge1600\ln n$,
every one-pass quantum streaming algorithm that solves $\NGC_{n,k}$ with
error at most $1/3$ in uniformly random edge order requires $\bOm{n}$
space.
\end{theorem}

To prove the robust theorem we need the structured hard distribution of
Assadi and Sundaresan.  We spell out the construction because it explains
both the definition of DHX and the later clean-coordinate reduction.

\subsection{The Assadi--Sundaresan graph construction}
\label{subsec:as-construction}

\paragraph{Grouped layered graphs.}
Fix a width $w$ and a depth $d$.  A grouped layered graph has layers
$V^1,\ldots,V^d$,
each containing $2w$ vertices, partitioned into $w$ ordered pairs
\[
    g^\ell_j=(a^\ell_j,b^\ell_j),
    \qquad
    j\in[w],\ \ell\in[d].
\]
Between every two consecutive layers there is a perfect matching with the
property that the two vertices of any group $g^\ell_j$ are matched into
the two vertices of a single group in the next layer.  Consequently, the
two vertices of every first-layer group start two vertex-disjoint paths
through the graph.

All widths, depths, and multiplicities in these constructions are positive
integers.  Two elementary matchings are used to build the hard instances.

\begin{definition}[XOR matching]
\label{def:xor-matching}
For $x\in\bool^w$, $\operatorname{XOR\text{-}Matching}(x)$ is the
width-$w$, depth-$2$ grouped layered graph in which, for every $j\in[w]$,
\begin{itemize}[leftmargin=1.6em]
    \item if $x_j=0$, the matching contains
    $(a^1_j,a^2_j)$ and $(b^1_j,b^2_j)$;
    \item if $x_j=1$, the matching contains
    $(a^1_j,b^2_j)$ and $(b^1_j,a^2_j)$.
\end{itemize}
Thus the matching stays inside group $j$ and swaps the two paths exactly
when $x_j=1$.
\end{definition}

\begin{definition}[Permutation matching]
\label{def:perm-matching}
For a permutation $\sigma\in S_w$,
$\operatorname{Perm\text{-}Matching}(\sigma)$ is the width-$w$,
depth-$2$ grouped layered graph containing
\[
    (a^1_j,a^2_{\sigma(j)})
    \qquad\text{and}\qquad
    (b^1_j,b^2_{\sigma(j)})
\]
for every $j\in[w]$.
\end{definition}

The basic four-layer gadget conjugates an XOR matching by a random
permutation.

\begin{definition}[Block]
\label{def:as-block}
For $x\in\bool^w$ and $\sigma\in S_w$, define
\[
    \operatorname{Block}(x,\sigma)
    :=
    \operatorname{Perm\text{-}Matching}(\sigma)
    \circ
    \operatorname{XOR\text{-}Matching}(x)
    \circ
    \operatorname{Perm\text{-}Matching}(\sigma^{-1}),
\]
where consecutive gadgets are concatenated by identifying the last layer
of one with the first layer of the next.
\end{definition}

\begin{observation}[Path behavior in one block]
\label{obs:block-path-behavior}
Let $B=\operatorname{Block}(x,\sigma)$, with layers $1,2,3,4$.  For every
$j\in[w]$ the two paths beginning at $g^1_j$ end at $g^4_j$.  More
precisely,
\begin{itemize}[leftmargin=1.6em]
    \item if $x_{\sigma(j)}=0$, then
    $a^1_j$ is connected to $a^4_j$ and $b^1_j$ to $b^4_j$;
    \item if $x_{\sigma(j)}=1$, then
    $a^1_j$ is connected to $b^4_j$ and $b^1_j$ to $a^4_j$.
\end{itemize}
\end{observation}

\begin{proof}
The first and third matchings relabel the group index by $\sigma$ and
$\sigma^{-1}$, so the final group is again $j$.  The only operation that
can exchange the two vertices inside the group is the middle XOR matching,
and it exchanges them exactly when $x_{\sigma(j)}=1$.
\end{proof}

Now concatenate $t$ blocks.

\begin{definition}[Multi-Block graph]
\label{def:multi-block}
Let
\[
    X=(x^1,\ldots,x^t),
    \qquad x^i\in\bool^w,
\]
and
\[
    \Sigma=(\sigma^1,\ldots,\sigma^t),
    \qquad \sigma^i\in S_w.
\]
The graph
\[
    G=\operatorname{Multi\text{-}Block}(X,\Sigma)
\]
is the concatenation
\[
    \operatorname{Block}(x^1,\sigma^1)
    \circ\cdots\circ
    \operatorname{Block}(x^t,\sigma^t).
\]
It has width $w$ and depth $3t+1$.  For each $j\in[w]$, define
\[
    z_G(j)
    :=
    \bigoplus_{i=1}^t x^i_{\sigma^i(j)}.
\]
\end{definition}

\begin{observation}[The hidden XOR is a path-connectivity bit]
\label{obs:multiblock-path-parity}
Let $G=\operatorname{Multi\text{-}Block}(X,\Sigma)$ have depth
$d=3t+1$.  For every $j\in[w]$:
\begin{itemize}[leftmargin=1.6em]
    \item if $z_G(j)=0$, the two paths connect
    $a^1_j$ to $a^d_j$ and $b^1_j$ to $b^d_j$;
    \item if $z_G(j)=1$, the two paths connect
    $a^1_j$ to $b^d_j$ and $b^1_j$ to $a^d_j$.
\end{itemize}
\end{observation}

\begin{proof}
By \cref{obs:block-path-behavior}, block $i$ swaps the two paths from group
$j$ exactly when $x^i_{\sigma^i(j)}=1$.  The composition swaps the paths
exactly when an odd number of blocks swap, namely when
$\bigoplus_i x^i_{\sigma^i(j)}=1$.
\end{proof}

\subsubsection{The hard distribution for Noisy Gap Cycle}
\label{subsubsec:ngc-hard-distribution}

Assume for the moment that
$k=3t+1, \qquad n=4km$,
so a width-$2m$ Multi-Block graph of depth $k$ has exactly $n$ vertices.
The hard distribution $\mu_{\NGC{}}$ is generated as follows.
\begin{enumerate}[leftmargin=1.6em]
    \item Sample $\theta\in\bool$ uniformly.
    \item Sample
    \[
        X=(x^1,\ldots,x^t)\in(\bool^{2m})^t
    \]
    and
    \[
        \Sigma=(\sigma^1,\ldots,\sigma^t)\in(S_{2m})^t
    \]
    uniformly subject to
    \[
        z_G(j)=\theta
        \qquad\text{for every }j\in[m],
    \]
    where $G=\operatorname{Multi\text{-}Block}(X,\Sigma)$.
    \item For every $j\in[m]$, add the two auxiliary edges
    \[
        (a^k_j,a^1_j)
        \qquad\text{and}\qquad
        (b^k_j,b^1_j).
    \]
\end{enumerate}
The groups $j>m$ receive no auxiliary edges.

\begin{lemma}[Validity of the NGC hard distribution]
\label{lem:ngc-hard-validity}
Every graph in the support of $\mu_{\NGC{}}$ is a valid instance of
$\NGC_{n,k}$.  If $\theta=0$ it contains $n/(2k)$ disjoint $k$-cycles;
if $\theta=1$ it contains $n/(4k)$ disjoint $2k$-cycles.  In either case
it also contains $n/(2k)$ disjoint paths of length $k-1$.
\end{lemma}

\begin{proof}
For every $j\le m$, \cref{obs:multiblock-path-parity} gives two paths of
length $k-1$ connecting the first and final copies of group $j$.
If $z_G(j)=0$, the paths are parallel, and the two auxiliary edges close
them separately into two $k$-cycles.  Hence the first $m$ groups produce
$2m=n/(2k)$ disjoint $k$-cycles.

If $z_G(j)=1$, the paths are crossed.  The same two auxiliary edges then
join the two paths into a single cycle of length $2k$.  Hence the first
$m$ groups produce $m=n/(4k)$ disjoint $2k$-cycles.

For every $j>m$ there are no auxiliary edges, so the two paths beginning
at $a^1_j$ and $b^1_j$ remain paths of length $k-1$.  There are $m$ such
groups and therefore $2m=n/(2k)$ noisy paths.  All components are
disjoint because every inter-layer graph is a perfect matching.
\end{proof}

The hard distribution forces the first $m$ parities to agree.  The
following hybrid argument isolates one of them.

\subsection{Decorrelation: from Noisy Gap Cycle to DHX}
\label{app:qngc-hybrids}
\label{subsec:ngc-to-dhx-main}

\begin{definition}[Distributional Hidden XOR]
\label{def:DHX-main}
Fix positive integers $w,t$.  In $\DHX_{w,t}$, sample independently and uniformly
\[
    X=(x^1,\ldots,x^t),
    \qquad x^i\in\bool^w,
\]
and
\[
    \Sigma=(\sigma^1,\ldots,\sigma^t),
    \qquad \sigma^i\in S_w,
\]
and form
\[
    G=\operatorname{Multi\text{-}Block}(X,\Sigma).
\]
Every edge of $G$ is then independently assigned to Alice or Bob with
probability $1/2$.  Alice receives her labeled edges, Bob receives his
labeled edges, Alice sends one message, and Bob must output
\[
    z_G(1)
    =
    \bigoplus_{i=1}^t x^i_{\sigma^i(1)}.
\]
Success probability is over the random choice of $(X,\Sigma)$, the random
edge partition, and the protocol's randomness.
\end{definition}

Unlike NGC, DHX has independent strings and permutations and asks for
only one path parity.  The reduction pays for this independence by
retaining advantage only $\bT{1/w}$.  Continue to assume
$k=3t+1$ and $n=4km$.

For $h\in\set{0,\ldots,m}$ let $\mathcal {H}(h)$ be the distribution obtained
by sampling a width-$2m$ Multi-Block instance uniformly subject to
\[
    z_G(j)=
    \begin{cases}
        0,&j\le h,\\
        1,&h<j\le m,
    \end{cases}
\]
and then adding the same auxiliary edges as in
\cref{subsubsec:ngc-hard-distribution}.  Finally, every edge is robustly
partitioned between Alice and Bob.
Then $\mathcal {H}(0)$ and $\mathcal {H}(m)$ are precisely the two endpoint
cases of $\mu_{\NGC{}}$ (with opposite values of $\theta$).
Intermediate hybrids need not satisfy the NGC promise; the assumed
protocol is extended to those inputs with the same bounded message
register; no correctness on the intermediate hybrids is required.

\begin{lemma}[Quantum neighboring hybrid]
\label{lem:neighboring-hybrid-main}
\label{lem:app-quantum-ngc-hybrid}
Suppose a fixed-coin $q$-qubit protocol distinguishes the two endpoint
distributions $\mathcal {H}(0)$ and $\mathcal {H}(m)$, with equal prior
probabilities, with success at least $2/3$.  Then for some $h^\star\in[m]$, the neighboring pair
$\mathcal {H}(h^\star-1), \qquad \mathcal {H}(h^\star)$
can be distinguished, again with equal priors, with success at least
$\frac{1}{2}+\frac{1}{6m}$.
\end{lemma}

\begin{proof}
Let $\omega_h$ denote Bob's cq-state immediately before decoding under
$\mathcal {H}(h)$, consisting of his classical input registers and the
received quantum message: if $A,B$ denote the two local edge sets and
$\rho_A$ is Alice's encoding, then
$\omega_h=\E[(A,B)\sim\mathcal{H}(h)]{\ket{B}\!\bra{B}\otimes\rho_A}$.
Write $D(\rho,\sigma)=\tfrac12\trnorm{\rho-\sigma}$ for trace distance.
The optimal probability of distinguishing equiprobable states $\rho$
and $\sigma$ is $(1+D(\rho,\sigma))/2$.  Since the original decoder
succeeds with probability at least $2/3$, we have
$D(\omega_0,\omega_m)\ge1/3$.  Thus
\[
 \frac13\le\sum_{h=1}^mD(\omega_{h-1},\omega_h),
\]
so some $h^\star$ has distance at least $1/(3m)$.
For this pair the optimal distinguishing probability is therefore at
least $1/2+1/(6m)$.
This argument uses the pre-decoding cq-states; it does not retain a copy
of a measured message.

\end{proof}

The point of neighboring hybrids is that they differ in only one hidden
parity, $z_G(h^\star)$.  A DHX instance can be embedded precisely into
that parity while all other constrained parities are filled with public
randomness.

\label{app:qngc-reduction}
\begin{lemma}[Local embedding of DHX into a neighboring hybrid]
\label{lem:ngc-dhx-main}
\label{lem:app-ngc-to-dhx}
Let
\[
    k=3t+1,
    \qquad
    n=4km,
    \qquad
    w=m+1.
\]
A $q$-qubit protocol distinguishing the neighboring pair from
\cref{lem:neighboring-hybrid-main} with success
$1/2+1/(6m)$ induces a $q$-qubit protocol for $\DHX_{w,t}$ with success at
least
$\frac{1}{2}+\frac{1}{6w}$.
\end{lemma}

\begin{proof}
Let
\[
    G_{\DHX{}}
    =
    \operatorname{Multi\text{-}Block}(Y,\Phi)
\]
be the DHX instance, where
$Y=(y^1,\ldots,y^t)$ and
$\Phi=(\phi^1,\ldots,\phi^t)$, each $y^i\in\bool^{m+1}$ and
$\phi^i\in S_{m+1}$.
We construct a width-$2m$ Multi-Block instance
$G'=\operatorname{Multi\text{-}Block}(X,\Sigma)$.  After adding the
public closing edges and the robust ownership described below, the
augmented, partitioned instance is distributed as the uniform mixture of
$\mathcal {H}(h^\star-1)$ and $\mathcal {H}(h^\star)$.

For every block $i$, use public randomness to choose a uniformly random
injective partial map
\[
    \sigma^i:
    [m]\setminus\set{h^\star}
    \longrightarrow [2m].
\]
Let
\[
    I_i^\circ
    :=
    \set{\sigma^i(j):j\in[m]\setminus\set{h^\star}}.
\]
Thus $\abs{I_i^\circ}=m-1$ and its complement has size $m+1=w$.
Using further public randomness, choose the bits
$x^i_u$ for $u\in I_i^\circ$, for all $i$, uniformly subject to
\[
    \bigoplus_{i=1}^t x^i_{\sigma^i(j)}
    =
    \begin{cases}
        0,&j<h^\star,\\
        1,&h^\star<j\le m.
    \end{cases}
\]
These are exactly the hybrid constraints away from $h^\star$.

Fix publicly a bijection
\[
    f_i:[m+1]\longrightarrow [2m]\setminus I_i^\circ.
\]
The remaining domain positions of $\sigma^i$ are
\[
    \set{h^\star}\cup\set{m+1,\ldots,2m},
\]
which also has size $m+1$.  Fill these positions using the DHX permutation
$\phi^i$: set
\[
    \sigma^i(h^\star)=f_i(\phi^i(1))
\]
and, for $s\in\set{2,\ldots,m+1}$, set
\[
    \sigma^i(m+s-1)=f_i(\phi^i(s)).
\]
Finally fill the still-unassigned coordinates of $x^i$ using the DHX
string:
\[
    x^i_{f_i(s)}=y^i_s,
    \qquad s\in[m+1].
\]

We verify the two properties needed for the reduction.
First, $\Sigma$ is a tuple of independent uniform permutations of
$[2m]$.  Indeed, the images of the first $m-1$ fixed domain positions are
a uniformly random ordered $(m-1)$-tuple of distinct values, and the
remaining images are filled by an independent uniform permutation
$\phi^i$ of the remaining $m+1$ values.  Conditional on $\Sigma$, the
coordinates of $X$ away from the prescribed hybrid parities are uniform:
the publicly sampled coordinates are uniform subject only to the required
constraints, while the coordinates filled from $Y$ are independent
uniform bits.  Thus the augmented, partitioned $G'$ has exactly the mixture distribution
\[
    \frac{1}{2}\mathcal {H}(h^\star-1)
    +
    \frac{1}{2}\mathcal {H}(h^\star),
\]
because the one unconstrained parity $z_{G'}(h^\star)$ remains uniform.

Second, the unconstrained parity is precisely the DHX answer.  For every
block $i$,
\[
    x^i_{\sigma^i(h^\star)}
    =
    x^i_{f_i(\phi^i(1))}
    =
    y^i_{\phi^i(1)},
\]
so
\[
    z_{G'}(h^\star)
    =
    \bigoplus_{i=1}^t y^i_{\phi^i(1)}
    =
    z_{G_{\DHX{}}}(1).
\]
The two neighboring hybrids differ exactly in the value of this bit, up
to the fixed convention of which hybrid is labeled $0$.

The construction is local under the robust partition.  Edges originating
from the DHX instance are only relabeled, so Alice and Bob can relabel
their own edges independently.  All filler edges and all auxiliary edges
are determined by public randomness, and their owners can also be sampled
publicly and independently.  Therefore no communication is needed before
Alice prepares the message of the assumed neighboring-hybrid protocol.
The quantum message is prepared once, after the classical embedding; it is
never copied or rewound.

Running the neighboring-hybrid distinguisher therefore recovers
$z_{G_{\DHX{}}}(1)$ with success
\[
    \frac{1}{2}+\frac{1}{6m}
    >
    \frac{1}{2}+\frac{1}{6(m+1)}
    =
    \frac{1}{2}+\frac{1}{6w}.
\]
The message length remains $q$.  The embedding introduced public coins
independent of the DHX input and its partition.  Before applying the DHX
lower bound, fix a seed whose average success under that input law is at
least the average over these coins.  This preserves the DHX input law and
message size; the subsequent argument needs only that average success,
not that the fixed-seed embedding itself has the full hybrid law.
\end{proof}

Thus proving a low-advantage lower bound for DHX is enough to prove the
robust NGC lower bound.  We now isolate the quantum information problem
inside DHX.

\subsection{Lower-bounding DHX: clean coordinates and residual Block Hidden XOR}
\label{subsec:quantum-dhx-main}
\label{app:qdhx-bookkeeping}

We expose a BHX instance by fixing only classical input-generation
variables.  Three points need separate checks: the success loss from
conditioning, the residual product input law, and a local simulation of
each party's edge share.  Preserving message dimension alone would not
establish the reduction.

\subsubsection{Representing the random edge partition}

Fix one block $B=\operatorname{Block}(x,\sigma)$ with layers
$V^1,V^2,V^3,V^4$.  Because every inter-layer graph is a perfect matching,
a uniform independent partition of its edges can be generated by three
independent uniform functions
\[
    f_L,f_M,f_R:[2w]\to\set{\alpha,\beta}.
\]
The convention is:
\begin{itemize}[leftmargin=1.6em]
    \item an edge from $V^1$ to a vertex $u\in V^2$ belongs to Alice iff
    $f_L(u)=\alpha$;
    \item an edge from $u\in V^2$ to $V^3$ belongs to Alice iff
    $f_M(u)=\alpha$;
    \item an edge from $v\in V^3$ to $V^4$ belongs to Alice iff
    $f_R(v)=\alpha$.
\end{itemize}
The value $\beta$ assigns the edge to Bob.  Since each relevant vertex is
incident to exactly one edge across the specified interface and all
function values are independent, this is exactly the independent
$1/2$--$1/2$ robust edge partition.  Crucially, these partition functions
are independent of $x$ and $\sigma$.

\begin{definition}[Clean coordinate]
\label{def:clean-coordinate}
An index $j\in[w]$ is \emph{clean} in $B$ if
\[
\begin{aligned}
    f_L(a^2_j)=f_L(b^2_j)&=\beta,\\
    f_M(a^2_j)=f_M(b^2_j)&=\alpha,\\
    f_R(a^3_j)=f_R(b^3_j)&=\beta.
\end{aligned}
\]
Write $C=\operatorname{clean}(B)$ for the set of clean coordinates.
\end{definition}

At a clean coordinate $j$, Alice owns the middle XOR edges encoding
$x_j$, while Bob owns the adjacent permutation edges revealing
$\sigma^{-1}(j)$.  To rule out Alice inferring the selected coordinate
from her other edges, we will establish the full residual law and local
simulation below.

A coordinate is clean with probability $2^{-6}$, independently across
$j$, so $\abs{C}\sim\operatorname{Bin}(w,1/64)$.
Recall that a Bernoulli indicator of success probability $p$ has mean
$p$ and variance $p(1-p)$.  The sum of $s$ independent such indicators
with the same $p$ has the binomial law $\operatorname{Bin}(s,p)$,
with mean $sp$ and variance $sp(1-p)$.
More generally, a sum $Y$ of independent Bernoulli indicators, with
possibly different success probabilities and mean $\nu$, obeys
\[
 \prob{Y\le(1-\eta)\nu}\le e^{-\eta^2\nu/2},\qquad 0<\eta<1.
\]
Applying this Chernoff bound with $\nu=w/64$ and $\eta=9/25$ gives,
for sufficiently large $w$,
\begin{equation}
\label{eq:clean-many}
    \prob{\abs{C}\ge\frac{w}{100}}
    \ge1-w^{-4}.
\end{equation}

For DHX, write $F_i=(f_L^i,f_M^i,f_R^i)$ for the partition-function
tuple in block $i$, and let $C_i$ be its clean set.  Call block $i$
\emph{usable} if
\[
    \abs{C_i}\ge\frac{w}{100}
    \qquad\text{and}\qquad
    \sigma^i(1)\in C_i.
\]
The second condition says that the particular bit from block $i$ appearing
in $z_G(1)$ lies among the clean coordinates.

\begin{lemma}[Many usable blocks]
\label{lem:many-usable}
\label{lem:app-many-usable}
For sufficiently large $w$, if $t\ge512\ln w$, then with probability at
least $1-w^{-9/8}$ there are at least $t/512$ usable blocks.
\end{lemma}

\begin{proof}
Because $C_i$ depends only on the partition functions and $\sigma^i$ is
uniform and independent of them,
\[
    \prob{\sigma^i(1)\in C_i}
    =
    \E{\frac{\abs{C_i}}{w}}
    =
    \frac{1}{64}.
\]
Together with \eqref{eq:clean-many},
\[
\begin{aligned}
    \prob{i\text{ is usable}}
    &\ge
    \prob{\sigma^i(1)\in C_i}
    -
    \prob{\abs{C_i}<w/100}\\
    &\ge
    \frac{1}{64}-w^{-4}
    \ge
    \frac{1}{128}
\end{aligned}
\]
for sufficiently large $w$.  Usability is independent across blocks.
If $R$ is the number of usable blocks, then
$\E{R}\ge t/128$.  A multiplicative Chernoff bound gives
\[
    \prob{R<t/512}
    \le
    \exp(-9t/4096)
    \le
    w^{-9/8}
\]
when $t\ge512\ln w$.
\end{proof}

The following elementary conditioning facts track message size and
success.  The hard input law will be checked separately.

\begin{lemma}[Classical conditioning and fixing]
\label{lem:classical-conditioning-fixing}
\label{lem:app-classical-conditioning}
\label{lem:app-fix-classical-branch}
\label{lem:app-high-prob-conditioning}
Let Alice's encoding be $x\mapsto\rho_x$ in a $q$-qubit register.
The classical input-generation variables considered here are independent
of protocol-private randomness conditional on the classical inputs.
Then the following hold.
\begin{enumerate}[leftmargin=1.6em]
    \item Conditioning the classical input distribution on any
    positive-probability classical event leaves an ensemble of states in
    the same $q$-qubit register.
    \item If $C$ is a finite classical random variable, some positive-probability
    value $c$ has
    conditional protocol success at least the average success over $C$.
    \item If the protocol succeeds with probability at least
    $1/2+\gamma$, where $0\le\gamma\le1/2$, and a classical event
    $E$ has probability at least $1-\delta$, where $0\le\delta<1$, then
    \[
        \prob{\mathrm{success}\mid E}
        \ge
        \frac{1/2+\gamma-\delta}{1-\delta}
        =
        \frac{1}{2}+\frac{\gamma-\delta/2}{1-\delta}.
    \]
\end{enumerate}
\end{lemma}

\begin{proof}
Conditioning changes ensemble weights, not the message register.
The second claim follows by averaging conditional success over $C$.
For the third,
\[
 \prob{\mathrm{success}}
 \le1-\prob{E}\bigl(1-\prob{\mathrm{success}\mid E}\bigr)
 \le\delta+(1-\delta)\prob{\mathrm{success}\mid E}.
\]
Rearrange using $\prob{\mathrm{success}}\ge1/2+\gamma$.

\end{proof}

\begin{lemma}[The DHX advantage survives usable-block conditioning]
\label{lem:success-usable}
\label{lem:app-success-usable}
If a protocol for $\DHX_{w,t}$ succeeds with probability at least
$1/2+1/(6w)$ and $t\ge512\ln w$, then for all sufficiently large $w$ its
success conditioned on having at least $t/512$ usable blocks is at least
$\frac{1}{2}+\frac{1}{8w}$.
\end{lemma}

\begin{proof}
Let $E$ be the event from \cref{lem:many-usable}; then
$\prob{\neg E}\le w^{-9/8}$.  Apply \cref{lem:classical-conditioning-fixing} with
$\gamma=1/(6w)$ and $\delta=w^{-9/8}$.  Since
$w^{-9/8}\le1/(12w)$ for sufficiently large $w$, the resulting
conditional advantage is at least $1/(8w)$.
\end{proof}

\subsubsection{The residual distribution in a usable block}

Condition on the event in \cref{lem:many-usable}.  By
\cref{lem:classical-conditioning-fixing}, fix a branch consisting of:
\begin{itemize}[leftmargin=1.6em]
    \item the exact usable set $I\subseteq[t]$, with
    $r:=\abs{I}\ge t/512$;
    \item all partition functions in every block;
    \item every nonusable block completely; and
    \item in each usable block $i$, all triples
    \[
        \bigl((\sigma^i)^{-1}(j),j,x^i_j\bigr)
        \qquad\text{for }j\notin C_i.
    \]
\end{itemize}
Choose the fixed branch so that the success probability is still at least
$1/2+1/(8w)$.  The branch is selected by averaging its conditional success; within it,
we retain the original conditional input law, without conditioning on the
success event or any message measurement.  All fixed values may be
hardwired into the resulting protocol.

Fix one usable block $i\in I$ and write
$w_i:=\abs{C_i}\ge\frac{w}{100}$.
The fixed inverse images of the nonclean coordinates determine the set
\[
    P_i
    :=
    [w]\setminus
    \set{(\sigma^i)^{-1}(j):j\notin C_i}.
\]
Thus $\abs{P_i}=\abs{C_i}=w_i$.  Since the block is usable,
$\sigma^i(1)\in C_i$, and hence $1\in P_i$.

\begin{lemma}[Residual usable-block distribution]
\label{lem:residual-usable}
\label{lem:app-residual-usable}
Under the preceding fixing, for every usable block $i$:
\begin{enumerate}[leftmargin=1.6em]
    \item the clean bits $(x^i_j)_{j\in C_i}$ are independent uniform
    bits;
    \item the remaining restriction
    $\sigma^i|_{P_i}:P_i\to C_i$ is a uniform bijection;
    \item consequently $J_i:=\sigma^i(1)$ is uniform on $C_i$;
    \item conditioned on $J_i=j$, the rest of the bijection is a uniform
    bijection $P_i\setminus\set{1}\to C_i\setminus\set{j}$;
    \item the clean bits and remaining permutation randomness are
    independent.
\end{enumerate}
The residual variables of distinct usable blocks are independent.
\end{lemma}

\begin{proof}
The clean set $C_i$ is determined solely by the already-fixed partition
functions.  Since $x^i$ was originally uniform and independent of
$(\sigma^i,F_i)$, fixing the bits outside $C_i$ leaves the bits inside
$C_i$ independent and uniform.

Similarly, $\sigma^i$ was a uniform permutation.  Fixing the inverse image
of every value outside $C_i$ leaves exactly a uniformly random bijection
from the remaining domain $P_i$ to the remaining image set $C_i$.
Because $1\in P_i$, symmetry makes $\sigma^i(1)$ uniform on $C_i$; fixing
its value leaves a uniform bijection on the remaining points.  The string
and permutation were independent initially, and all conditioning is
blockwise, so the claimed independence is preserved, both within and
across usable blocks.
\end{proof}

We can now see the Block Hidden XOR instance explicitly.  For each usable
$i$, fix an arbitrary bijection
$\tau_i:[w_i]\to C_i$,
and define
\[
    Y^i_a:=x^i_{\tau_i(a)},
    \qquad
    K_i:=\tau_i^{-1}(\sigma^i(1)).
\]
By \cref{lem:residual-usable}, the $Y^i$ are independent uniform strings,
the $K_i$ are independent uniform indices, and the two families are
independent.  All nonusable blocks were fixed, so their contribution is a known offset.
Thus
\begin{equation}
\label{eq:dhx-residual-bhx}
    z_G(1)
    =
    \left(\bigoplus_{i\notin I}x^i_{\sigma^i(1)}\right)
    \oplus\bigoplus_{i\in I}Y^i_{K_i}.
\end{equation}

There is one final simulation point.  In the actual DHX instance Bob sees
more than the single index $K_i$: his clean permutation edges reveal the
rest of the residual permutation as well.  This extra information does
not make the reduction invalid.  Given $K_i$, Bob can generate exactly
that extra randomness himself by sampling a uniform bijection
\[
    P_i\setminus\set{1}
    \longrightarrow
    C_i\setminus\set{\tau_i(K_i)}.
\]
All edges outside the clean coordinates are determined by the fixed
triples and ownership functions.  At clean coordinates Bob owns only
permutation edges, and Alice owns only middle XOR edges.  Thus Bob can
generate his entire share from the indices and his sampled bijections;
Alice can generate her entire share from the strings without knowing
those indices.  A protocol for
the fixed DHX branch therefore gives a protocol for the unequal-block Hidden XOR
instance in \eqref{eq:dhx-residual-bhx}, with no increase in message
length.

\begin{theorem}[Quantum Distributional Hidden XOR]
\label{thm:quantum-dhx-main}
For sufficiently large $w$ and
$t\ge512\ln w$,
every one-way quantum protocol for $\DHX_{w,t}$ with success
probability at least
$\frac{1}{2}+\frac{1}{6w}$
communicates $\bOm{wt}$ qubits.
\end{theorem}

\begin{proof}
If the DHX protocol uses independent public coins, first fix a seed
with at least its average success under the DHX input and partition law.
Then condition and fix the classical variables as above.  The residual problem
is unequal-block Hidden XOR with
\[
    r\ge\frac{t}{512},
    \qquad
    w_i\ge\frac{w}{100},
\]
and success at least $1/2+1/(8w)$.  Since
$t\ge512\ln w$, we have $r\ge\ln w$.  By
\cref{thm:qbhx-main},
$q=\bOm{wr}=\bOm{wt}$.
The message register remains the original $q$-qubit register throughout
the conditioning, by \cref{lem:classical-conditioning-fixing}.
\end{proof}

\begin{remark}[Where the quantum proof differs from the classical proof]
\label{rem:qdhx-classical-comparison}
This is the precise point where the quantum proof departs from the
classical Assadi--Sundaresan proof.  The clean-coordinate reduction and all
classical conditioning are the same kind of operations as before.  Once
the residual product parity is exposed, however, the classical proof
would condition the input distribution on Alice's message and apply a
scalar Fourier/KKL argument to that conditional distribution.  Our proof instead invokes the matrix-Fourier
bound of \cref{thm:qbhx-main} on the full quantum encoding.
\end{remark}

\subsection{Completing the robust Noisy Gap Cycle lower bound}
\label{subsec:finish-robust-ngc}
\label{subsec:quantum-ngc-main}

For the main parameter regime, combining the decorrelation reduction and
DHX lower bound is now immediate.  Two small technical issues remain:
$w=n/(4k)+1$ can be bounded when $k=\bT{n}$, while the Multi-Block
construction directly gives only cycle lengths $k\equiv1\pmod{3}$.

\begin{lemma}[Bounded-width DHX]
\label{lem:bounded-width-dhx}
\label{lem:app-bounded-width-dhx}
Fix a constant $W\ge2$.  For every $2\le w\le W$ and all sufficiently
large $t$, every one-way quantum protocol for $\DHX_{w,t}$ with success at
least $1/2+1/(6w)$ communicates $\bOm[W]{t}=\bOm[W]{wt}$ qubits.
\end{lemma}

\begin{proof}
First fix any independent public protocol seed by averaging, preserving
success under the DHX distribution.  Call a block fully clean if all $w$ of its coordinates are clean.  This
has probability $2^{-6w}\ge2^{-6W}$.  Hence with probability
$1-\exp(-\bOm[W]{t})$ there are $\bOm[W]{t}$ fully clean blocks.
Condition on this event.  Fix, by averaging, the exact set $F$ of fully
clean blocks, all ownership data, and all strings and permutations in
blocks outside $F$.  Retain precisely the blocks in $F$, rather than
all merely usable blocks.  Within each retained block the bit string and
permutation remain independent and uniform, by the same ownership and
local-simulation argument as in \cref{lem:residual-usable}.  The fixed
blocks contribute only a known XOR offset.  Because $w$ is bounded, the original
advantage $1/(6w)$ is a positive constant depending only on $W$, so the
high-probability conditioning leaves advantage at least $1/(8w)$ for
large $t$.

The residual problem is equal-width Hidden XOR on
$r=\bOm[W]{t}$ blocks of width $w$.  Apply
\cref{thm:qbhx-master} with $\lambda=1$.  Since
$w\left(\frac{1}{4w}\right)^{2/r}\longrightarrow w\ge2$,
the resulting bound is
$q\ge r\log_2 w-2\log_2(4w)=\bOm[W]{r}=\bOm[W]{t}$,
and boundedness of
$w$ makes this also $\bOm[W]{wt}$.
\end{proof}

We can now prove \cref{thm:quantum-robust-ngc-main}.

\begin{proof}[Proof of \cref{thm:quantum-robust-ngc-main}]
First suppose $k\equiv1\pmod{3}$, and write
\[
    k=3t+1,
    \qquad
    m=\frac{n}{4k},
    \qquad
    w=m+1.
\]
By \cref{lem:neighboring-hybrid-main,lem:ngc-dhx-main}, a $q$-qubit
protocol solving the hard NGC distribution with success $2/3$ induces a
$q$-qubit protocol for $\DHX_{w,t}$ with success at least
$1/2+1/(6w)$.  If $w$ is sufficiently large, apply
\cref{thm:quantum-dhx-main}; if $w$ lies in the bounded exceptional range,
apply \cref{lem:bounded-width-dhx}.  In either case the communication is
$\bOm{wt}$.

Now remove the congruence assumption.  If $k\equiv0\pmod{3}$, set
\[
    k_0=k-2,
    \qquad
    n_0=\frac{k-2}{k}n,
\]
and if $k\equiv2\pmod{3}$, set
\[
    k_0=k-1,
    \qquad
    n_0=\frac{k-1}{k}n.
\]
Then $k_0\equiv1\pmod{3}$ and in both cases
$\frac{n_0}{4k_0}=\frac{n}{4k}$.
Starting from the hard instance with cycle length $k_0$, each party
first deletes its own public closing edges from layer $k_0$ to layer one.
Append $k-k_0$ layers joined by identity perfect matchings on all groups
and both vertex types.  Replace the old closing edges by the same public
closing pattern from the new final layer to layer one.  Thus paths gain
$k-k_0$ edges and the closed components have the required cycle lengths.
All removed and added endpoints are public; assign the new edges fresh
independent fair owners.  Retained and new ownership coins have exactly
the robust partition law conditional on the padded graph.  This padding
costs no communication and does not manipulate Alice's quantum message.

Write $k_0=3t+1$.  For sufficiently large $n$ and
$k\ge1600\ln n$,
\[
    t
    \ge
    \frac{k-3}{3}
    \ge
    512\ln n
    \ge
    512\ln w.
\]
Therefore the preceding DHX lower bound applies.  Finally,
\[
\begin{aligned}
    wt
    &\ge
    \frac{n}{4k}\cdot\frac{k-3}{3}\\
    &=
    \frac{n}{12}\left(1-\frac{3}{k}\right)
    =
    \bOm{n}.
\end{aligned}
\]
Thus the fixed hard distribution requires $\bOm{n}$ qubits.

To pass to robust public-coin communication, suppose a public-coin
protocol of smaller cost succeeded with probability at least $2/3$ on
every fixed valid graph.  Average its success first over the hard graph
distribution and the random robust partition and then over the public
random string.  Some fixed public random string has average success at
least $2/3$ on the hard distribution.  Fixing that string yields a
fixed-coin quantum protocol contradicting the distributional
lower bound above.  Hence the robust communication complexity is
$\bOm{n}$.
\end{proof}

\subsection{From robust communication to standard random order}
\label{subsec:robust-to-streaming-main}
\label{app:qro-transfer}

The transfer uses the following exact permutation identity.

\begin{lemma}[Random partition followed by local random orders is uniform]
\label{lem:random-partition-order}
\label{lem:app-random-partition-order}
Let $E$ be a set of $M$ labeled stream items.  Assign each item
independently to $E_A$ or $E_B$ with probability $1/2$.  Conditional on
the partition, choose independent uniform permutations $\pi_A$ of $E_A$
and $\pi_B$ of $E_B$.  Then
$\pi_A(E_A)\circ\pi_B(E_B)$
is a uniformly random permutation of $E$.
\end{lemma}

\begin{proof}
Fix a target permutation
$\pi=(e_{\pi(1)},\ldots,e_{\pi(M)})$.  For the concatenation to equal
$\pi$, there must be some cut position $a\in\set{0,\ldots,M}$ such that
Alice receives exactly the first $a$ items of $\pi$ and Bob receives the
rest.  For fixed $a$, the probability of this ownership pattern and the
required two local orders is
$2^{-M}\frac{1}{a!}\frac{1}{(M-a)!}$.
Summing over $a$,
\[
 \prob{\pi_A(E_A)\circ\pi_B(E_B)=\pi}
 =\frac{2^{-M}}{M!}\sum_{a=0}^M{M\choose a}
 =\frac1{M!}.
\]
Thus every global permutation has the same probability.
\end{proof}

\begin{proposition}[Random-order streaming gives robust communication]
\label{prop:streaming-to-robust-main}
\label{lem:app-streaming-to-robust}
Suppose a one-pass quantum streaming algorithm solves a graph problem
$\mathcal {P}$ with error at most $\eps$ in uniformly random edge
order using total live space at most $S$, with all output produced after
the pass.  Assume continuation is determined by the complete live state,
public parameters, and Bob's remaining input: in particular, a necessary
position/control value is either counted or locally reconstructible.
Then there is a
robust one-way quantum protocol for $\mathcal {P}$ with error at most
$\eps$ and message dimension at most $2^S$.
\end{proposition}

\begin{proof}
Fix the underlying graph and sample its robust random edge partition
$E=E_A\sqcup E_B$.  Alice orders $E_A$ uniformly at random and runs the
streaming algorithm on that prefix.  She then sends the algorithm's
\emph{complete live memory} to Bob.  By the space convention of this
paper, this hybrid classical--quantum state has dimension at most $2^S$.
Bob independently orders $E_B$ uniformly at random, continues the same
streaming computation from the received state, and produces the final
output.  No emitted prefix transcript is needed under the final-output
hypothesis.  The local order coins are independent of the simulated
algorithm's coins, and no ownership or cut information is supplied to
that algorithm beyond its original interface.

By \cref{lem:random-partition-order}, the concatenated sequence seen by
the simulated streaming algorithm is exactly a uniform random permutation
of the full edge set.  Hence the communication protocol inherits the
streaming algorithm's error guarantee.  The live state is transferred
only once, so the reduction never needs to clone, rewind, or branch a
quantum state.
\end{proof}

\begin{proof}[Proof of \cref{thm:quantum-random-order-ngc-main}]
An $S$-space random-order streaming algorithm for $\NGC_{n,k}$ would give,
by \cref{prop:streaming-to-robust-main}, a robust one-way quantum protocol
with message dimension at most $2^S$.  The known edge count is
$M=n-n/(2k)$ in both promise cases, so Bob can recover the cut position as
$M-\abs{E_B}$ if the transition rule needs it.  If an algorithm emits its
single decision earlier, retaining that bit and an emitted flag until the
end costs only $\bO{1}$ extra space, leaving the lower bound unchanged.  By
\cref{thm:quantum-robust-ngc-main}, this requires $S=\bOm{n}$.
\end{proof}

\subsection{Consequences of the Assadi--Sundaresan reductions}
\label{subsec:qro-consequences-main}
\label{app:qngc-applications}

The following consequences use classical reductions from the
Assadi--Sundaresan framework.  We check their input/order laws and use
only one execution of the streaming algorithm, with classical output
available after the pass.

\begin{corollary}[Quantum random-order consequences of NGC]
\label{cor:qro-consequences-main}
The following lower bounds hold for one-pass quantum algorithms in
uniformly random order, with end-of-stream output and success at least
$2/3$ for every input (the sampling guarantee is specified separately).
The space impossibilities are uniform in the displayed parameters, in
particular along the linked accuracy/size sequences used in the proof.
\begin{enumerate}[leftmargin=1.6em]
    \item \textbf{Connected components.}
    There is no algorithm that, for every $\epsilon>0$, estimates the
    number of connected components to additive error $\epsilon n$ with
    success at least $2/3$ using
    $2^{\lO{1/\epsilon}}n^{\lO{1}}$
    space.

    \item \textbf{Maximum matching.}
    The same statement holds for estimating maximum matching size to
    additive error $\epsilon n$, even on planar graphs of maximum degree
    two.

    \item \textbf{Planar maximum independent set.}
    The same statement holds for estimating maximum independent set size
    to additive error $\epsilon n$ on planar graphs of maximum degree two.

    \item \textbf{Minimum spanning tree.}
    For connected graphs with integer edge weights in $[W]$, $W\ge2$, there is no
    algorithm that, for every $\epsilon>0$, obtains a $(1+\epsilon)$
    approximation to the MST weight with success at least $2/3$ using
    $2^{\lO{W/\epsilon}}n^{\lO{1}}$
    space.

    \item \textbf{Random-walk generation.}
    The algorithm outputs classical walks after its single pass.
    Its joint output law must be within total variation $1/10$ of
    independent samples, each starting at a uniform vertex and, conditional
    on its start, pointwise within a factor $1\pm1/10$ of the true
    length-$\ell$ walk law.  For sufficiently large absolute constants
    $c_1,c_2$, generating $c_1$ such walks requires
    $2^{\bOm{\sqrt{\ell}}}$ space when $\ell=\bT{\log^2 n}$, whereas
    generating $c_2 4^\ell$ such walks requires $2^{\bOm{\ell}}$ space
    when $\ell=\bT{\log n}$.
\end{enumerate}
\end{corollary}

\begin{proof}
We check why each classical reduction remains valid with quantum memory.

For connected components no transformation is needed.  The two NGC cases
differ in their number of cycle components by $n/(4k)$, while the noisy
paths contribute equally.  Thus additive error smaller than $n/(8k)$
distinguishes the cases.  Taking $k=\bT{\log n}$ gives
$1/\epsilon=\bT{\log n}$; a hypothetical
$2^{\lO{1/\epsilon}}n^{\lO{1}}$-space algorithm would use $\lO{n}$ space and
contradict \cref{thm:quantum-random-order-ngc-main}.

For maximum matching, take $k$ odd.  Every $2k$-cycle has matching size
$k$, every $k$-cycle has matching size $(k-1)/2$, and every noisy
$(k-1)$-path has matching size $(k-1)/2$.  The two NGC cases therefore
differ by exactly $n/(4k)$ in maximum matching size.  The hard graphs are
disjoint unions of cycles and paths, hence planar with maximum degree two.
The same parameter choice gives the claimed lower bound.

For maximum independent set, again take $k$ odd.  A $2k$-cycle has
independence number $k$, a $k$-cycle has independence number $(k-1)/2$,
and a $(k-1)$-path on $k$ vertices has independence number $(k+1)/2$.
The noisy paths cancel between the two cases, and the cycle contribution
again differs by $n/(4k)$.  This yields the same reduction and the same
space lower bound.

For MST, use the structured hard distribution and its robust lower bound
from the proof of \cref{thm:quantum-robust-ngc-main}.  Work on an admissible
family with $k\equiv1\pmod3$, $k=\bT{\ln n}\ge1600\ln n$, and
$m=n/(4k)\ge2$.  Give original edges weight one.  In the public layered
labels, add weight-$W$ edges $\set{a^k_j,b^k_j}$ for $j\in[m]$;
add weight-one edges from $a^1_j$ to both $a^1_{m+j}$ and $b^1_{m+j}$;
and add weight-one edges $\set{a^1_j,b^1_{j+1}}$, with $j+1$ read cyclically
in $[m]$.  This is the augmentation of
Assadi--Sundaresan~\cite{AssadiSundaresan2023}.
For crossed parity the weight-one graph is connected, so its MST weighs
$n-1$.  For parallel parity the weight-one graph has $m$ components,
connected cyclically by the heavy edges, and its MST weighs
$n-m+W(m-1)$.

Fix $W=2$ and choose accuracy $1/(64k)\le1/4$.  The relative gap is
at least $1/(8k)$, eight times this accuracy, so the two symmetric
relative-error intervals are disjoint.  Moreover,
$W/(1/(64k))=128k=\bT{\ln n}$, and the proposed space bound is
$n^{\lO{1}}=\lO{n}$.  Fixed $W=2$ already suffices to contradict a
uniform bound of the stated form.  The restriction $W\ge2$ is necessary:
for $W=1$ the MST weight is always $n-1$.
For the simulation, assign every public added edge a fresh independent
robust ownership coin, and order each party's whole augmented share
uniformly.  By \cref{lem:random-partition-order}, the resulting stream
has exactly uniform edge order.  Transfer the complete quantum memory
once.  This contradicts the hard-distribution robust lower bound without
requiring an online augmentation of an already ordered stream.

For the random-walk bounds, take an admissible family
$k=\bT{\ln n}\ge1600\ln n$.  A uniform starting vertex is on a cycle with
probability $1/2$.  On a cycle of length at most $2k$, lift the simple walk
to the integer line.  First hitting displacement $2k$ or $-2k$ has expected
time $4k^2$ and traverses every edge of the cycle.  Since the expected hitting time is $4k^2$, Markov's inequality shows
that a length-$40k^2$ walk traverses the cycle with probability at least
$0.9$, conditional on starting on a cycle.  Pointwise approximation
leaves per-walk detection probability at least $0.405$.  With, for example,
eight independent walks, detection probability exceeds $0.98$, and joint
TV error $1/10$ leaves success above $2/3$.  The traversed edge set reveals
a cycle and its length; a noisy path cannot produce a false cycle.
Thus $\ell=40k^2=\bT{\log^2 n}$ gives
$\bOm{n}=2^{\bOm{\sqrt\ell}}$ space.

For the many-walk bound take $\ell=2k$.  Conditional on a cyclic start,
following one cyclic direction for all $\ell$ steps has probability at
least $2^{-\ell}$ and traverses the entire cycle.  After pointwise
approximation the unconditional detection probability is at least
$0.45\,2^{-\ell}$.  For $c_2 4^\ell$ independent samples the failure
probability is at most $\exp(-0.45c_2 2^\ell)$; joint TV error $1/10$
still permits success at least $2/3$ for large parameters.  Here
$\ell=\bT{\log n}$, so the NGC lower bound is
$\bOm{n}=2^{\bOm{\ell}}$.  Bob runs the walk generator once and
postprocesses its classical output.  Neither argument copies or branches
its quantum memory.  The first regime differs from the common
$\ell=\bT{\log n}$ wording in the cited Corollary 5.10: its constant-budget
proof, as here, uses $\ell=\bT{k^2}$.
\end{proof}

\begin{remark}[Why we do not automatically include matrix rank]
\label{rem:rank-caveat-main}
Assadi and Sundaresan also derive a matrix-rank consequence from maximum
matching using the Tutte matrix.  We do not claim that consequence here.
An undirected edge $\set{u,v}$ creates the pair of nonzero entries
\[
    T_{uv}=x_{uv},
    \qquad
    T_{vu}=-x_{uv}.
\]
A uniform random order of graph edges therefore induces a uniform random
order of \emph{pairs} of matrix entries, not obviously a uniform random
order of the individual nonzero entries.  This is an order-preservation
issue, independent of quantum information, and needs a separate argument
before the random-entry-order rank lower bound can be claimed.
\end{remark}

\end{document}